\documentclass{article}
\usepackage[lang = british]{ems-msl} %% change to `american' if you use American English

\usepackage{subfigure}
\usepackage{caption}
\usepackage{comment}
\usepackage{verbatim}
\usepackage{booktabs}
\usepackage{diagbox}
\usepackage{parskip}
\usepackage{url,amsfonts,color, multicol}
\usepackage{bm}
\usepackage{mathtools}
\usepackage{algorithm}
\usepackage{algorithmic}
\usepackage{multirow}

\usepackage{fancyhdr}
\usepackage{tikz}

\definecolor{mygray}{gray}{0.6}

\usepackage{fixfoot}
\DeclareFixedFootnote{\rep}{This refers to the point with minimum value according to the function $f$.}

\usetikzlibrary{arrows,automata}

\newtheorem{theorem}{Theorem}
\newtheorem{observation}{Observation}
\newtheorem{lemma}[theorem]{Lemma}

\newtheorem{corollary}[theorem]{Corollary}
\newtheorem{definition}[theorem]{Definition}
\newtheorem{remark}{Remark}

\newtheorem{assumption}{Assumption}
\usepackage{xcolor}
\usepackage{tablefootnote}

\renewcommand{\Vec}[1]{\mathbf{#1}}
\newcommand\flsdescent{\texttt{FLSD}}
\newcommand\dacs{\texttt{DACS}}
\usepackage{hyperref}
\hypersetup{colorlinks=true,linkcolor=blue,citecolor=blue}

\newcommand{\rank}{\mathrm{rank}}
\newcommand{\tablecell}[2]{\begin{tabular}{@{}#1@{}}#2\end{tabular}}

\definecolor{arsenic}{rgb}{0.23, 0.27, 0.29}
\definecolor{cadet}{rgb}{0.33, 0.41, 0.47}
\definecolor{airforceblue}{rgb}{0.36, 0.54, 0.66}
\definecolor{midnightblue}{rgb}{0.1, 0.1, 0.44}

\begin{document}

\title{The Query Complexity of Local Search and Brouwer in Rounds}

\titlemark{Local Search and Brouwer in Rounds}
\emsauthor{1}{Simina Br\^anzei}{S.~Br\^anzei}
\emsauthor{2}{Jiawei Li}{J.~Li}

\emsaffil{1}{Postal address: Purdue University, Department of Computer Science \email{simina.branzei@gmail.com}}

\emsaffil{2}{Postal address: University of Texas  at Austin, Department of Computer Science \email{davidlee@cs.utexas.edu}}

% TODO: correct the classification if it becomes relevant 
\classification[60Gxx]{91Axx}

\keywords{local minima, local search, grid, rounds of adaptivity, query complexity, Brouwer fixed point, parallel computation}

\begin{abstract}%
We consider the query complexity of finding a local minimum of a function defined on a graph. This abstract problem is fundamental to many optimization tasks, such as finding a local minimum of the loss function when training deep neural networks. In such applications, each query is an expensive loss evaluation, making it crucial to parallelize computations. This motivates our study of local search where at most $k$ rounds of interaction (aka adaptivity) with the oracle are allowed.

%We consider the query complexity of finding a local minimum of a function defined on a graph, where at most $k$ rounds of interaction (aka adaptivity) with the oracle are allowed.  Adaptivity is a fundamental concept  studied due to the need to parallelize computation and understand the speedups attainable.  
We focus on the $d$-dimensional grid $\{1, 2, \ldots, n \}^d$, where the dimension $d \geq 2$ is a constant. Our main contribution is to give algorithms and lower bounds that characterize the query complexity of finding a local minimum  in $k$ rounds, when $k$ is constant and polynomial in $n$, respectively.
 Our proof technique for lower bounding the query complexity in rounds may be of independent interest as an alternative to the classical relational adversary method of Aaronson from the fully adaptive setting.
The local search analysis also enables us to characterize the query complexity of  computing a Brouwer fixed point in rounds. 
\end{abstract}

\maketitle

\section{Introduction}

Local search is a powerful heuristic embedded in many natural processes, which is often used to solve hard optimization problems.
Algorithms based on local search include gradient methods, the Lin-Kernighan algorithm for the traveling salesman problem, and the Metropolis-Hastings algorithm for sampling. The complexity of local search has been extensively studied in both  the white box  model (see, e.g., \cite{JohnsonPY88}) and the black box (aka query) model (see, e.g., \cite{aldous1983minimization}).

In the black box model, there is a graph $G = (V,E)$ and an unknown function $f : V \to \mathbb{R}$. An algorithm must query a vertex $v$ to learn $f(v)$. The goal is to find a vertex $v$ that is a  local minimum: $f(v) \leq f(u)$ for all $(u,v) \in E$. The  query complexity is the number of oracle queries needed to find a local minimum in the worst case. 

The query complexity of local search was first considered by Aldous~\cite{aldous1983minimization}, who showed that steepest descent with a warm start is a good randomized algorithm: first query $t$ vertices $x_1, \ldots, x_t$ selected uniformly at random and pick the vertex $x^*$ that minimizes the function among these \footnote{That is, the vertex $x^*$ is defined as: $x^* = x_j$, where $j = \arg \min_{i=1}^t f(x_i)$.}. Then run steepest descent from $x^*$ and stop when no further improvement can be made, returning the final vertex reached. When $t = \sqrt{ \Delta N}$, where $N$ is the number of vertices and $\Delta$ the maximum degree of the graph $G$, the algorithm issues $O(\sqrt{\Delta N})$ queries in expectation and has roughly as many rounds of interaction with the oracle. 

\paragraph{Local search in rounds.} Aldous' algorithm described above is highly sequential, i.e. requires many rounds of interaction with the oracle, even though its total query complexity is essentially optimal for graphs such as the Boolean hypercube and the $d$-dimensional grid~\cite{aldous1983minimization, sun2009quantum, zhang2009tight}. 
Multiple rounds of interaction can be expensive in applications. For instance, when algorithms such as gradient descent are run on data stored in the cloud, there can be delays due to back and forth messaging across the network. A remedy for such delays is designing protocols with fewer  rounds of communication.

In this paper, we analyze the query complexity of local search  in rounds, focusing on the case where the graph $G$ is the $d$-dimensional grid of side length $n$.  When an algorithm has $k$ rounds of interaction with the oracle, it asks a batch of queries  in each round $i$,  receives the answers, and then issues the batch of queries for   round $i+1$ (so, within a round, queries are not adaptive). The algorithm stops and outputs an answer by the end of round $k$. 

\paragraph{Examples.} To illustrate our model, we describe an analogy to the classical optimization problem of linear regression with mean squared error (MSE) loss. Given a data set of $m$ labeled examples $\{(\vec{a}_i,b_i)\}_{i=1}^{m}\subseteq\mathbb{R}^{d}\times\mathbb{R}$, linear regression aims to find a vector $\vec{x}\in\mathbb{R}^{d}$ of regression coefficients that minimize the loss function:
\[ 
L(\vec{x})=\frac{1}{m}\sum_{i=1}^{m}(\langle\vec{a}_i,\vec{x}\rangle - b_i)^2.
\]
While the ideal coefficients exist in a continuous space ($\mathbb{R}^d$), any practical optimization algorithm operates by taking discrete steps within a bounded region, iteratively updating an estimate $\vec{x}_{\text{old}}$ to a new one, $\vec{x}_{\text{new}}$.
%any practical search or numerical optimization is performed in {discrete steps} and within a {bounded region}. For instance, an algorithm iteratively refines an estimate $\vec{x}_{\text{old}}$ to a new one, $\vec{x}_{\text{new}}$. 
To formally study the complexity of such a search, we can abstract this process by modeling the search space as a discrete grid. We consider a bounded region of potential coefficients, such as $[-B,B]^d$ for some $B > 0$, and discretize it into the grid $[n]^d$.
%This reinterpretation creates a direct mapping to our local search framework:
%\begin{itemize}
Each {vertex} $\vec{x}\in[n]^d$ corresponds to a candidate vector of regression coefficients.
    {Querying} a vertex $\vec{x}$ returns its MSE loss, $f(\vec{x}) = L(\vec{x})$.
    The {neighborhood} of a vertex $\vec{x}$ consists of all adjacent grid points, representing the smallest possible adjustments to a single regression coefficient.

A local minimum in this model is a set of coefficients whose loss cannot be improved by changing one coefficient by a discrete step.
Performing local search in \( k \) rounds corresponds to evaluating multiple candidate vectors simultaneously in $k$ batches. In each round, an algorithm queries a batch of points simultaneously, uses the results to select the next batch, and repeats this process $k$ times. 

While linear regression with MSE loss is convex, many practical variations are not. For instance, adding a non-convex regularizer like the Minimax Concave Penalty (MCP)~\cite{zhang2010mcp} creates a loss function with multiple local minima:
\[
L_{\text{MCP}}(\vec{x}) = \frac{1}{m}\sum_{i=1}^{m}(\langle\vec{a}_i,\vec{x}\rangle - b_i)^2+\sum_{j=1}^{d}\text{MCP}_{\lambda,\gamma}(x_j),
\]
where 
\(MCP_{\lambda ,\gamma }(x_{j})\) represents the MCP penalty function for the \(j\)-th regression coefficient \(x_{j}\); \(\lambda \) is a regularization parameter that controls the overall strength of the penalty; and \(\gamma \) is a tuning parameter that controls the concavity of the MCP function. % and how quickly the penalty tapers off for larger coefficients.

Since finding the global minimum of  non-convex functions is computationally infeasible in general, the practical goal shifts to efficiently finding a local minimum.  %highlighting the importance of local search algorithms.

%\paragraph{Corollaries for Brouwer.} Our results for local search imply bounds for the problem of finding an approximate Brouwer fixed point on the d-dimensional cube. The connection between local search and Brouwer  stems from game theory. Finding a Nash equilibrium in a general game is equivalent in complexity to solving the Brouwer problem. However, for specific classes of games, such as potential games, pure Nash equilibria are guaranteed to exist and can be found via a sequence of best-response moves—a process equivalent to local search. 

\paragraph{Roadmap to the paper.}  
%Section~\ref{sec:model_results} defines the model. Section~\ref{sec:our_results} presents our main results. Section~\ref{sec:related_work} discusses related work.  Section~\ref{sec:local_search_main} overviews the results and proof techniques for local search, with references to the appendix for the full proofs. Section~\ref{sec:brouwer_main} overviews the results and proofs for the Brouwer problem, which has a connection to local search via game theory. %\textbf{\textcolor{red}{TODO. Add a bit more here, also for model/related work to make a smooth transition to the rest}}

The remainder of the paper is organized as follows. We formalize the model in Section~\ref{sec:model_results} and present our main theorems in Section~\ref{sec:our_results}.  Related work is discussed in Section~\ref{sec:related_work}. The core technical sections follow, with Section~\ref{sec:local_search_main} providing proof overviews for our local search results, while Section~\ref{sec:brouwer_main} details the application to the Brouwer fixed-point problem, which has a connection to local search via game theory. Full technical proofs are deferred to the appendix.

%============================================================================
\section{Model} \label{sec:model_results}

We introduce the model for local search and Brouwer, then define the query complexity.

\paragraph{Local Search} Let $G = (V,E)$ be an undirected graph and $f : V \rightarrow \mathbb{R}$ a function that assigns a value to every vertex. The function $f$ is unknown but can be queried; a query is a vertex $\Vec{x} \in V$ and the answer to the query  is the value $f(\Vec{x})$ of the function at $\Vec{x}$. A local minimum is a vertex $\Vec{x}$ with the property that $f(\Vec{x}) \leq f(\Vec{y})$ for all neighbours $\Vec{y}$ of $\Vec{x}$. The goal is to find a local minimum using as few queries as possible.

We focus on the setting where the graph is a $d$-dimensional grid of side length $n$. Thus $V = [n]^d$, where $[n] = \{1, 2, \ldots, n \}$ and $(\Vec{x},\Vec{y}) \in E$ if $\|\Vec{x}-\Vec{y}\|_{1}=1$.  The dimension $d $ for both local search and Brouwer fixed-point is a constant. Unless otherwise specified, we have $d \geq 2$.

The local search results imply bounds for the problem of finding a Brouwer fixed point, which is defined next.

%\paragraph{Query complexity and rounds.}   

\paragraph{Brouwer}

In the Brouwer setting, we are given an $L$-Lipschitz function \[ f: [0,1]^d \rightarrow [0,1]^d,\] where $L > 1$  is a constant\footnote{When $L < 1$ we obtain the Banach fixed-point theorem, where the unique fixed point can be approximated quickly.} such that
$\|f(\Vec{x})-f(\Vec{y})\|_{\infty} \leq L\|\Vec{x}-\Vec{y}\|_{\infty}, \forall \Vec{x},\Vec{y} \in [0,1]^d\,.$ %Assume $L$ is a constant larger than $1$.

The computational problem is: given a tuple $(\epsilon, L, d, f)$, find a point $\Vec{x}^* \in [0,1]^d$ such that $\|f(\Vec{x}^*)-\Vec{x}^*\|_{\infty} \leq \epsilon$. An exact fixed point exists by Brouwer's fixed point theorem.

\paragraph{Query complexity and rounds}
We have query (aka oracle) access to the function $f$ and   at most $k$ rounds of interaction with the oracle. An algorithm running in $k$ rounds will submit in each round $j$ a batch of queries, then wait for the answers, and then submit the batch of queries for round $j+1$. The choice of queries submitted in round $j$ can only depend on the results of queries from earlier rounds.
At the end of the $k$-th round, the algorithm must stop and output a solution.

The deterministic query complexity is the total number of queries necessary and sufficient to find a solution. The randomized query complexity is the expected number of queries required to find a solution with probability at least $2/3$ for any input, where the expectation is taken over the coin tosses of the protocol.
%\jiawei{I found this definition a little bit tricky. In the \textbf{constant} rounds setting, is there any easy way to show that any positive accuracy are equivalent? In polynomial rounds setting, we don't have this problem because we could run several copies of this algorithm in parallel, and use one more round to verify the correctness of each solution.}
\footnote{Any other constant greater than $1/2$ will suffice.}

\section{Our Results} \label{sec:our_results}

\section*{Local Search}

To set the stage, the query complexity of local search on the $d$-dimensional grid $[n]^d$ is well understood in the fully adaptive setting. The classical result by Llewellyn, Tovey, and Trick~\cite{llewellyn1989local} showed that the query complexity for deterministic fully adaptive algorithms is $\Theta(n^{d-1})$, and this upper bound is achieved by a divide-and-conquer algorithm using only $O(\log n)$ rounds; in the randomized setting, Aldous' algorithm~\cite{aldous1983minimization} takes $O(n^{d/2})$ queries and rounds, 
and it was proven to be optimal later (\cite{zhang2009tight, sun2009quantum}). More related works are discussed in Section~\ref{sec:related_work}.

We show the following bounds for  local search on the $d$-dimensional grid $[n]^d$ in $k$ rounds, which quantify the trade-offs between the number of rounds of adaptivity and the total number of queries.

\begin{theorem}(Local search, constant rounds)\label{thm:1}
	Let $d, k\in \mathbb{N}$ be constant, where $d \geq 2$ and $k \geq 1$. The query complexity of local search in $k$ rounds on the $d$-dimensional grid $[n]^d$ is $\Theta\bigl(n^{\frac{d^{k+1} - d^k}{d^k - 1}}\bigl)$, for both deterministic and randomized algorithms.
\end{theorem}

For example, when  $k=2$, the query complexity is $\Theta\bigl(n^{\frac{d^{2}}{d+1}}\bigl)$ for both deterministic and randomized algorithms.
When $k \rightarrow \infty$, the bound of Theorem~\ref{thm:1} is close to $\Theta(n^{d-1})$, with gap smaller than any polynomial.

Recall that the $O(\log n)$ rounds divide-and-conquer algorithm take $O(n^{d-1})$ queries (\cite{llewellyn1989local}). 
Thus, our result fills the gap between one round algorithms and logarithmic rounds algorithms except for a small margin. 
For constant number of rounds, Theorem~\ref{thm:1} shows that deterministic algorithms are optimal, so there is no difference between the deterministic and randomized query complexity.
%This theorem also implies that randomness does not help when the number of rounds is constant.

\begin{theorem}(Local search, polynomial rounds)\label{thm:2}
	Consider the $d$-dimensional grid with side length $n\in \mathbb{N}$, where $d \geq 2$.
Let   $k = n^{\alpha} \in \mathbb{N}$, where $\alpha \in \left(0,  {d}/{2} \right)$ is a constant. The randomized query complexity of local search in $k$ rounds on the $d$-dimensional grid $[n]^d$ is:
	\begin{itemize}
		\item $\;\;\;$ $\Theta\left(n^{(d-1) - \frac{d-2}{d}\alpha}\right)$ when $d \geq 5$;
		\item $\;\;\;$ $O\left(n^{3 - \frac{\alpha}{2}}\right)$ and $\widetilde{\Omega}\left(n^{3 - \frac{\alpha}{2}}\right)$ when $d = 4$;
		\item $\;\;\;$ $O\left(n^{2 - \frac{\alpha}{3}}\right)$ and $\Omega\left(\max(n^{2-\frac{2\alpha}{3}}, n^{\frac{3}{2}})\right)$ when $d = 3$. 
	\end{itemize}
	%at most $O\left(n^{(d-1) - \frac{d-2}{d}\alpha}\right)$ and at least $\Omega\left(n^{(d-1) - \frac{d-2}{d}\alpha}\right)$ if $d \geq 5$; $\widetilde{\Omega}\left(n^{(d-1) - \frac{d-2}{d}\alpha}\right)$ if $d = 4$; $\Omega\left(n^{(d-1) - \frac{d-1}{d}\alpha}\right)$ if $d = 3$.
\end{theorem}
When $\alpha \rightarrow 0$, the bound approaches $\Theta(n^{d-1})$, i.e., the bound of constant and logarithmic rounds algorithm. When $\alpha \rightarrow ({d}/{2})$, the upper bound is close to $\Theta(n^{\frac{d}{2}})$, i.e., the randomized fully adaptive algorithm (\cite{aldous1983minimization,zhang2009tight}). Thus, our result fills the gaps between constant (or logarithmic) rounds algorithms and fully adaptive algorithms, except for a small gap when $d \in \{3, 4\}$.

%\begin{itemize}
	%\item When the number of rounds $k$ is constant, the query complexity of local search in $k$ rounds is $\Theta\bigl(n^{\frac{d^{k+1} - d^k}{d^k - 1}}\bigl)$, for both deterministic and randomized algorithms, where $d \geq 2$ (Theorem~\ref{thm:1}). %E.g., the query complexity on   $[n]^2$ in two rounds   is $\Theta(n^{4/3})$.
%	\item  When the number of rounds is polynomial, i.e. $k = n^{\alpha}$ for $0 < \alpha < d/2$,  the randomized query complexity is $\Theta\left(n^{(d-1) - \frac{d-2}{d}\alpha}\right)$ for all $d \geq 5$. For  $d=3$ and $d=4$, we show the same upper bound expression holds and give almost matching lower bounds (Theorem~\ref{thm:2}); the bound for $d=2$ with polynomial rounds was known.
	%\item     We also consider the case $d=1$ and show the query complexity on the 1D grid $[n]$ is $\Theta(n^{1/k})$, for both deterministic and randomized algorithms (Theorem~\ref{thm:1D}).
%\end{itemize}
The bound for $d=2$ in polynomial number of rounds was known, since the divide-and-conquer approach by Llewellyn, Tovey, and Trick~\cite{llewellyn1989local} gives an upper bound of $O(n)$ using only $O(\log{n})$ rounds, while Sun and Yao~\cite{sun2009quantum} (Theorem 1) give a lower bound of $\Omega\left(n^{1-\delta}\right)$ for all fixed $\delta > 0$ for fully adaptive randomized algorithms.

When $d=1$, the query complexity of computing a local minimum on the $1$-dimensional grid $[n]$ in $k$ rounds is $\Theta\left(n^{1/k}\right)$, for both deterministic and randomized algorithms (Theorem~\ref{thm:1D}). 

\medskip 

A summary of our results for local search on the $d$-dimensional grid can be found in Table~\ref{tab:LS}, together with the bounds known in the existing literature.
%\jiawei{I noticed that in the table 1 below, the citation for deterministic result doesn't have parenthesis, while the citation for randomized result have parenthesis.}

 \medskip 

% \jiawei{Can we further polish this table? I want to show the number of rounds needed in the fully adaptive setting in the table, rather than in the caption.}

% \setcitestyle{super}

\begin{table}[h!]
	\centering
	\begin{tabular}{| l | |l |l|}
		\hline
		\tablecell{c}{Local search on\\ the  grid $[n]^d$} & Deterministic & Randomized \\ [0.5ex] 
		\hline\hline 
		\tablecell{c}{{Constant rounds:}\\{$k = O(1)$}} &  \tablecell{c}{$\Theta\bigl(n^{\frac{d^{k+1} - d^k}{d^k - 1}}\bigl)$ \textcolor{purple}{(*)}} & {$\Theta\bigl(n^{\frac{d^{k+1} - d^k}{d^k - 1}}\bigl)$ \textcolor{purple}{(*)}} \\
		[1em] \hline 
		\tablecell{c}{Polynomial rounds:\\$k = n^{\alpha},$ \\$\alpha \in (0, d/2)$} & \tablecell{c}{$\Theta\left(n^{d-1}\right)$ \\ \cite{llewellyn1989local,DBLP:journals/dam/LlewellynT93,AK93}}  %since they gave a lower bound of $\Omega\left(n^{d-1}\right)$ for fully adaptive deterministic algorithms and a divide-and-conquer algorithm with such efficiency that runs in $O(\log n)$ rounds.}
		 & 
		\tablecell{l}{ $d \geq 5$: $\Theta\left(n^{(d-1) - \frac{d-2}{d}\alpha}\right)$ \textcolor{purple}{(*)}\\
			[1em] \hline 
			$d=4$: $O\left(n^{3 - \frac{\alpha}{2}}\right)$ and $\widetilde{\Omega}\left(n^{3 - \frac{\alpha}{2}}\right)$  \textcolor{purple}{(*)} \\
			[1em] \hline  
			$d=3$: $O\left(n^{2 - \frac{\alpha}{3}}\right)$ and \\  $\Omega\left(\max(n^{2-\frac{2\alpha}{3}}, n^{\frac{3}{2}})\right)$  \textcolor{purple}{(*)} \\
			[1em] \hline  
			$d=2$: $\widetilde{\Theta}(n)$~\cite{sun2009quantum, llewellyn1989local}
			%
			% \item $O\left(n^{3 - \frac{\alpha}{2}}\right)$ and $\widetilde{\Omega}\left(n^{3 - \frac{\alpha}{2}}\right)$ if $d = 4$;
			%\item $O\left(n^{2 - \frac{\alpha}{3}}\right)$ and $\Omega\left(\max(n^{2-\frac{\alpha}{3}}, n^{\frac{3}{2}})\right)$ if $d = 3$. 
		}
		\\ 
		[1em] \hline 
		\tablecell{c}{Fully adaptive:\\$k = \infty$} & \tablecell{c}{$\Theta\left(n^{d-1}\right)$ \\ %\tablefootnote{The algorithm is given by \citet{llewellyn1989local}, while the lower bound is from \citet{DBLP:journals/dam/LlewellynT93} and \citet{AK93}.}
		\cite{llewellyn1989local,DBLP:journals/dam/LlewellynT93,AK93}
		} 
		& 
		\tablecell{l}{$d \geq 3$: $\widetilde{\Theta} \bigl(n^{\frac{d}{2}} \bigr)$~\cite{aldous1983minimization,zhang2009tight}\\
			\vspace{2mm} 
			$d=2$: $\widetilde{\Theta}(n)$~\cite{sun2009quantum, llewellyn1989local} }
		\\ 
		\hline
	\end{tabular}
	\caption{Query complexity of local search in $k$ rounds on $d$-dimensional grid of side length $n$, for $d \geq 2$. Our results are marked with \textcolor{purple}{(*)}. The deterministic divide-and-conquer algorithm~(\cite{llewellyn1989local}) takes $O(\log n)$ rounds, while the randomized warm-start algorithm~(\cite{aldous1983minimization}) needs $O\bigl(n^{\frac{d}{2}}\bigr)$ rounds. % For deterministic fully adaptive algorithms, the algorithm is given by \cite{llewellyn1989local}, while the lower bound is from \cite{DBLP:journals/dam/LlewellynT93} and \cite{AK93}.
	}
	\label{tab:LS}
\end{table}
%The bound of $\Theta(n^{d-1})$ for deterministic algorithms in polynomial rounds can be inferred from  \cite{llewellyn1989local}, \cite{DBLP:journals/dam/LlewellynT93} and \cite{AK93}.

%\setcitestyle{authoryear}
%\simina{Explains what happens for $d=2$, randomized poly rounds. Note this is explained later in the paper, perhaps just introduce a cell in the table citing the paper that showed 2D -- this explanation is included later, so keep it only in one place: We always assume $d \geq 3$ for polynomial rounds, because \cite{sun2009quantum} proved a lower bound of $\widetilde{\Omega}(n)$ for fully adaptive algorithm in 2D and the divide-and-conquer algorithm by \cite{llewellyn1989local} achieves this bound with only $O(\log(n))$ rounds.}

At a high level, when the number of rounds $k$ is constant, the optimal algorithm is closer to the deterministic divide-and-conquer algorithm (\cite{llewellyn1989local}). When the number of rounds is polynomial (i.e. $k = n^{\alpha}$, for $0 < \alpha < d/2$), the algorithm is closer to the randomized algorithm from the fully adaptive setting, which does steepest descent with a warm start (\cite{aldous1983minimization}).

%The trade-off between the number of rounds and the total number of queries can also be seen as a transition from deterministic to randomized algorithms, with rounds imposing a limit on how much randomness the algorithm can use. 
%\textcolor{red}{TODO: See comments from reviewer and address the phrasing to explain more.}
%\jiawei{In our rebuttal: ''Theorem 1 shows that with constantly many rounds, randomness is not very useful %(Note it may
%actually still help a bit, since the analysis does not zoom in constant factors). However, when
%there are polynomial rounds (Theorem 2), the advantages of randomness increase as the number of
%rounds grows. There could be a gradual increase in the advantage obtained from randomness as one
%increases the number of rounds from 1 to unbounded.''}

\section*{Brouwer} 
Our local search results above   also imply a characterization for the query complexity of finding an approximate Brouwer fixed point in constant number of rounds on the $d$-dimensional cube.  
The connection between local search and Brouwer  stems from game theory. Finding a Nash equilibrium in a general game is equivalent in complexity to solving the Brouwer problem. However, for specific classes of games, such as potential games, pure Nash equilibria are guaranteed to exist and can be found via a sequence of best-response moves—a process equivalent to local search \cite{Rosenthal1973,MS96,BSN19}.

For Brouwer we consider only constant rounds, since Brouwer can be solved optimally in $O(\log (1/\epsilon))$ rounds (\cite{chen2005algorithms}).

%The Brouwer problem captures the complexity of finding a Nash equilibrium in a general game. For certain classes of games, such as potential games, pure Nash equilibria are guaranteed to exist and can be reached through a sequence of best-response moves, i.e. a local search process.

% a logarithmic number of rounds. %; more rounds only influence the query complexity by a sub-polynomial term.

\begin{theorem} \label{thm:brouwer_main}[Brouwer, constant rounds]
	Let $d,k \in \mathbb{N}$ be constant. For any $\epsilon > 0$, the query complexity of the $\epsilon$-approximate Brouwer fixed-point problem in the $d$-dimensional unit cube $[0,1]^d$ with $k$ rounds is $\Theta\bigl((1/\epsilon) ^{\frac{d^{k+1} - d^k}{d^k - 1}}\bigr)$, for both deterministic and randomized algorithms.
\end{theorem}

%Let $k \in \mathbb{N}$ be a constant. For any $\epsilon > 0$, when $d \geq 2$, the query complexity of finding an $\epsilon$-approximate Brouwer fixed-point on the $d$-dimensional unit cube $[0,1]^d$ in $k$ rounds is $\Theta\left((1/\epsilon) ^{\frac{d^{k+1} - d^k}{d^k - 1}}\right)$, for both deterministic and randomized algorithms (Theorem~\ref{thm:brouwer_main}). 

We also show that when $d = 1$, the query complexity of finding an $\epsilon$-approximate Brouwer fixed point in $k$ rounds is $\Theta\left((1/\epsilon)^{1/k}\right)$, for both deterministic and randomized algorithms. % (Corollary~\ref{thm:1D_BR_epsilon}). 

\section{Related Work} \label{sec:related_work}
%\simina{Polish and reorder these paragraphs}
%\jiawei{I'd like to present related work in three lines: black-box, white-box, and round. I have just copy each paragraph into the corresponding place.}

\paragraph{Query complexity of local search.} The query complexity of local search was studied first experimentally by Tovey~\cite{tovey81}, while Aldous~\cite{aldous1983minimization} provided the first theoretical analysis.
%For any graph $G$ with $N$ vertices and maximum degree $d$, Aldous' algorithm can be seen as steepest descent with warm start and works as follows: first query $t$ vertices $x_1, \ldots, x_t$ selected uniformly at random and pick the vertex $x^*$ that minimizes the function among these \footnote{That is, the vertex $x^*$ is defined as: $x^* = x_j$, where $j = \arg \min_{i=1}^t f(x_i)$.}. Then run steepest descent from $x^*$ and stop when no further improvement can be made, returning the final vertex reached. When $t = \sqrt{dN}$, the algorithm issues $O(\sqrt{dN})$ queries in expectation and has roughly as many rounds of interaction with the oracle. %\cite{aldous1983minimization} also gave lower bounds for special classes of graphs such as the hypercube, by showing that the function defined by the hitting time [TODO?] 
%\jiawei{``warm-start'' or ``warm start''? We're not consistent now.}
 For the Boolean hypercube $\{0,1\}^n$, Aldous~\cite{aldous1983minimization} gave an upper bound of  $O(n \cdot 2^{n/2})$ using steepest descent with a warm start %In the special case where the graph is the $n$-dimensional grid $[n]^d$, Aldous' algorithm takes $O(n^{d/2})$ queries in expectation and runs in $O(n^{d/2})$ rounds.
 and a lower bound of $\Omega(2^{n/2 - o(n)})$ for randomized algorithms. 
 The lower bound construction is as follows. Consider an initial vertex $v_0$ uniformly at random. The function value at $v_0$ is $f(v_0) = 0$. From this vertex, start an unbiased random walk $v_0, v_1, \ldots$ For each vertex $v$ in the graph, set $f(v)$ equal to the first hitting time of the walk at $v$; that is, $f(v) = \min\{t \mid v_t=v\}$. The function $f$ defined this way has a unique local minimum at $v_0$. By analyzing this distribution, Aldous~\cite{aldous1983minimization} showed a lower bound of $2^{n/2-o(n)}$ on the Boolean hypercube. %Follow-up work by \cite{} sharpened the bounds for the $d$-dimensional grid and the hypercube.

The steepest descent with a warm start algorithm \cite{aldous1983minimization} is  essentially optimal for graphs such as the hypercube and the $d$-dimensional grid (\cite{aldous1983minimization, sun2009quantum, zhang2009tight}). At the same time, the algorithm is highly sequential.
%However, multiple rounds of interaction can be expensive in applications. For example,
%when algorithms such as gradient descent are run on data stored in the cloud, there can be delays due to back and forth messaging across the network. A remedy for such delays is designing protocols with fewer rounds, which is the focus of our work.   % For example, on the $d$-dimensional grid, the algorithm has expected query complexity $O(n^{d/2})$  and runs in  $O(n^{d/2})$ rounds.
%The grid is a well-known graph that arises in applications where there is a continuous search space, which can be discretized to obtain approximate solutions (e.g. for computing a fixed point or a stationary point of a function).
Llewellyn, Tovey, and Trick~\cite{llewellyn1989local} analyzed the deterministic query complexity of local search and devised a divide-and-conquer approach, which has higher total query complexity but uses fewer rounds. Their algorithm is deterministic and identifies in the first step a vertex separator $S$ of the input graph $G$ \footnote{A vertex separator is a set of vertices $S \subseteq V$ with the property: there exist vertices $u,v \in V$, where $V$ is the set of vertices of $G$, such that any path between $u$ and $v$ passes through $S$.}. Afterwards, it queries all the vertices in $S$ to find the minimum $v$ among these. If $v$ is a local minimum of $G$, then return it. Otherwise, there is a neighbour $w$ of $v$ with $f(w) < f(v)$. Repeat the whole procedure on the new graph $G'$, defined as the connected component of $G \setminus S$ containing $w$. Correctness holds since the steepest descent from $w$ cannot escape $G'$. On the $d$-dimensional grid, the vertex separator $S$ can be defined as the $(d-1)$-dimensional wall that divides the current connected component evenly; thus a local optimum can be found with $O(n^{d-1})$ queries in $O(\log{n})$ rounds.
%This algorithm finds a local minimum of $f$ on any planar graph $G$, with $O(\sqrt{n} + \delta \log{n})$ queries when the maximum degree of the graph is $\delta$. \cite{llewellyn1989local} also study the hypercube, showing that the deterministic query complexity of local search is $\Theta(2^n/\sqrt{n})$.
%The deterministic query complexity of local search was first studied by Llewellyn,  Tovey, and Trick~\cite{llewellyn1989local}, who gave upper and lower bounds.

%\jiawei{This paragraph is appearing too late here. Need to restructure this section.}
We observe the contrast between the steepest descent with a warm start algorithm given by Aldous \cite{aldous1983minimization}, which is randomized and almost sequential---running in $O(n^{d/2})$ rounds---and the deterministic divide-and-conquer algorithm from \cite{llewellyn1989local}, which can be implemented in $O(\log{n})$ rounds. Even though the randomized algorithm in \cite{aldous1983minimization}  is (essentially) optimal in terms of number of queries, it takes many rounds. % and so it cannot be parallelized directly .
Thus it is natural to ask whether steepest descent with a warm start  can be parallelized and what is the tradeoff between the total query complexity and the number of rounds.

Alth{\"{o}}fer and Koschnick~\cite{AK93}, and Llewellyn and Tovey~\cite{DBLP:journals/dam/LlewellynT93} applied the adversarial argument 
%proposed in \cite{llewellyn1989local}
to show that $\Omega(n^{d-1})$ queries are necessary for any deterministic algorithm on the $d$-dimensional grid of side length $n$. Llewellyn, Tovey, and Trick~\cite{llewellyn1989local} also studied arbitrary graphs, showing that $O(\sqrt{N} + \Delta \log{N})$ queries are sufficient on graphs with $N$ vertices when the maximum degree of the graph is $\Delta$ and the graph has constant genus.

Aaronson~\cite{aaronson2006lower} improved the  lower bounds for randomized algorithms by designing a novel technique called the \textit{relational adversarial method} inspired by the adversarial method in quantum computing. This method avoids analyzing the posterior distribution during the execution directly and gave improved lower bounds for both the hypercube and the grid. % a lower bound of $\Omega(2^{n/2}/n^2)$ for the randomized query complexity of local search on the hypercube in the classical setting and $\Omega(2^{n/4}/n)$ for quantum algorithms.\jiawei{This is not clear. Aaronson works on both hypercube and grid, while Zhang and Sun and Yao only works on grid.}
Follow-up work by Zhang~\cite{zhang2009tight} and Sun and Yao~\cite{sun2009quantum} obtained even tighter lower bounds for the grid using this method with better choices on the random process; their lower bound is $\widetilde{\Omega} \bigl(n^{\frac{d}{2}} \bigr)$, which is nearly optimal. %Together with the randomized algorithm of \cite{aldous1983minimization} which has a warm start, these works concluded a near optimal bound of $\widetilde{\Theta} \bigl(n^{\frac{d}{2}} \bigr)$ for randomized algorithms on the hypercube in the classical setting.

Santha and Szegedy~\cite{santha2004quantum} gave a quantum lower bound of $\Omega\left( \sqrt[8]{\frac{s}{\Delta}} / \log(N) \right)$, where $s$ is the separation number of the graph and $\Delta$ the maximum degree. This  implies the same lower bound in a randomized context, using the spectral method.
Meanwhile, the best known upper bound is $O((s + \Delta) \cdot \log N)$ due to \cite{santha2004quantum}, which was obtained via a refinement of the divide-and-conquer procedure (\cite{llewellyn1989local}). 

Dinh and Russell~\cite{dinh2010quantum} studied Cayley and vertex transitive graphs and gave lower bounds for local search as a function of the number of vertices and the diameter of the graph. 
Verhoeven~\cite{Verhoeven06} obtained upper bounds as a function of the genus of the graph. Br\^anzei, Choo, and Recker~\cite{BCR24} gave lower bounds for arbitrary graphs, as a function of the vertex congestion and separation number of the graph. These imply a lower bound of $\Omega(\sqrt{N}/\log{N})$ for bounded-degree expanders, nearly matching the upper bound of $O(\sqrt{N})$ given by the randomized algorithm from Aldous~\cite{aldous1983minimization} for such graphs.

% Aaronson~\cite{aaronson2006lower} provided the first non-trivial randomized lower bound for the $d$-dimensional grid by designing a novel technique called the \textit{relational adversarial method} inspired by the adversarial method in quantum computing. Follow-up work by \cite{zhang2009tight} and Sun and \cite{sun2009quantum} obtained tighter lower bounds using this method with better choices on the random process, and eventually showed a bound of $\widetilde{\Omega}(n^{d/2})$ for any $d$-dimensional grid. Therefore, the divide-and-conquer algorithm and the warm-start algorithm are the optimal deterministic and randomized algorithm respectively. \cite{aaronson2006lower} also studied the randomized query complexity of the hypercube.

\paragraph{Communication complexity of local search.} Babichenko,  Dobzinski, and Nisan~\cite{BSN19} analyzed the communication complexity of local search, which  captures the hardness of finding a local optimum in distributed environments, where data may be stored on different computers, from the point of view of the communication cost. %\jiawei{I still think the structure of this two paragraph is weird. We start with talking ``parallel complexity'', and then say something about the $d$-grid, and then back to parallel complexity.}

\paragraph{White box complexity of local search and related problems.} The computational complexity of local search in the white-box setting is captured by the class PLS, which was defined by Johnson, Papadimitriou, and Yannakakis~\cite{JohnsonPY88} to model the difficulty of finding locally optimal solutions to optimization problems. 
%More formally, we are circuits $N$ and $V$ that assign to each $n$-bit string $x$ a neighbourhood $N(x)$ of polynomial size and a non-negative integer value $V(x)$. The goal is to find a string $x$ with $V(x) \geq V(y)$ for all $y \in N(x)$. 
%This setting is known as the white box model, since the representation of the function $f$ is part of the input.  Natural PLS complete problems include finding a pure Nash equilibrium in a congestion game~\cite{FPT04} and a locally optimum maximum cut in a graph~\cite{SchafferY91}.  

A class related to PLS is PPAD, introduced by Papadimitriou~\cite{Papadimitriou_1994} to study the computational complexity of finding a Brouwer fixed-point. PPAD contains many natural problems that are computationally equivalent to finding a Brouwer fixed point (\cite{CD09}), such as finding an approximate Nash equilibrium in a multi-player or two-player game (\cite{DGP09,CDT09}), an Arrow-Debreu equilibrium in a market (\cite{VY11,CPY17}), and a local min-max point (\cite{DSZ20}). 
The query complexity of computing an $\epsilon$-approximate Brouwer fixed point was studied in a series of papers for fully adaptive algorithms (\cite{hirsch1989exponential, chen2005algorithms, chen2007paths}).
%starting with \cite{hirsch1989exponential}, later improved by \cite{chen2005algorithms} and \cite{chen2007paths}.

The classes PLS and PPAD are both  a subset of TFNP, which contains total function problems that can be solved in nondeterministic polynomial time.  Fearnley et al.~\cite{FGHS20} showed that the class CLS, introduced by Daskalakis and Papadimitriou~\cite{DP11} to capture continuous local search, is equal to PPAD $\cap$ PLS. The query complexity of continuous local search has also been studied (see, e.g., \cite{hubavcek2017hardness}).

The $d$-dimensional grid with constant dimension  is important in the study of TFNP in complexity theory. For example, the complexity class CLS was first defined as any problem that could be reduced to the  \textsc{3D-Continuous-Localopt} problem (\cite{DP11}). Göös et al.~\cite{goos2022further} showed that
EOPL = PPAD $\cap$ PLS using the 2D grid as a unified model to define problems from several
complexity classes (PPAD, PPADS, PLS, EOPL, and SOPL). 

%\jiawei{Do we still need this paragraph? I think previous paragraph is enough.}

\begin{comment}
\cite{NEMIROVSKI} considered the parallel complexity of optimization. This was analyzed in a series of follow-up works for submodular functions (see, e.g., \cite{BalkanskiS18,BalkanskiRS19,EN19,BalkanskiS20}). %\cite{DBLP:conf/colt/BubeckM20} studied 
Algorithms with few rounds of interaction with the function oracle (aka low depth) for the problem of computing stationary points were studied in ~\cite{DBLP:conf/colt/BubeckM20}. \cite{BubeckJLLS19} studied parallel convex optimization where the oracle can answer $poly(d)$ queries in each round, where $d$ is the dimension.

\end{comment}

Our results also imply bounds for the computational problem \textsc{Brouwer}, where the goal is to compute an approximate fixed point of a Lipschitz function defined on the $d$-dimensional cube, the existence of which is guaranteed by Brouwer's theorem. \textsc{Brouwer} is computationally equivalent to problems such as finding a Nash equilibrium in an $n$-player game (\cite{Nash48,Papadimitriou_1994, DGP09, CDT09}) and a local min-max equilibrium in a two player game with nonconvex-nonconcave utilities (\cite{DSZ20}), the latter of which has applications to training generative adversarial networks (\cite{NIPS2014_GANs,ACB17,DSZ20}). %\cite{} give lower bounds in the oracle model for Brouwer 

\paragraph{Parallel complexity.} Parallel complexity is a fundamental concept, which was studied extensively for problems such as sorting, selection, finding the maximum element of a vector, and the sorted top-$k$ elements (\cite{Val75,pippenger1987sorting,bollobas1988sorting,alon1986tight,wigderson1999expanders,GasarchGK03,BravermanMW16,BravermanMP19,Cohen-AddadMM20}). An overview on parallel sorting algorithms is given in the book of Akl~\cite{Akl_book}. 

The parallel complexity of optimization was first considered by Nemirovski~\cite{NEMIROVSKI}. This was analyzed in a series of follow-up works for submodular functions (see, e.g., \cite{BalkanskiS18,BalkanskiRS19,EN19,BalkanskiS20}). %\cite{DBLP:conf/colt/BubeckM20} studied 
Algorithms with few rounds of interaction with the function oracle (aka low depth) for the problem of computing stationary points were considered by Bubeck and Mikulincer~\cite{DBLP:conf/colt/BubeckM20}. Bubeck et al.~\cite{BubeckJLLS19} studied parallel convex optimization where the oracle can answer $poly(d)$ queries in each round, where $d$ is the dimension. 

Another setting of interest where rounds are important is active learning, where there is an ``active'' learner
who can submit queries —-- in the form of unlabeled instances —-- to be annotated by an oracle
(e.g., a human) (\cite{Set12}). However each round of interaction with the human annotator has a
cost, which can be captured through a budget on the number of rounds. %\jiawei{overfull here.}

\section{Local Search}  \label{sec:local_search_main}
In this section we state our results for local search and give an overview of the proofs.

\subsection{Local Search in Constant Rounds}    

When the number of rounds $k$ is a constant, we obtain the following bounds.

%\begin{theorem}(Local search, constant rounds)\label{thm:1}
%	Let $k \in N$ be a constant. The query complexity of local search in $k$ rounds on the $d$-dimensional grid $[n]^d$ is $\Theta\bigl(n^{\frac{d^{k+1} - d^k}{d^k - 1}}\bigl)$, for both deterministic and randomized algorithms.
%\end{theorem}

\noindent \textbf{Theorem \ref{thm:1}} (Local search, constant rounds, restated): \emph{Let $d, k\in \mathbb{N}$ be constant, where $d \geq 2$ and $k \geq 1$. The query complexity of local search in $k$ rounds on the $d$-dimensional grid $[n]^d$ is $\Theta\bigl(n^{\frac{d^{k+1} - d^k}{d^k - 1}}\bigl)$, for both deterministic and randomized algorithms.}

For example, when  $k=2$, the query complexity is $\Theta\bigl(n^{\frac{d^{2}}{d+1}}\bigl)$.

%When $k \rightarrow \infty$, this bound is close to $\Theta(n^{d-1})$, with gap smaller than any polynomial. The classical result in \cite{llewellyn1989local} showed that the query complexity of local search for deterministic algorithm is $\Theta(n^{d-1})$, and the upper bound is achieved by a divide-and-conquer algorithm with $O(\log n)$ rounds. 

%\jiawei{Reviewer suggests `` One round was not considered before.  A bit more of comment here''}

%Thus our result fills the gap between one round algorithms and logarithmic rounds algorithms except for a small margin. This theorem also implies that randomness does not help when the number of rounds is constant.

% \medskip 

% We provide an overview of the proofs; the complete proofs can be found in the appendix. 
\subsubsection{Upper bound overview for local search in constant rounds.}

The algorithm for constant rounds works as follows. 
%We divide the search space into many disjoint sub-cubes of side length $\ell_i \coloneqq n^{\frac{d^k-d^i}{d^k-1}}$ in round $i$, query their boundary, then continue the search into the one that contains the point with the minimal value. Note that the steepest descent from the minimal point could not escape from its own sub-cube. In the last round, we query all the points in the current sub-cube and get the solution. The side length $\ell_i$ of sub-cubes in round $i$ is chosen by equalizing the number of queries in each round. 
%The algorithm can be seen as generalization of the classic deterministic divide-and-conquer algorithm in~\cite{llewellyn1989local}.
%\jiawei{}
Initialize a cube $C_0$ to the whole grid $[n]^d$. 

In each round $i=1,\ldots,k-1$, the algorithm divides the current cube $C_{i-1}$ into a set of mutually exclusive sub-cubes $C_i^1, \ldots, C_{i}^{n_i}$ of side length $\ell_i$ (for a carefully set value of $\ell_i$) that cover $C_{i-1}$. 
	Then it queries all the points on the boundary of the sub-cubes $C_i^1, \ldots, C_{i}^{n_i}$ and select the point $\Vec{x}^*_i$ with minimal value among them.
		 Set $C_i = C_i^j$, where $C_i^j$ is the sub-cube that $\Vec{x}^*_i$ belongs to and repeat. 
		 Finally, in round $k$,  query all the points in the current cube $C_{k-1}$ and find the solution point.
		 
		 Figure~\ref{fig:HDD:main} shows an illustration for the case $k=2$ and $d=2$.

		 	\begin{figure}[h!]
		\centering
		\includegraphics[scale=0.5]{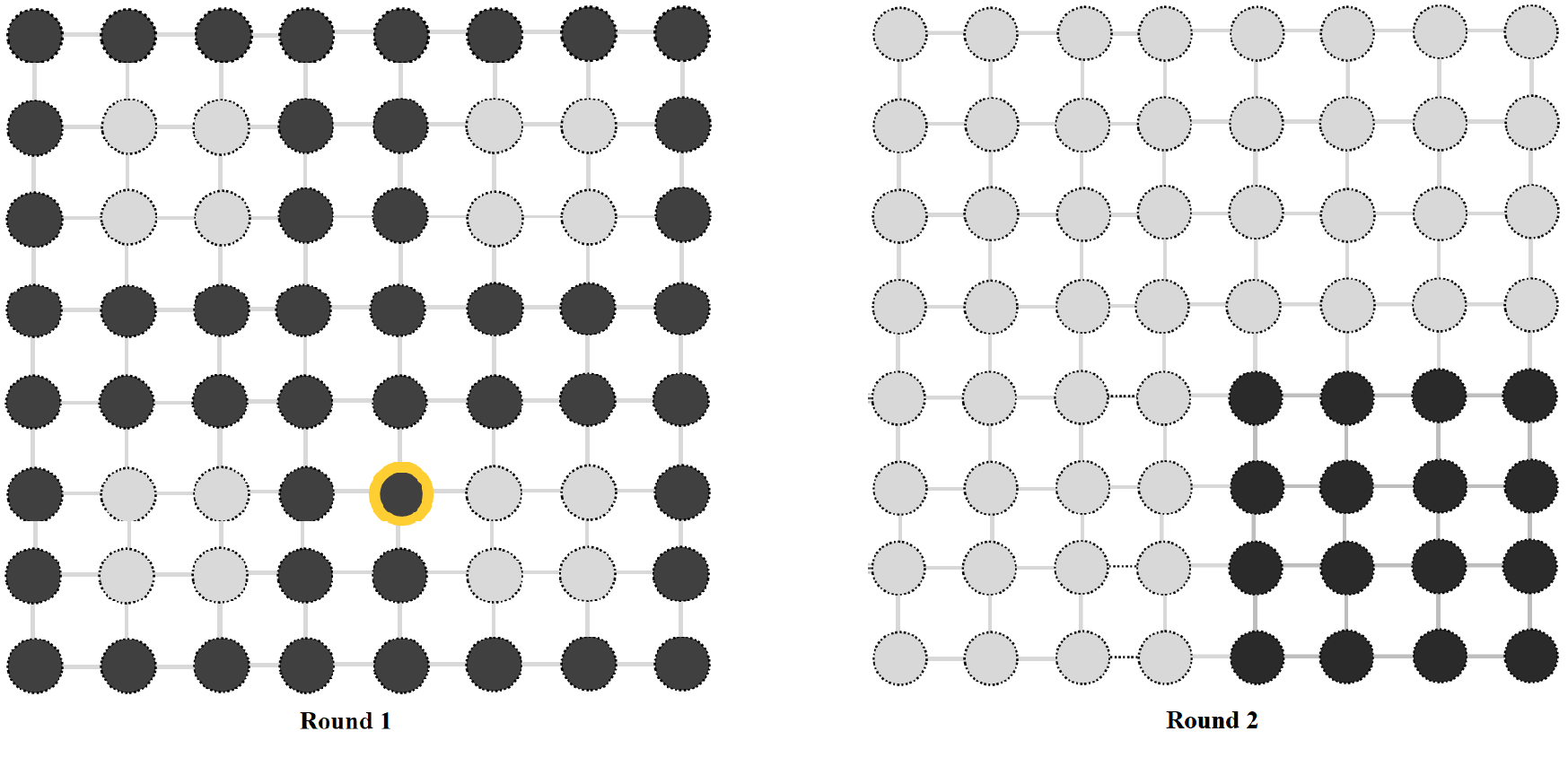}
		\caption{2D grid of size $8\times 8$. Suppose there are two rounds. In round $1$, illustrated on the left, the algorithm queries all the black points and selects the minimum among all these points (shown in  yellow). In round 2, illustrated on the right, it queries the entire sub-square in which the minimum from the first round was found. When $d=2$ and $k=2$, the query complexity is $\Theta\bigl(n^{\frac{4}{3}}\bigl)$ by setting $\ell_1 \coloneqq n^{\frac{2}{3}}$.}
		\label{fig:HDD:main}
	\end{figure}

%The proof is in the full version of the paper.  %Appendix~\ref{sec:LS.ConstAlg}.
The main observation that  allows the proof to work is that when the space is divided into sub-cubes and $\Vec{x}^*_i$ is the point with smallest value on the boundaries of all the sub-cubes, steepest descent started at $\Vec{x}^*_i$ cannot escape the sub-cube it is contained in. 
By choosing the side lengths of the sub-cubes appropriately, we get the required upper bound. %This algorithm can be seen as a generalization of the divide-and-conquer algorithm in~\cite{llewellyn1989local}. 
%\jiawei{Rewrite this part depending on the structure of the paper.}
The detailed proof is included in Appendix~\ref{sec:LS.ConstAlg}, together with the other omitted proofs.

\subsubsection{Lower bound overview for local search in constant rounds.} To show randomized lower bounds, we apply Yao's minimax theorem~(\cite{Yao77}): first we provide a hard distribution of inputs, then show that no \emph{deterministic} algorithm could achieve accuracy larger than some constant on this distribution.
The hard distribution will be given by a \textit{staircase} construction (\cite{vavasis1993black,hirsch1989exponential}). A {staircase} will be a random path with the property that the unique local optimum is {hidden} at the end of the path. We present a sketch here, while the complete calculations can be found in Appendix~\ref{sec:LS.ConstLB}.
%Appendix~\ref{sec:LS.ConstLB} for the complete calculations.

%\jiawei{The figure on the left could look slightly better! The number is too small, and is not centered.}
Figure~\ref{fig:staircases_local} shows an example of a staircase, which consists of the black and red vertices. The bottom left black vertex is the starting point of the staircase and the value of the function there is set to zero. Then the value decreases by one with each step along the staircase, like going down the stairs. The value of the function at any point outside the {staircase} is equal to the distance to the entrance of the staircase. 

\vspace{2mm}
\begin{figure}[h!]
	\centering 
	\subfigure[Staircase in 2D for 3 rounds.]
	{
		\includegraphics[scale=0.38]{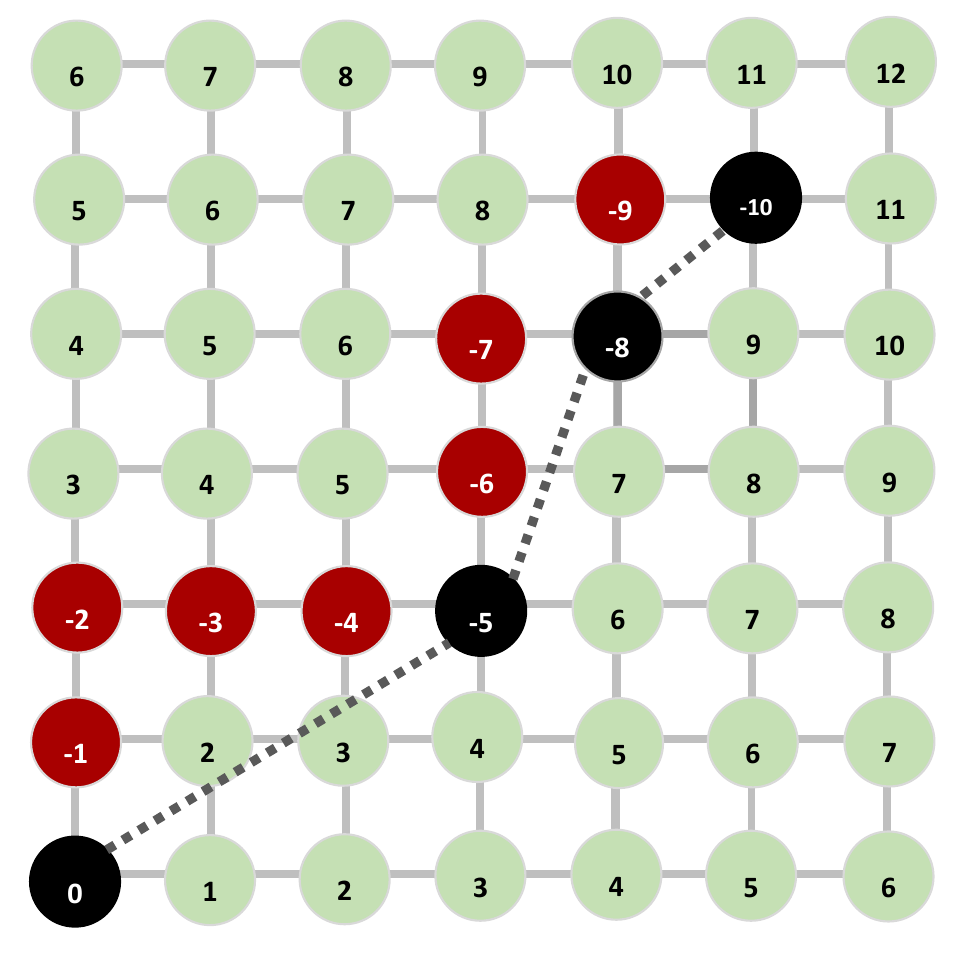}
		\label{fig:2D_staircase}
	}
	\subfigure[Stylized staircase in 2D.]
	{
		\includegraphics[scale=0.62]{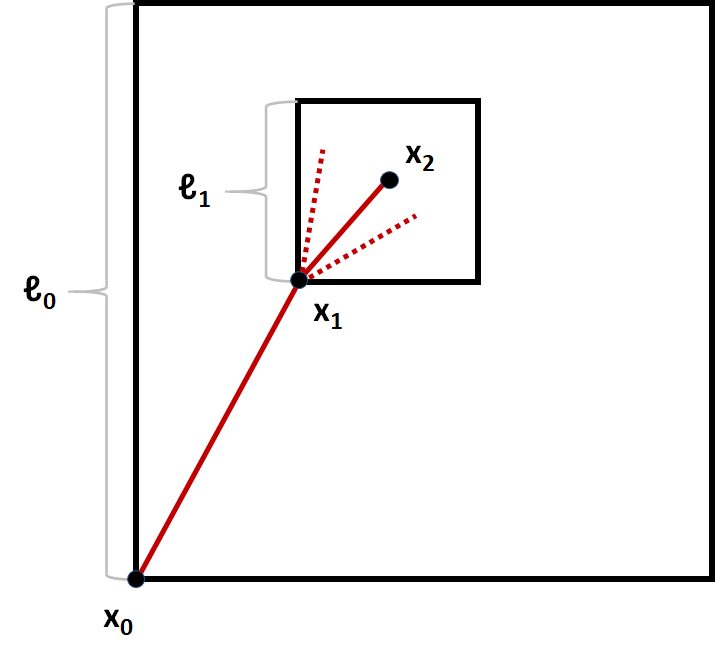}
		\label{fig:2D_staircase_stylized}
	}
	\caption{\small{Figure (a) shows a two dimensional staircase for $k=3$ rounds. The number of black points is equal to $k+1$. The black points are connected by ``folded-segments'' shown in red. The idea is that in each round, the algorithm will learn at most one more folded-segment.
			Figure (b) shows in a stylized way the process of selecting the next part of the staircase given that the first folded-segment was determined. The folded-segments are displayed here as straight lines for simplicity. }}
	\label{fig:staircases_local}
\end{figure}

Intuitively, the algorithm cannot find much useful structural information of the input and thus has no advantage over a path following algorithm. The staircase construction can be embedded in two  different tasks: finding a local minimum of a function (\cite{aldous1983minimization,aaronson2006lower,sun2009quantum,zhang2009tight,hubavcek2017hardness}) and computing a Brouwer fixed-point (\cite{chen2007paths,hirsch1989exponential}). 

% \jiawei{I think we could emphasize more on how challenging the proof of randomized lower bound is.}
The most challenging part is rigorously proving that such intuition is correct. The lower bound proof from \cite{aldous1983minimization} for the hypercube $\{0,1 \}^n$ has complex probability analysis on the random walk and relies on the structural property of hypercube $\{0,1 \}^n$.

All other works on the randomized lower bound of local search are based on Aaronson's relational adversarial method (\cite{aaronson2006lower}). However, we found it's inherently difficult to consider multiple queries in one round in the framework of relational adversarial method. % In the fully adaptive setting, Aldous~\cite{aldous1983minimization} directly analyzes the posterior distribution during the execution of the algorithm, with complex probability analysis on the random walk. In \cite{aaronson2006lower}, Aaronson proposed a novel technique called the \textit{relational adversarial method} inspired by the adversarial method in quantum computing, which avoids analyzing the posterior distribution during the execution directly. Later, Zhang~\cite{zhang2009tight}, and Sun and Yao~\cite{sun2009quantum} obtained a tighter lower bound using this method with better choices on the random process and sophisticated analysis.
Therefore, our main technical innovation is a new technique to incorporate the round limit into the randomized lower bounds.%, as we were not able to obtain a lower bound for rounds using the methods previously mentioned.  
This could also serve as a simpler alternative method of the classical relational adversarial method (\cite{aaronson2006lower}) in the fully adaptive setting.

\paragraph{Staircase definition.} We define a staircase as an array of connecting grid points $(\Vec{x}_0, \Vec{x}_1, \ldots, \Vec{x}_t)$, for $0 \leq t \leq k$. A uniquely determined path called folded-segment is used to link every two consecutive points $(\Vec{x}_i, \Vec{x}_{i+1})$.
The start point $\Vec{x}_0$ is fixed at corner $\mathbf{1}$, and the remaining connecting points are chosen randomly in a smaller and smaller cube region with previous connecting points as corner. % See Figure~\ref{fig:staircases_local} for an example in 2D.
%Recall that the constant rounds algorithm is also narrowing down the search region after each round. Actually the size of cube region for choosing $i$-th connecting point $\Vec{x}_i$ conditioning on $\Vec{x}_{i-1}$ is same as the size of the search region of constant rounds algorithm after round $i-1$, which is not a coincidence.
For $k$ rounds algorithms, we choose a distribution of staircases of ``length'' $k$, where the length is defined as the number of connecting points in the staircase minus $1$. %See Figure~\ref{fig:staircases_local} for a typical staircase structure.

% \jiawei{Shall we move the stylized staircase figure to the appendix?}

\paragraph{Good staircases.} We say that a length $t$ staircase is ``{good}'' with respect to a deterministic algorithm $\mathcal{A}$ if for each $1 \leq i < t$, any point in the suffix of the staircase (i.e. after the connecting point $\Vec{x}_i$ \footnote{That is, the point $\Vec{x}_i$ is not included.}) is not queried in rounds $1, \ldots, i$, when $\mathcal{A}$ runs on the input generated by this staircase. 

The input functions generated by {good staircases} are like adversarial inputs: $\mathcal{A}$ could only (roughly) learn the location of the next connecting point $\Vec{x}_i$ in each round $i$, and still know little about the staircase from $\Vec{x}_i$ onwards.

We show that if ${9}/{10}$ of all possible length $k$ staircases are {good}, then the algorithm  will make a mistake with probability at least ${7}/{40}$. %(Lemma~\ref{lem:good1}).
We ensure that each possible staircase is chosen with the same probability; their total number is easy to estimate.

Thus the main technical part of our proof is counting the number of {good staircases}.
%Our proof could be equivalently written in the language of conditional probability, but it will incur many unnecessary complications.

\paragraph{Counting good staircases.}  We show the next properties about the \emph{prefix} of {good staircases}: %  are  proved in Lemma~\ref{lem:good0}:
\begin{description}
	\item[\texttt{P1}:] If $s$ is a {good staircase}, then any ``{prefix}'' $s'$ of $s$ is also a {good staircase}. %A prefix of $s$ is any staircase formed by a prefix of the array of the connecting points of $s$.
	\item[\texttt{P2}:] Let $s_1,s_2$ be two {good staircases} with respect to algorithm $\mathcal{A}$. If the first $i$ connecting points of the staircases are same, then $\mathcal{A}$ will submit the same queries in rounds $1, \ldots, i$ when given as input the functions generated by $s_1$ and $ s_2$, respectively.
\end{description}

% While a natural approach is to count  good staircases in increasing order of their length, this will turn out to not work.

%Therefore, it's natural to count the number of {good staircase} in the increasing order of their length. The straightforward idea is to estimate the number of length $i+1$ good staircases from the number of length $i$ good staircases. Fixing a length $i$ good staircase $s^{(i)}$, a length $i+1$ staircase $s^{(i+1)}$ with $s^{(i)}$ as prefix is good if and only if the tail part between the last two connecting points of $s^{(i+1)}$ is \emph{not} hit by any queried points in round $1, \ldots, i+1$, if $\mathcal{A}$ is running on the instance generated by $s^{(i)}$.
%However, the algorithm $\mathcal{A}$ will make different set of queries on different instances generated by different length $i$ staircases, which makes deriving an expression between the number of length $i$ and $i+1$ good staircases explicitly looks impossible. 
%Resolving this dilemma needs \emph{taking one step back}. 
We first fix a good staircase $s^{(i-1)}$ of length $i-1$ and consider two good staircases $s_1, s_2$ of length $i$ that have $s^{(i-1)}$ as prefix. By property \texttt{P2}, the algorithm $\mathcal{A}$ will make the \emph{same} queries in rounds $1, \ldots, i$ when running on the inputs generated by $s_1$ and $s_2$, respectively. This enables estimating the total number of good staircases of length $i+1$  with $s^{(i-1)}$ as prefix. 
%\jiawei{I found this part very difficult to understand.  The next sentence is also too vague. We should try to replace the next sentence. }
Then by summing over all good staircases of length $i-1$, we get a recursive equation between the number of good staircases of length $i-1$, $i$, and $i+1$. This will be used to show that most staircases of length $k$ are good.
% \jiawei{Finish the overview of two-stage analysis method. Maybe modify the stylized staircase figure to illustrate this two-stage analysis.}

%In the first step, we show all staircases of length $1$ are {good} by definition. In each step $1 < i \leq k$, 
%we derive a recursive inequality to calculate the numbers of length $i$ {good staircases} by the number of length $i-1$ and $i-2$ {good staircases}. Combining all these recursive inequality, we finally prove that $9/10$ of all possible length $k$ staircases are good, which concludes our proof.

% \jiawei{TODO for Jiawei: Briefly introduce two-stage analysis}

% \jiawei{Maybe add brief introduction of two-stage analysis in the technical overview part.}

% The lower bound for polynomial rounds is similar, except we use a different distribution of staircases (see Section~\ref{sec:LS.PolyLB}).

\subsection{Local Search in Polynomial Rounds} \label{sec:poly_rounds_main}

When the number of rounds $k$ is polynomial in $n$, that is $k = n^{\alpha}$ for some constant $\alpha > 0$, the algorithm that yields the upper bound in Theorem~\ref{thm:1} is no longer efficient.

We design a different algorithm for this regime and also show an almost matching lower bound.
% \jiawei{We should mention we only need to care $d > 2$ in the polynomial rounds case. }
With polynomial rounds we can focus on $d \geq 3$. Recall that Sun and Yao~\cite{sun2009quantum} proved a lower bound of $\widetilde{\Omega}(n)$ for fully adaptive randomized algorithms in 2D and the divide-and-conquer algorithm by Llewellyn, Tovey, and Trick~\cite{llewellyn1989local} achieves this bound with only $O(\log(n))$ rounds.

% \jiawei{The presentation of this theorem is not good. It's not even clear.}
\bigskip 

%\textbf{\textcolor{red}{TODO: Correct phrasing of the theorem, we were stating that $k$ is a constant. See also the other place where it's stated.}}
\noindent \textbf{Theorem  \ref{thm:2}} (Local search, polynomial rounds, restated): \emph{Consider the $d$-dimensional grid with side length $n\in \mathbb{N}$, where $d \geq 2$.
Let   $k = n^{\alpha} \in \mathbb{N}$, where $\alpha \in \left(0,  {d}/{2} \right)$ is a constant. The randomized query complexity of local search in $k$ rounds on the $d$-dimensional grid $[n]^d$ is:
	\begin{itemize}
		\item $\; \; \;$ $\Theta\left(n^{(d-1) - \frac{d-2}{d}\alpha}\right)$ when $d \geq 5$;
		\item $\; \; \;$ $O\left(n^{3 - \frac{\alpha}{2}}\right)$ and $\widetilde{\Omega}\left(n^{3 - \frac{\alpha}{2}}\right)$ when $d = 4$;
		\item $\; \; \;$ $O\left(n^{2 - \frac{\alpha}{3}}\right)$ and $\Omega\left(\max(n^{2-\frac{2\alpha}{3}}, n^{\frac{3}{2}})\right)$ when $d = 3$. 
	\end{itemize}
}

%\jiawei{Should we write a corollary like this: ``If $q=n^{\gamma}$ of queries in one rounds are allowed, xxx rounds are needed. }

%When $\alpha \rightarrow 0$, the bound approaches $\Theta(n^{d-1})$, i.e., the bound of constant and logarithmic rounds algorithm. When $\alpha \rightarrow ({d}/{2})$, the upper bound is close to $\Theta(n^{\frac{d}{2}})$, i.e., the fully adaptive algorithm. Thus, our result fills the gaps between constant (or logarithmic) rounds algorithms and fully adaptive algorithms, except for a small gap when $d \in \{3, 4\}$.

\subsubsection{Upper bound overview for local search in polynomial rounds.}    
Since the constant rounds algorithm is not optimal in this regime, we design an algorithm that randomly samples many points in round $1$ and then starts searching for the solution from the smallest valued point from round $1$. This is similar to the algorithm in~\cite{aldous1983minimization}, except the steepest descent part of Aldous' algorithm is highly sequential. 

To get better parallelism, we design a recursive procedure (``fractal-like steepest descent'') which parallelizes the steepest descent steps at the cost of more queries. We present the main ideas next: %; the formal proof can be found in Appendix~\ref{sec:LS.algo_poly_rounds}.  

%\jiawei{If this explanation is better, I'll change the section 2.2 and the comments in bullet-point code respectively.}
%\jiawei{The language of this explanation is still not good and need to be polished more. But I hope the  idea of it is clear.}
%
%\jiawei{I think we could have a little bit detailed technical overview, and simply start section 2.2 with Notation paragraph.}

\paragraph{Sequential procedure.} Let $C(\Vec{x},s) \coloneqq \{\Vec{y} \in [n]^d: \|\Vec{y}-\Vec{x}\|_{\infty} \leq s \}$ be the set of grid points in the $d$-dimensional cube of side length $2 \cdot s$, centered at point $\Vec{x}$. Let $\rank(\Vec{x})$ be the number of points with smaller function value than point $\Vec{x}$. 

Assume that we already have a procedure $P$ and a number $s < n$ such that $P(\Vec{x})$ will either return a point $\Vec{y} \in C(\Vec{x},s)$ with $\rank(\Vec{y}) \leq \rank(\Vec{x}) - s$, or output a correct solution and halt. Suppose in both cases $P(\Vec{x})$ takes at most $r$ rounds and $Q$ queries in total for any $\Vec{x}$. % \jiawei{Shall we say ``For example, the base case of procedure $P$ is the steepest descent ... '' This may help reader to understand what is going on?}
For a concrete example, the base case of the procedure $P$ is $s$-steps of steepest descent. In this case, $P(\Vec{x})$ takes $s$ rounds and $O(s)$ queries.

If we want to find a point $\Vec{x}^*$ with $\rank(\Vec{x}^*) \leq \rank(\Vec{x}_0) - t \cdot s$ for any given $\Vec{x}_0$ or output a correct solution, the naive approach is to run $P$ \emph{sequentially} $t$ times, taking $$\Vec{y}_1 = P(\Vec{x}_0), \Vec{y}_2 = P(\Vec{y}_1), \ldots, \Vec{x}^* = \Vec{y}_t = P(\Vec{y}_{t-1}).$$
Since each call of $P$ must wait for the result from the previous call, the naive approach will take $t \cdot r$ rounds and $t \cdot Q$ queries.

%\jiawei{TODO: Briefly explain what P is.}

%\simina{This figure is not helping, the caption is confusing and it's too late to fix it, I'm commenting it out}
% \begin{figure}[h!]
%\centering
%\includegraphics[scale=0.65]{figures/FLSD_intro_ver2.png}
%\caption{\simina{This caption is confusing; see skype} Process of parallelizing a two round algorithm $P$ in 2D, where $P(\Vec{x})$ searches for a point $\Vec{x'}$ with $\rank(x') \leq \rank(x)-s$ in the neighboring region of $\Vec{x}$ with radius $s$. In round $1$ (left figure), the points on the boundary of the square centered at $\Vec{x}_0$ (green dots) are queried and $\Vec{x}_1$ is the minimum among them. Algorithm $P$ is also initiated  starting at $\Vec{x}_0$ \emph{in parallel}; this will find the point $\Vec{y}_1$ in $r$ rounds. In the second round (right figure), the points on the boundary of the square centered at $\Vec{x}_1$ (green dots) are queried and $\Vec{x}_2$ is the minimum; another call to $P$ is made \emph{in parallel} starting at $\Vec{x}_1$.}
%\label{fig:FLSD_intro}
%\end{figure}

\paragraph{Parallel procedure.} We can parallelize the previous procedure $P$ using auxiliary variables that are more expensive in queries, but  cheaper in rounds. For $i \in [t]$, let $\Vec{x}_i$ be the point with minimum function value on the boundary of cube $C(\Vec{x}_{i-1},s)$, which can be found in only \emph{one} round with $O(s^{d-1})$ queries after getting $\Vec{x}_{i-1}$, i.e., we get the location of $\Vec{x}_{i-1}$ at the start of round $i$. Next, we take $\Vec{y}_i$ to be $P(\Vec{x}_{i-1})$ instead of $P(\Vec{y}_{i-1})$; thus the location of $\Vec{y}_i$ will be available at round $i+r$.
%See Figure~\ref{fig:FLSD_intro} for an illustration of this process. 
To ensure correctness, we will compare the value of point $\Vec{y}_i$ with the value of point $\Vec{x}_i$. If $f(\Vec{x}_i) \leq f(\Vec{y}_i)$ then 
\begin{equation}\label{eq:rank_intro}
	\rank(\Vec{x}_i) \leq \rank(\Vec{y}_i) \leq \rank(\Vec{x}_{i-1}) - s \,.
\end{equation}
Otherwise, since $\Vec{y}_i \in C(\Vec{x}_{i-1},s)$ has smaller value than any point on the boundary of $C(\Vec{x}_{i-1},s)$, we could use a slightly modified version of the divide-and-conquer algorithm of \cite{llewellyn1989local} to find the solution within the sub-cube $C(\Vec{x},s)$ in $\log n$ rounds and $O(s^{d-1})$ queries, and then halt all running procedures. If $f(\Vec{x}_i) \leq f(\Vec{y}_i)$ holds for any $i \in [t]$, applying inequality~\ref{eq:rank_intro} for $t$ times we will get $\rank(\Vec{x}) \leq \rank(\Vec{x}_t) - t \cdot s$, so we could return $\Vec{x}^* = \Vec{x}_t$ in this case. This parallel approach will take only $t + r$ rounds and $O(t \cdot (Q + s^{d-1}))$ queries.

The base case of  procedure $P$  is  the steepest descent algorithm. Then, multiple layers of the recursive process as described above are implemented to ensure the round limit is met. % met the round limit constraint. 
The parameters of the algorithm, such as $s$ and $t$ above, are described in Appendix~\ref{sec:LS.algo_poly_rounds}.
%\jiawei{Rewrite this part depending on the structure of the paper.}

\subsubsection{Lower bound overview for local search in polynomial rounds.}

%    \jiawei{Further polish this subsection.}

For polynomial rounds, we still use a staircase construction and hide the solution at the end of the staircase. Recall the bottom left vertex will be the starting point of the staircase and the value of the function there is set to zero. Then the value decreases by one with each step along the path. The value of the function at any point outside the {staircase} is equal to the distance to the entrance of the staircase. 
%The overall idea for proving lower bound in polynomial rounds is similar to the case for constant rounds. We still use a random walk to generate the staircase and hide the solution at the end of the staircase. %We still quantitatively show that any fixed deterministic algorithm could learn a little more about the staircase structure in expectation. 

However, the case of polynomial number of rounds is both conceptually and technically more challenging.  We explain the main ideas next; the full proof is in  Appendix~\ref{sec:LS.PolyLB}. 
%\jiawei{Rewrite this part depending on the structure of the paper.}

\paragraph{Choice of random walk} Let $\mathcal{Q}_k$ denote the total number of queries allowed for an algorithm that runs in $k$ rounds. Let $\mathcal{Q}= \mathcal{Q}_k /k$ be the average number of queries in each round. 
The minimum point among $\mathcal{Q}_k / 2$ uniformly random queries will be at most $100 \cdot n^d / \mathcal{Q}_k$ steps away from the solution with high probability. 

We set the number of points in the staircase to $\Theta(n^d / \mathcal{Q}_k)$. This strikes a balance between two extremes.
If the staircase is too long, then an algorithm like steepest descent with warm-start~(\cite{aldous1983minimization}), which starts by querying many random points in round $1$, is likely to hit the staircase in a region that is O$(n^d/\mathcal{Q}_k)$ close to the endpoint.
If the staircase is too short, then an algorithm such as steepest descent will find the end of the staircase in a few rounds.

%    \simina{This is not a formal statement right? Remove the observation environment, write it in plain text}
% We observe that {the most difficult random walk used in the fully-adaptive setting may not be the most difficult one in the limited round setting, and vice versa.}

Since we choose the staircase via a random walk, there are two key factors affect the difficulty of finding the solution: the {mixing time} and what we call the ``local predictability'' of the walk. We now explain how the algorithm could exploit a slow mixing time or a high local predictability with two random walks in Figure~\ref{fig:walk_comparison}:

\begin{figure}[h!]
	\centering
	\includegraphics[scale=0.63]{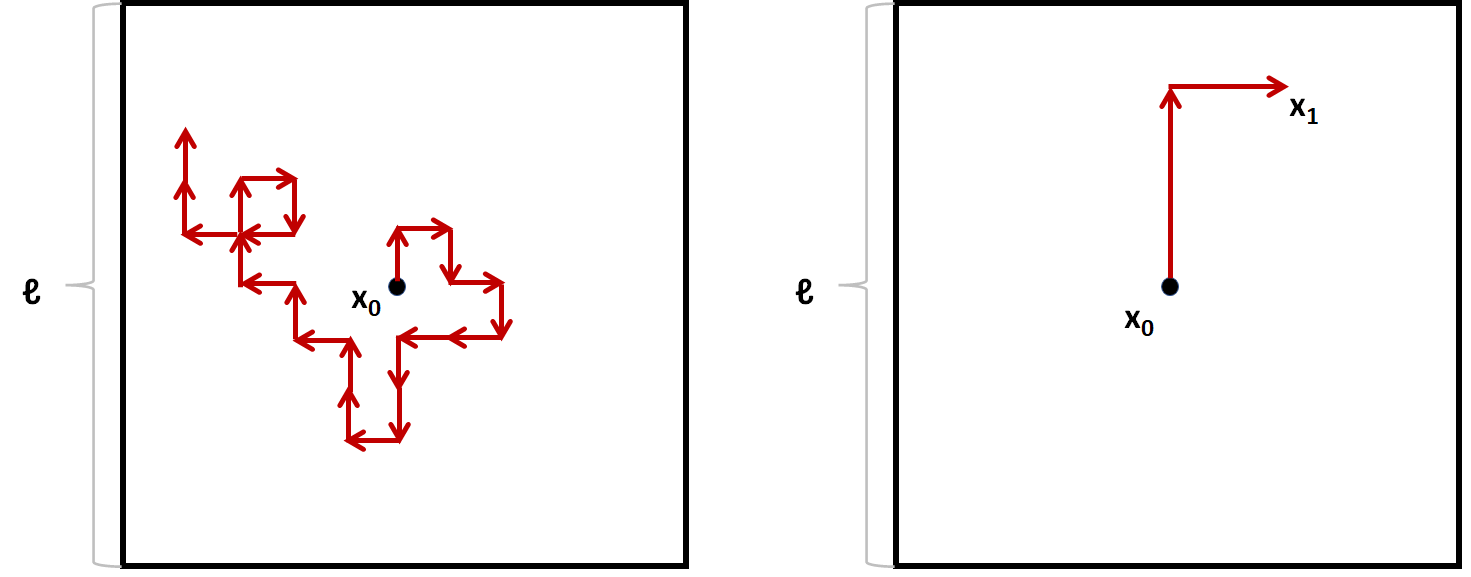}
	\caption{A 2D illustration for two types of random walk for the lower bound in polynomial rounds. The random walk on the left is more convoluted locally, but also mixes slower; the walk on the right has simpler structure, but mixes faster.}
	\label{fig:walk_comparison}
\end{figure}

%\jiawei{Maybe we could still improve these two figures.}

The first random walk (Figure~\ref{fig:walk_comparison}, left) randomly moves to one of its neighbor in each step. This random walk has very low local predictability and may be difficult for fully adaptive algorithms to learn. However, it mixes rather slowly, which could be exploited by the following simple algorithm: In each round, it queries all the points on the boundary of the cube of side length $s$, with the current best point as the center. This random walk will keep dwelling within this cube before reaching the boundary of it for $\Theta(s^2)$ steps in expectation. Therefore, the best point on the boundary gets $\Theta(s^2)$ steps closer to the end point than the previous best point in expectation. However, we only allow this algorithm to get $O(s)$ steps closer to the end if we want a tight bound.

The second random walk (Figure~\ref{fig:walk_comparison}, right) moves from $x_0$ to a uniform random point $x_1$ selected from a cube of side length $\ell$ centered at $x_0$ and uses $d$ straight segments to connect these two points.
% moves to a uniformly random point $\Vec{x}_1$ in a cube of side length $\ell$, with the current point $\Vec{x}_0$ as center, and uses $d$ straight segments to connect these two points. 
This random walk is more locally predictable, since each straight segment of length $O(\ell)$ could be learned with $O(\log \ell)$ queries via binary search. Thus, this walk could be tracked more efficiently than the naive steepest descent in the fully adaptive setting. On the other hand, this random walk mixes faster than the one in the left figure: it takes only $O(\ell)$ steps to mix in the cube region of side length $\ell$. Thus, the second random walk might be better when there aren't enough rounds to find every straight segment via binary search.

By controlling the parameter $\ell$, we get a trade-off between the {mixing time} and the {local predictability}. %when the total length of the walk is fixed. 
We choose $\ell = \Theta(\mathcal{Q}^{1/(d-1)}) = n^{1-2\alpha/d}$ for $k = n^{\alpha}$ in our proof, which corresponds to the max possible side length of the cube if using $\mathcal{Q}$ queries to cover its boundary.
%\jiawei{Not clear? Maybe explicitly show that $\ell^{d-1} = \mathcal{Q}?$}
% In the formal proof in Appendix~\ref{sec:LS.PolyLB}, we defined a constant $\Gamma$ in 
\paragraph{Measuring the Progress}
{Good staircases} are a central concept in the proof for constant rounds. Roughly, the algorithm can learn the location of exactly one more connecting point in each round. However, such a requirement is too strong with polynomial rounds. 

Instead, we allow the algorithm to learn more than one connecting points in some rounds, while showing that it learns no more than two connecting points in each round in expectation. 

Using amortized analysis, we quantify the maximum possible \emph{progress} of an algorithm in each round by a constant $\Gamma$, which only depends on the random walk, not the algorithm. Constant $\Gamma$ could be viewed as the difficulty of the random walk, which takes both the mixing time and the local predictability into account.

\begin{comment}
\paragraph{Self Intersection}
Notice that there is no self intersection of staircase in constant round case, as they will only going monotonically in each dimension. However, the new random walk may generate some staircases with self intersection, which are undesirable because
\begin{enumerate}
\item they makes the value of the intersecting point not well-defined;
\item we establish the lower bound for stationary point by reducing local search to it in section~\ref{sec:SP}, and self intersection will significantly complicates the process of embedding a staircase to 1-smooth function. 
\end{enumerate}
\end{comment}

%\vspace{-2mm}
\section{Brouwer} \label{sec:brouwer_main}
%\jiawei{This section could possibly be longer.}

The problem of finding an $\epsilon$-approximate fixed point of a continuous function was defined in Section~\ref{sec:model_results}. To quantify the query complexity of this problem, it is useful to consider a discrete version, obtained by discretizing the unit cube $[0,1]^d$. The discrete version of Brouwer is equivalent to the approximate fixed point problem in the continuous setting (\cite{chen2005algorithms}). 

%\begin{theorem} \label{thm:brouwer_main}
%	Let $k \in \mathbb{N}$ be a constant. For any $\epsilon > 0$, the query complexity of the $\epsilon$-approximate Brouwer fixed-point problem in the $d$-dimensional unit cube $[0,1]^d$ with $k$ rounds is $\Theta\left((1/\epsilon) ^{\frac{d^{k+1} - d^k}{d^k - 1}}\right)$, for both deterministic and randomized algorithms.
%\end{theorem}

The  algorithm for  Brouwer is reminiscent of the constant rounds algorithm for local search. 
We divide the space in sub-cubes  then find the one guaranteed to have a solution by checking a boundary condition given in~\cite{chen2005algorithms}. Then a parity argument will show there is always a sub-cube satisfying the boundary condition. In the last round, the algorithm queries all the points in the remaining sub-cube and returns the solution. 

The lower bound for Brouwer is obtained by reducing local search instances {generated by staircases} to discrete fixed-point instances. We can naturally let the staircase within the local search instance to be a long path in discrete fixed-point problem. See Appendix~\ref{sec:BR} for details.

\section{Discussion}

A natural direction for future research is analyzing general graphs.  In the fully adaptive setting, prior work has successfully established  bounds for local search on general graphs by relating query complexity to structural properties such as expansion, separation number (which is equivalent to the treewidth), and vertex congestion.
Determining the relationship between these graph parameters and query complexity in rounds remains an interesting open question.

\section*{Acknowledgements} We thank the anonymous COLT and MSL reviewers for very helpful feedback and suggestions. In particular, the question regarding providing bounds for local search in rounds on  general graphs and the associated technical challenges was posed by one of the MSL reviewers. 

Simina Br\^anzei was supported in part by US National Science Foundation CAREER grant CCF-2238372.

\bibliographystyle{emss}
\bibliography{ref}

\begin{thebibliography}{GGK03b}

\bibitem[Aar06]{aaronson2006lower}
Scott Aaronson.
\newblock Lower bounds for local search by quantum arguments.
\newblock {\em SIAM Journal on Computing}, 35(4):804--824, 2006.

\bibitem[AAV86]{alon1986tight}
Noga Alon, Yossi Azar, and Uzi Vishkin.
\newblock Tight complexity bounds for parallel comparison sorting.
\newblock In {\em 27th Annual Symposium on Foundations of Computer Science
  (sfcs 1986)}, pages 502--510. IEEE, 1986.

\bibitem[AK93]{AK93}
Ingo Alth{\"{o}}fer and Klaus{-}Uwe Koschnick.
\newblock On the deterministic complexity of searching local maxima.
\newblock {\em Discret. Appl. Math.}, 43(2):111--113, 1993.

\bibitem[Akl14]{Akl_book}
Selim Akl.
\newblock {\em Parallel Sorting Algorithms}.
\newblock Academic Press, 2014.

\bibitem[Ald83]{aldous1983minimization}
David Aldous.
\newblock Minimization algorithms and random walk on the $ d $-cube.
\newblock {\em The Annals of Probability}, 11(2):403--413, 1983.

\bibitem[BM20]{DBLP:conf/colt/BubeckM20}
S{\'{e}}bastien Bubeck and Dan Mikulincer.
\newblock How to trap a gradient flow.
\newblock In {\em Conference on Learning Theory, {COLT} 2020, 9-12 July 2020,
  Virtual Event [Graz, Austria]}, volume 125 of {\em Proceedings of Machine
  Learning Research}, pages 940--960. {PMLR}, 2020.

\bibitem[BMP19]{BravermanMP19}
Mark Braverman, Jieming Mao, and Yuval Peres.
\newblock Sorted top-k in rounds.
\newblock In Alina Beygelzimer and Daniel Hsu, editors, {\em Conference on
  Learning Theory, {COLT} 2019, 25-28 June 2019, Phoenix, AZ, {USA}}, volume~99
  of {\em Proceedings of Machine Learning Research}, pages 342--382. {PMLR},
  2019.

\bibitem[BMW16]{BravermanMW16}
Mark Braverman, Jieming Mao, and S.~Matthew Weinberg.
\newblock Parallel algorithms for select and partition with noisy comparisons.
\newblock In Daniel Wichs and Yishay Mansour, editors, {\em Proceedings of the
  48th Annual {ACM} {SIGACT} Symposium on Theory of Computing, {STOC} 2016,
  Cambridge, MA, USA, June 18-21, 2016}, pages 851--862. {ACM}, 2016.

\bibitem[Bol88]{bollobas1988sorting}
B{\'e}la Bollob{\'a}s.
\newblock Sorting in rounds.
\newblock {\em Discrete Mathematics}, 72(1-3):21--28, 1988.

\bibitem[Bro10]{Brouwer10}
L.E.J. Brouwer.
\newblock Continuous one-one transformations of sur-faces in themselves.
\newblock {\em Koninklijke Neder-landse Akademie van Weteschappen, Proceedings
  Series B Physical Sciences}, 13:967--977, 1910.

\bibitem[Bro11]{Brouwer11}
L.E.J. Brouwer.
\newblock Continuous one-one transformations of surfaces in themselves.
\newblock {\em Koninklijke Neder-landse Akademie van Weteschappen Proceedings
  Series B Physical Sciences}, 14:300--310, 1911.

\bibitem[BS18]{BalkanskiS18}
Eric Balkanski and Yaron Singer.
\newblock The adaptive complexity of maximizing a submodular function.
\newblock In Ilias Diakonikolas, David Kempe, and Monika Henzinger, editors,
  {\em Proceedings of the 50th Annual {ACM} {SIGACT} Symposium on Theory of
  Computing, Los Angeles, CA, USA, June 25-29, 2018}, pages 1138--1151. {ACM},
  2018.

\bibitem[CD05]{chen2005algorithms}
Xi~Chen and Xiaotie Deng.
\newblock On algorithms for discrete and approximate brouwer fixed points.
\newblock In {\em Proceedings of the thirty-seventh annual ACM symposium on
  Theory of computing}, pages 323--330, 2005.

\bibitem[CD09]{CD09}
Xi~Chen and Xiaotie Deng.
\newblock On the complexity of 2d discrete fixed point problem.
\newblock {\em Theor. Comput. Sci.}, 410(44):4448--4456, 2009.

\bibitem[CDT09]{CDT09}
Xi~Chen, Xiaotie Deng, and Shang{-}Hua Teng.
\newblock Settling the complexity of computing two-player nash equilibria.
\newblock {\em J. {ACM}}, 56(3):14:1--14:57, 2009.

\bibitem[CMM20]{Cohen-AddadMM20}
Vincent Cohen{-}Addad, Frederik Mallmann{-}Trenn, and Claire Mathieu.
\newblock Instance-optimality in the noisy value-and comparison-model.
\newblock In Shuchi Chawla, editor, {\em Proceedings of the 2020 {ACM-SIAM}
  Symposium on Discrete Algorithms, {SODA} 2020, Salt Lake City, UT, USA,
  January 5-8, 2020}, pages 2124--2143. {SIAM}, 2020.

\bibitem[CPY17]{CPY17}
Xi~Chen, Dimitris Paparas, and Mihalis Yannakakis.
\newblock The complexity of non-monotone markets.
\newblock {\em J. {ACM}}, 64(3):20:1--20:56, 2017.

\bibitem[CT07]{chen2007paths}
Xi~Chen and Shang-Hua Teng.
\newblock Paths beyond local search: A tight bound for randomized fixed-point
  computation.
\newblock In {\em 48th Annual IEEE Symposium on Foundations of Computer Science
  (FOCS'07)}, pages 124--134. IEEE, 2007.

\bibitem[DGP09]{DGP09}
Constantinos Daskalakis, Paul~W. Goldberg, and Christos~H. Papadimitriou.
\newblock The complexity of computing a nash equilibrium.
\newblock {\em {SIAM} J. Comput.}, 39(1):195--259, 2009.

\bibitem[DP11]{DP11}
Constantinos Daskalakis and Christos~H. Papadimitriou.
\newblock Continuous local search.
\newblock In Dana Randall, editor, {\em Proceedings of the Twenty-Second Annual
  {ACM-SIAM} Symposium on Discrete Algorithms, {SODA} 2011, San Francisco,
  California, USA, January 23-25, 2011}, pages 790--804. {SIAM}, 2011.

\bibitem[DSZ21]{DSZ20}
Constantinos Daskalakis, Stratis Skoulakis, and Manolis Zampetakis.
\newblock The complexity of constrained min-max optimization.
\newblock In {\em Proceedings of the ACM Symposium on Theory of Computing}.
  Association for Computing Machinery, 2021.

\bibitem[FGHS21]{FGHS20}
John Fearnley, Paul~W. Goldberg, Alexandros Hollender, and Rahul Savani.
\newblock The complexity of gradient descent: {CLS} = {PPAD} {\(\cap\)} {PLS}.
\newblock In {\em Proceedings of the ACM Symposium on Theory of Computing}.
  Association for Computing Machinery, 2021.

\bibitem[FPT04]{FPT04}
Alex Fabrikant, Christos~H. Papadimitriou, and Kunal Talwar.
\newblock The complexity of pure nash equilibria.
\newblock In L{\'{a}}szl{\'{o}} Babai, editor, {\em Proceedings of the 36th
  Annual {ACM} Symposium on Theory of Computing, Chicago, IL, USA, June 13-16,
  2004}, pages 604--612. {ACM}, 2004.

\bibitem[GGK03a]{GGK03}
William Gasarch, Evan Golub, and Clyde Kruskal.
\newblock Constant time parallel sorting: an empirical view.
\newblock {\em Journal of Computer and System Sciences}, 67(1):63--91, 2003.

\bibitem[GGK03b]{GasarchGK03}
William~I. Gasarch, Evan Golub, and Clyde~P. Kruskal.
\newblock Constant time parallel sorting: an empirical view.
\newblock {\em J. Comput. Syst. Sci.}, 67(1):63--91, 2003.

\bibitem[HPV89]{hirsch1989exponential}
Michael~D Hirsch, Christos~H Papadimitriou, and Stephen~A Vavasis.
\newblock Exponential lower bounds for finding brouwer fix points.
\newblock {\em Journal of Complexity}, 5(4):379--416, 1989.

\bibitem[HY17]{hubavcek2017hardness}
Pavel Hub{\'a}{\v{c}}ek and Eylon Yogev.
\newblock Hardness of continuous local search: Query complexity and
  cryptographic lower bounds.
\newblock In {\em Proceedings of the Twenty-Eighth Annual ACM-SIAM Symposium on
  Discrete Algorithms}, pages 1352--1371. SIAM, 2017.

\bibitem[Iim03]{iimura2003discrete}
Takuya Iimura.
\newblock A discrete fixed point theorem and its applications.
\newblock {\em Journal of Mathematical Economics}, 39(7):725--742, 2003.

\bibitem[JPY88]{JohnsonPY88}
David~S. Johnson, Christos~H. Papadimitriou, and Mihalis Yannakakis.
\newblock How easy is local search?
\newblock {\em J. Comput. Syst. Sci.}, 37(1):79--100, 1988.

\bibitem[LL10]{lawler2010random}
Gregory~F Lawler and Vlada Limic.
\newblock {\em Random walk: a modern introduction}, volume 123.
\newblock Cambridge University Press, 2010.

\bibitem[LT93]{DBLP:journals/dam/LlewellynT93}
Donna~Crystal Llewellyn and Craig~A. Tovey.
\newblock Dividing and conquering the square.
\newblock {\em Discret. Appl. Math.}, 43(2):131--153, 1993.

\bibitem[LTT93]{llewellyn1993local}
Donna~Crystel Llewellyn, Craig Tovey, and Michael Trick.
\newblock Local optimization on graphs: Discrete applied mathematics 23 (1989)
  157--178.
\newblock {\em Discrete Applied Mathematics}, 46(1):93--94, 1993.

\bibitem[Nem94]{NEMIROVSKI}
A.~Nemirovski.
\newblock On parallel complexity of nonsmooth convex optimization.
\newblock {\em Journal of Complexity}, 10(4):451 -- 463, 1994.

\bibitem[Pap94]{Papadimitriou_1994}
Christos~H. Papadimitriou.
\newblock On the complexity of the parity argument and other inefficient proofs
  of existence.
\newblock {\em J. Comput. Syst. Sci.}, 48(3):498--532, 1994.

\bibitem[Pip87]{pippenger1987sorting}
Nicholas Pippenger.
\newblock Sorting and selecting in rounds.
\newblock {\em SIAM Journal on Computing}, 16(6):1032--1038, 1987.

\bibitem[Set12]{Set12}
Burr Settles.
\newblock Active learning.
\newblock {\em Synthesis Lectures on Artificial Intelligence and Machine
  Learning}, 6(1):1--114, 2012.

\bibitem[SY91]{SchafferY91}
Alejandro~A. Sch{\"{a}}ffer and Mihalis Yannakakis.
\newblock Simple local search problems that are hard to solve.
\newblock {\em {SIAM} J. Comput.}, 20(1):56--87, 1991.

\bibitem[SY09]{sun2009quantum}
Xiaoming Sun and Andrew Chi-Chih Yao.
\newblock On the quantum query complexity of local search in two and three
  dimensions.
\newblock {\em Algorithmica}, 55(3):576--600, 2009.

\bibitem[Tov81]{tovey81}
Craig Tovey.
\newblock Polynomial local improvement algorithms in combinatorial
  optimization, 1981.
\newblock Ph.D. thesis, Stanford University.

\bibitem[Val75]{Val75}
Leslie~G. Valiant.
\newblock Parallelism in comparison problems.
\newblock {\em SIAM Journal on Computing}, 4(3):348--355, 1975.

\bibitem[Vav93]{vavasis1993black}
Stephen~A Vavasis.
\newblock Black-box complexity of local minimization.
\newblock {\em SIAM Journal on Optimization}, 3(1):60--80, 1993.

\bibitem[VY11]{VY11}
Vijay~V. Vazirani and Mihalis Yannakakis.
\newblock Market equilibrium under separable, piecewise-linear, concave
  utilities.
\newblock {\em J. {ACM}}, 58(3):10:1--10:25, 2011.

\bibitem[WZ99]{wigderson1999expanders}
Avi Wigderson and David Zuckerman.
\newblock Expanders that beat the eigenvalue bound: Explicit construction and
  applications.
\newblock {\em Combinatorica}, 19(1):125--138, 1999.

\bibitem[Yao77]{Yao77}
Andrew Yao.
\newblock Probabilistic computations: Toward a unified measure of complexity.
\newblock In {\em Proceedings of the 18th IEEE Symposium on Foundations of
  Computer Science (FOCS)}, pages 222--227, 1977.

\bibitem[Zha09]{zhang2009tight}
Shengyu Zhang.
\newblock Tight bounds for randomized and quantum local search.
\newblock {\em SIAM Journal on Computing}, 39(3):948--977, 2009.

\end{thebibliography}


\begin{thebibliography}{10}
\providecommand{\url}[1]{\texttt{#1}}
\providecommand{\urlprefix}{URL }
\providecommand{\eprint}[2][]{\url{#2}}

\bibitem{aaronson2006lower}
S.~Aaronson, Lower bounds for local search by quantum arguments. \emph{SIAM
  Journal on Computing} \textbf{35} (2006), no.~4, 804--824

\bibitem{Akl_book}
S.~Akl, \emph{Parallel sorting algorithms}. Academic Press, 2014

\bibitem{aldous1983minimization}
D.~Aldous, Minimization algorithms and random walk on the $ d $-cube. \emph{The
  Annals of Probability} \textbf{11} (1983), no.~2, 403--413

\bibitem{alon1986tight}
N.~Alon, Y.~Azar, and U.~Vishkin, Tight complexity bounds for parallel
  comparison sorting. In \emph{27th annual symposium on foundations of computer
  science (sfcs 1986)}, pp. 502--510, IEEE, 1986

\bibitem{AK93}
I.~Alth{\"{o}}fer and K.~Koschnick, On the deterministic complexity of
  searching local maxima. \emph{Discret. Appl. Math.} \textbf{43} (1993),
  no.~2, 111--113

\bibitem{ACB17}
M.~Arjovsky, S.~Chintala, and L.~Bottou, Wasserstein generative adversarial
  networks. In \emph{Proceedings of the 34th international conference on
  machine learning - volume 70}, p. 214–223, JMLR.org, 2017

\bibitem{BSN19}
Y.~Babichenko, S.~Dobzinski, and N.~Nisan, The communication complexity of
  local search. In \emph{Proceedings of the 51st annual acm sigact symposium on
  theory of computing (stoc)}, p. 650–661, Association for Computing
  Machinery, 2019

\bibitem{BalkanskiRS19}
E.~Balkanski, A.~Rubinstein, and Y.~Singer, An exponential speedup in parallel
  running time for submodular maximization without loss in approximation. In
  \emph{Proceedings of the thirtieth annual {ACM-SIAM} symposium on discrete
  algorithms (soda) 2019}, pp. 283--302, {SIAM}, 2019

\bibitem{BalkanskiS18}
E.~Balkanski and Y.~Singer, The adaptive complexity of maximizing a submodular
  function. In \emph{Proceedings of the 50th annual {ACM} {SIGACT} symposium on
  theory of computing, los angeles, ca, usa, june 25-29, 2018}, edited by
  I.~Diakonikolas, D.~Kempe, and M.~Henzinger, pp. 1138--1151, {ACM}, 2018

\bibitem{BalkanskiS20}
E.~Balkanski and Y.~Singer, A lower bound for parallel submodular minimization.
  In \emph{Proccedings of the 52nd annual {ACM} {SIGACT} symposium on theory of
  computing (stoc)}, pp. 130--139, {ACM}, 2020

\bibitem{bollobas1988sorting}
B.~Bollob{\'a}s, Sorting in rounds. \emph{Discrete Mathematics} \textbf{72}
  (1988), no. 1-3, 21--28

\bibitem{BCR24}
S.~Br{\^{a}}nzei, D.~Choo, and N.~J. Recker, The sharp power law of local
  search on expanders. In \emph{Soda}, 2024

\bibitem{BravermanMP19}
M.~Braverman, J.~Mao, and Y.~Peres, Sorted top-k in rounds. In \emph{Conference
  on learning theory, {COLT} 2019, 25-28 june 2019, phoenix, az, {USA}}, edited
  by A.~Beygelzimer and D.~Hsu, pp. 342--382, Proceedings of Machine Learning
  Research 99, {PMLR}, 2019

\bibitem{BravermanMW16}
M.~Braverman, J.~Mao, and S.~M. Weinberg, Parallel algorithms for select and
  partition with noisy comparisons. In \emph{Proceedings of the 48th annual
  {ACM} {SIGACT} symposium on theory of computing, {STOC} 2016, cambridge, ma,
  usa, june 18-21, 2016}, edited by D.~Wichs and Y.~Mansour, pp. 851--862,
  {ACM}, 2016

\bibitem{Brouwer11}
L.~Brouwer, Continuous one-one transformations of surfaces in themselves.
  \emph{Koninklijke Neder-landse Akademie van Weteschappen Proceedings Series B
  Physical Sciences} \textbf{14} (1911), 300--310

\bibitem{BubeckJLLS19}
S.~Bubeck, Q.~Jiang, Y.~T. Lee, Y.~Li, and A.~Sidford, Complexity of highly
  parallel non-smooth convex optimization. In \emph{Advances in neural
  information processing systems 32: Annual conference on neural information
  processing systems 2019, neurips 2019, december 8-14, 2019, vancouver, bc,
  canada}, edited by H.~M. Wallach, H.~Larochelle, A.~Beygelzimer,
  F.~d'Alch{\'{e}}{-}Buc, E.~B. Fox, and R.~Garnett, pp. 13900--13909, 2019

\bibitem{DBLP:conf/colt/BubeckM20}
S.~Bubeck and D.~Mikulincer, How to trap a gradient flow. In \emph{Conference
  on learning theory, {COLT} 2020, 9-12 july 2020, virtual event [graz,
  austria]}, pp. 940--960, Proceedings of Machine Learning Research 125,
  {PMLR}, 2020

\bibitem{chen2005algorithms}
X.~Chen and X.~Deng, On algorithms for discrete and approximate brouwer fixed
  points. In \emph{Proceedings of the thirty-seventh annual acm symposium on
  theory of computing}, pp. 323--330, 2005

\bibitem{CD09}
X.~Chen and X.~Deng, On the complexity of 2d discrete fixed point problem.
  \emph{Theor. Comput. Sci.} \textbf{410} (2009), no.~44, 4448--4456

\bibitem{CDT09}
X.~Chen, X.~Deng, and S.~Teng, Settling the complexity of computing two-player
  nash equilibria. \emph{J. {ACM}} \textbf{56} (2009), no.~3, 14:1--14:57

\bibitem{CPY17}
X.~Chen, D.~Paparas, and M.~Yannakakis, The complexity of non-monotone markets.
  \emph{J. {ACM}} \textbf{64} (2017), no.~3, 20:1--20:56

\bibitem{chen2007paths}
X.~Chen and S.-H. Teng, Paths beyond local search: A tight bound for randomized
  fixed-point computation. In \emph{48th annual ieee symposium on foundations
  of computer science (focs'07)}, pp. 124--134, IEEE, 2007

\bibitem{Cohen-AddadMM20}
V.~Cohen{-}Addad, F.~Mallmann{-}Trenn, and C.~Mathieu, Instance-optimality in
  the noisy value-and comparison-model. In \emph{Proceedings of the 2020
  {ACM-SIAM} symposium on discrete algorithms, {SODA} 2020, salt lake city, ut,
  usa, january 5-8, 2020}, edited by S.~Chawla, pp. 2124--2143, {SIAM}, 2020

\bibitem{DGP09}
C.~Daskalakis, P.~W. Goldberg, and C.~H. Papadimitriou, The complexity of
  computing a nash equilibrium. \emph{{SIAM} J. Comput.} \textbf{39} (2009),
  no.~1, 195--259

\bibitem{DP11}
C.~Daskalakis and C.~H. Papadimitriou, Continuous local search. In
  \emph{Proceedings of the twenty-second annual {ACM-SIAM} symposium on
  discrete algorithms, {SODA} 2011, san francisco, california, usa, january
  23-25, 2011}, edited by D.~Randall, pp. 790--804, {SIAM}, 2011

\bibitem{DSZ20}
C.~Daskalakis, S.~Skoulakis, and M.~Zampetakis, The complexity of constrained
  min-max optimization. In \emph{Proceedings of the acm symposium on theory of
  computing}, Association for Computing Machinery, 2021

\bibitem{dinh2010quantum}
H.~Dinh and A.~Russell, Quantum and randomized lower bounds for local search on
  vertex-transitive graphs. \emph{Quantum Information \& Computation}
  \textbf{10} (2010), no.~7, 636--652

\bibitem{EN19}
A.~Ene and H.~L. Nguyen, Submodular maximization with nearly-optimal
  approximation and adaptivity in nearly-linear time. In \emph{Proceedings of
  the thirtieth annual acm-siam symposium on discrete algorithms}, p.
  274–282, SODA '19, Society for Industrial and Applied Mathematics, USA,
  2019

\bibitem{FGHS20}
J.~Fearnley, P.~W. Goldberg, A.~Hollender, and R.~Savani, The complexity of
  gradient descent: {CLS} = {PPAD} {\(\cap\)} {PLS}. In \emph{Proceedings of
  the acm symposium on theory of computing}, Association for Computing
  Machinery, 2021

\bibitem{GasarchGK03}
W.~I. Gasarch, E.~Golub, and C.~P. Kruskal, Constant time parallel sorting: an
  empirical view. \emph{J. Comput. Syst. Sci.} \textbf{67} (2003), no.~1,
  63--91

\bibitem{NIPS2014_GANs}
I.~Goodfellow, J.~Pouget-Abadie, M.~Mirza, B.~Xu, D.~Warde-Farley, S.~Ozair,
  A.~Courville, and Y.~Bengio, Generative adversarial nets. In \emph{Advances
  in neural information processing systems}, 27, Curran Associates, Inc., 2014

\bibitem{goos2022further}
M.~Göös, A.~Hollender, S.~Jain, G.~Maystre, W.~Pires, R.~Robere, and R.~Tao,
  Further collapses in tfnp. 2022,
  \urlprefix\url{https://arxiv.org/abs/2202.07761}

\bibitem{hirsch1989exponential}
M.~D. Hirsch, C.~H. Papadimitriou, and S.~A. Vavasis, Exponential lower bounds
  for finding brouwer fix points. \emph{Journal of Complexity} \textbf{5}
  (1989), no.~4, 379--416

\bibitem{hubavcek2017hardness}
P.~Hub{\'a}{\v{c}}ek and E.~Yogev, Hardness of continuous local search: Query
  complexity and cryptographic lower bounds. In \emph{Proceedings of the
  twenty-eighth annual acm-siam symposium on discrete algorithms}, pp.
  1352--1371, SIAM, 2017

\bibitem{iimura2003discrete}
T.~Iimura, A discrete fixed point theorem and its applications. \emph{Journal
  of Mathematical Economics} \textbf{39} (2003), no.~7, 725--742

\bibitem{JohnsonPY88}
D.~S. Johnson, C.~H. Papadimitriou, and M.~Yannakakis, How easy is local
  search? \emph{J. Comput. Syst. Sci.} \textbf{37} (1988), no.~1, 79--100

\bibitem{lawler2010random}
G.~F. Lawler and V.~Limic, \emph{Random walk: a modern introduction}. 123,
  Cambridge University Press, 2010

\bibitem{llewellyn1989local}
D.~C. Llewellyn, C.~Tovey, and M.~Trick, Local optimization on graphs.
  \emph{Discrete Applied Mathematics} \textbf{23} (1989), no.~2, 157--178

\bibitem{DBLP:journals/dam/LlewellynT93}
D.~C. Llewellyn and C.~A. Tovey, Dividing and conquering the square.
  \emph{Discret. Appl. Math.} \textbf{43} (1993), no.~2, 131--153

\bibitem{MS96}
D.~Monderer and L.~Shapley, Potential games. \emph{Games and Economic Behavior}
  \textbf{14} (1996), no.~1, 124--143

\bibitem{Nash48}
J.~F. Nash, Equilibrium points in n-person games. \emph{Proceedings of the
  National Academy of Sciences} \textbf{36} (1950), no.~1, 48--49

\bibitem{NEMIROVSKI}
A.~Nemirovski, On parallel complexity of nonsmooth convex optimization.
  \emph{Journal of Complexity} \textbf{10} (1994), no.~4, 451 -- 463

\bibitem{Papadimitriou_1994}
C.~H. Papadimitriou, On the complexity of the parity argument and other
  inefficient proofs of existence. \emph{J. Comput. Syst. Sci.} \textbf{48}
  (1994), no.~3, 498--532

\bibitem{pippenger1987sorting}
N.~Pippenger, Sorting and selecting in rounds. \emph{SIAM Journal on Computing}
  \textbf{16} (1987), no.~6, 1032--1038

\bibitem{Rosenthal1973}
R.~W. Rosenthal, A class of games possessing pure-strategy nash equilibria.
  \emph{International Journal of Game Theory} \textbf{2} (1973), 65--67

\bibitem{santha2004quantum}
M.~Santha and M.~Szegedy, Quantum and classical query complexities of local
  search are polynomially related. In \emph{Proceedings of the thirty-sixth
  annual acm symposium on theory of computing}, pp. 494--501, 2004

\bibitem{Set12}
B.~Settles, Active learning. \emph{Synthesis Lectures on Artificial
  Intelligence and Machine Learning} \textbf{6} (2012), no.~1, 1--114

\bibitem{sun2009quantum}
X.~Sun and A.~C.-C. Yao, On the quantum query complexity of local search in two
  and three dimensions. \emph{Algorithmica} \textbf{55} (2009), no.~3, 576--600

\bibitem{tovey81}
C.~Tovey, Polynomial local improvement algorithms in combinatorial
  optimization. 1981, ph.D. thesis, Stanford University

\bibitem{Val75}
L.~G. Valiant, Parallelism in comparison problems. \emph{SIAM Journal on
  Computing} \textbf{4} (1975), no.~3, 348--355

\bibitem{vavasis1993black}
S.~A. Vavasis, Black-box complexity of local minimization. \emph{SIAM Journal
  on Optimization} \textbf{3} (1993), no.~1, 60--80

\bibitem{VY11}
V.~V. Vazirani and M.~Yannakakis, Market equilibrium under separable,
  piecewise-linear, concave utilities. \emph{J. {ACM}} \textbf{58} (2011),
  no.~3, 10:1--10:25

\bibitem{Verhoeven06}
Y.~F. Verhoeven, Enhanced algorithms for local search. \emph{Information
  Processing Letters} \textbf{97} (2006), no.~5, 171--176

\bibitem{wigderson1999expanders}
A.~Wigderson and D.~Zuckerman, Expanders that beat the eigenvalue bound:
  Explicit construction and applications. \emph{Combinatorica} \textbf{19}
  (1999), no.~1, 125--138

\bibitem{Yao77}
A.~Yao, Probabilistic computations: Toward a unified measure of complexity. In
  \emph{Proceedings of the 18th ieee symposium on foundations of computer
  science (focs)}, pp. 222--227, 1977

\bibitem{zhang2010mcp}
C.-H. Zhang, {Nearly unbiased variable selection under minimax concave
  penalty}. \emph{Annals of Statistics} \textbf{38} (2010), no.~2, 894--942

\bibitem{zhang2009tight}
S.~Zhang, Tight bounds for randomized and quantum local search. \emph{SIAM
  Journal on Computing} \textbf{39} (2009), no.~3, 948--977

\end{thebibliography}

\newpage

\appendix

\section{Algorithm for Local Search in Constant Rounds}\label{sec:LS.ConstAlg}

In this section we give an algorithm for local search on the $d$-dimensional grid in constant rounds, which will imply the upper bound of Theorem~\ref{thm:1}.

\paragraph{Notation}
A $d$-dimensional \emph{cube} is a Cartesian product of $d$ connected (integer) intervals. We use \emph{cube} to indicate $d$-dimensional \emph{cube} for brevity, unless otherwise specified. The boundary of cube $C$ is defined as all the points $\Vec{x} \in C$ with fewer than $2d$ neighbors in $C$.

\begin{proof}[Proof of Upper Bound of Theorem~\ref{thm:1}]
	Given the $d$-dimensional grid $[n]^d$, consider a sequence of cubes contained in each other: $C_0 \supset C_1 \supset C_2 \supset \ldots C_{k-1}$, where $C_0 = [n]^d$ is the whole grid. 
	
	For each $0 \leq i < k$, set $\ell_i = n^{\frac{d^k-d^i}{d^k-1}}$ as the side length of cube $C_i$. The values of $\ell_i$ are chosen for balancing the number of queries in each round, which will be proved later.
	Note $\ell_i$ is an integer divisor of $n$. Consider Algorithm 1, which is illustrated on the following example. 
	
	\begin{figure}[h!]
		\centering
		\includegraphics[scale=0.45]{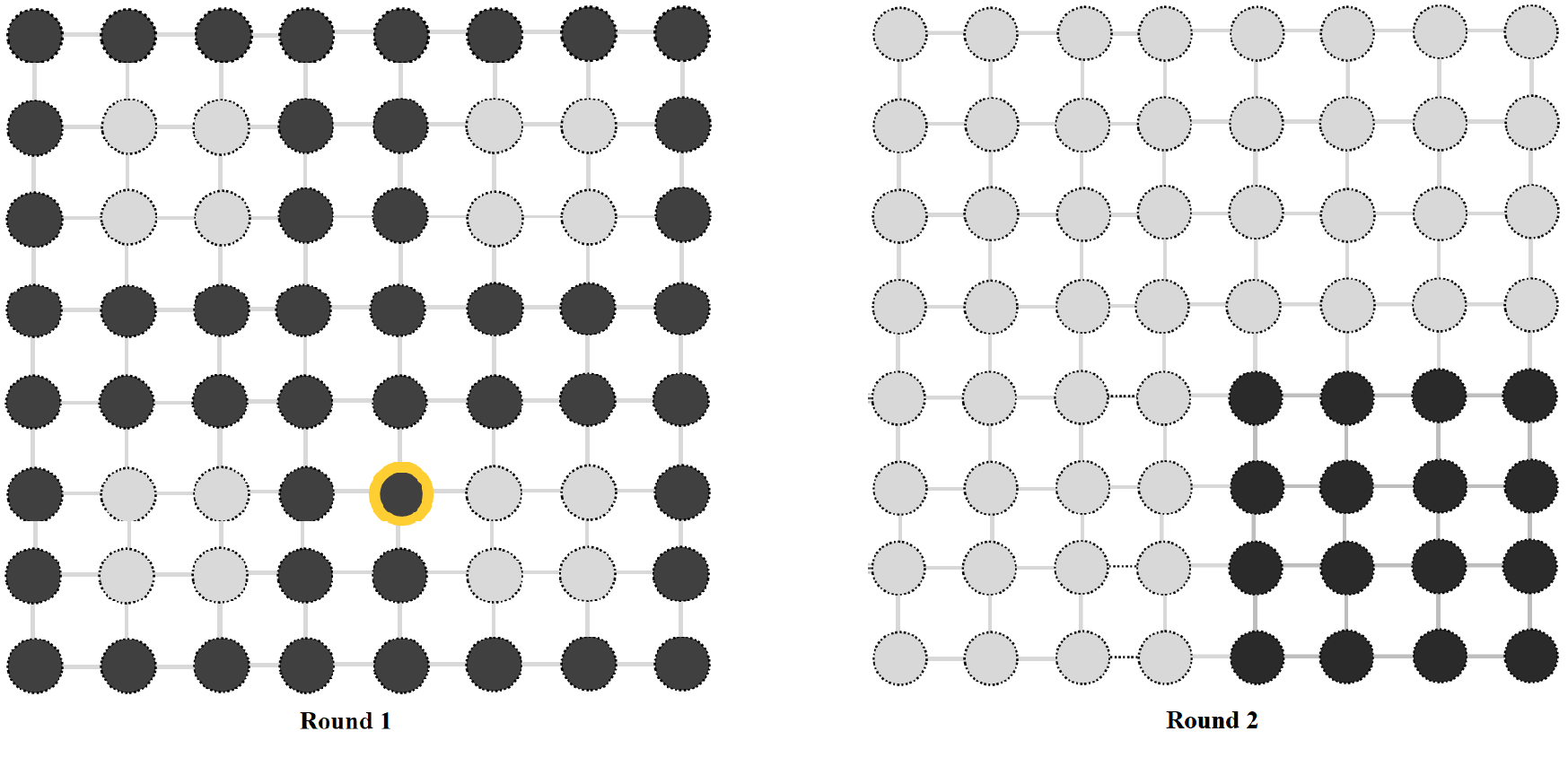}
		\caption{Two dimensional grid of size $8\times 8$. Suppose there are two rounds. In round $1$ the algorithm queries all the black points and selects the minimum among all these points (illustrated with a yellow boundary). In round 2, it queries the entire sub-square in which the minimum from the first round was found.}
		\label{fig:HDD}
	\end{figure}
	%\vspace{6mm}
	
	\paragraph{Algorithm 1: Local search in constant rounds}
	\begin{enumerate} 
		\item Initialize the current cube to $C_0$. 
		\item In each round $i \in \{1, \ldots, k-1\}$:
		\begin{itemize}
			\item $\; \; \;$ Divide the current cube $C_{i-1}$ into a set of mutually exclusive sub-cubes $C_i^1, \ldots, C_{i}^{n_i}$ of side length $\ell_i$ that cover $C_{i-1}$. Note that $n_i = (\ell_{i-1}/\ell_i)^d$ by definition.
			\item  $\; \; \;$ Query all the points on the boundary of sub-cubes $C_i^1, \ldots, C_{i}^{n_i}$. Let $\Vec{x}^*_i$ be the point with minimal value among them.
			\item  $\; \; \;$ Set $C_i = C_i^j$, where $C_i^j$ is the sub-cube that $\Vec{x}^*_i$ belongs to. 
		\end{itemize}
		\item In round $k$, query all the points in the current cube $C_{k-1}$ and find the solution point.
	\end{enumerate} 
	To argue that the algorithm finds a local minimum, note that in each round $i \leq k-1$, the steepest descent starting from $\Vec{x}^*_i$ will never leave the sub-cube $C_i$, since if it did it would have to exit $C_i$ through a point of even smaller value than $\Vec{x}_i^*$, which contradicts the definition of $\Vec{x}_i^*$. Thus there must exist a local optimum within $C_i$.
	
	Now we calculate the number of queries in each round. In round $i \in \{1, \ldots, k-1\}$, the number of points on the boundary of all sub-cubes $S_1, \ldots, S_{n_i}$ is
    \begin{align*}
        n_i \cdot (2d \cdot {\ell_i}^{d-1}) &= \left({\ell_{i-1}}/{\ell_i}\right)^d\cdot 2d \cdot {\ell_i}^{d-1} \\
        &= 2d \cdot \ell_{i-1}^{d}\cdot (1 / \ell_{i}) \\
        &= 2d\cdot n^{\frac{d\cdot(d^{k}-d^{i-1}) - (d^k - d^i)}{d^k-1}} \\
        &= 2d \cdot n^{\frac{d^{k+1} - d^k}{d^k - 1}}.
    \end{align*}
	The number of queries in round $k$ is 
 \[(\ell_{k-1})^{d} = n^{\frac{d^{k+1} - d^k}{d^k - 1}}\,.\] 
	Since $k$ and $d$ are constants, the algorithm makes $O\left( n^{\frac{d^{k+1} - d^k}{d^k - 1}} \right)$ queries in total as required.
\end{proof}

\bigskip 
\section{Algorithm for Local Search in Polynomial Rounds} \label{sec:LS.algo_poly_rounds}

In this section we prove the upper bound of Theorem~\ref{thm:2}. Thus we show the following upper bounds on the query complexity of local search in polynomial rounds: $O\bigl(n^{(d-1) - \frac{d-2}{d}\alpha}\bigr)$ when $d \geq 5$,  $O\bigl(n^{3 - \frac{\alpha}{2}}\bigr)$ when $d = 4$, and 
		 $O\left(n^{2 - \frac{\alpha}{3}}\right)$  when $d = 3$. 
The algorithm for polynomial number of rounds was described in Section~\ref{sec:poly_rounds_main}. Here we give the pseudocode with precise definitions and the proof of correctness.

\paragraph{Notation.} 
Let the number of rounds be $k = n^{\alpha}$, where $\alpha \in (0, {d}/{2})$ is a constant. 
Given a point $\Vec{x}$ on the grid, let $\rank(\Vec{x})$ be the number of grid points with smaller value than point $\Vec{x}$. Let $C(\Vec{x},s) \coloneqq \{\Vec{y} \in [n]^d: \|\Vec{y}-\Vec{x}\|_{\infty} \leq s \}$ be the set of grid points in the $d$-dimensional cube of side length $2 \cdot s$, centered at point $\Vec{x}$. We say a point $\Vec{x}$ is the minimum in a set $S$ if $\Vec{x} \in S$ and $f(\Vec{x}) \leq f(\Vec{y})$ for all $\Vec{y} \in S$.

\medskip

There are multiple query requests from different procedures in one round. These queries will first be collected together and then submitted at once at the end of each round.

\paragraph{\textcolor{midnightblue}{Algorithm 2: Local search in polynomial number of rounds}.} \label{algo:HDR}
\textcolor{midnightblue}{\textbf{\emph{Input:}}} Size of the instance $n$, dimension $d$, round limit $k = n^{\alpha}$, value function $f$. These are global parameters accessible from any subroutine.
\textcolor{midnightblue}{\emph{\textbf{Output:}}} Local minimum $\Vec{x}^*$ in $[n]^d$.
\smallskip 
\begin{enumerate}[label=\textbf{\textcolor{midnightblue}{\arabic*.}}]
	\item Set $\beta \leftarrow (d-1) - \alpha \bigl(\frac{d-2}{d}\bigr)$ ; $h \leftarrow \bigl \lfloor \frac{1}{\alpha}+\frac{d-2}{d} \bigr \rfloor + 1$ ; $\widetilde{k} \leftarrow {k}/{h}$ %$h \leftarrow \Bigl \lfloor \frac{1+\frac{d-2}{d}\alpha}{\alpha} \Bigr \rfloor + 1$
	\item Query ${Q}_1 = 100 \cdot h \cdot n^{\beta}$ points chosen u.a.r. in round $1$ and set $\Vec{x}_1$ to the minimum of these %\footnote{That is, the point $\Vec{x}$ attains the minimum value of the function $f$ among these points.}
	\item Set $s \leftarrow ({1}/{h}) \cdot {n^{1+\alpha\left({d-2}\right)/{d}}}$
	\item {Return} \flsdescent{}$(s, h, \Vec{x}_1, 2)$ 
\end{enumerate}

\paragraph{\textcolor{midnightblue}{Procedure Fractal-like Steepest Descent (FLSD)}}
\textcolor{midnightblue}{\emph{\textbf{Input:}}} size $s$, depth $D$, grid point $\Vec{x}$, round $r$. 
\textcolor{midnightblue}{\emph{\textbf{Output:}}} point $\Vec{x}^* \in C(\Vec{x}, s)$ with $\rank(\Vec{x}^*) \leq \rank(\Vec{x}) - s$; if in the process of searching for such a point it finds a local minimum, then it outputs it and halts everything.
\smallskip 
\begin{enumerate}[label=\textbf{\textcolor{midnightblue}{\arabic*.}}]
	\item Set $\Vec{x}_0 \leftarrow \Vec{x} \; ; \; step \leftarrow {s}/{\widetilde{k}} \; \; $  \emph{\textcolor{mygray}{// executed in round $r$}}
	\item If $D = 1$ then: \emph{\textcolor{mygray}{// make $s$ steps of steepest descent, since $s$ is small enough when $D=1$}} 
	\begin{enumerate}% [label=\alph*.] % label=\color{blue}]
		\item For $i=0$ to $s-1$: \emph{\textcolor{mygray}{// {executed in rounds $r+i$ to $r+i+1$}}}
		\begin{enumerate}
			\item Query all the neighbors of $\Vec{x}_i$; let $\Vec{x}_{i+1}$ be the minimum among them  % \emph{// {round $r+i+1$}}
			\item If $\Vec{x}_i = \Vec{x}_{i+1}$ then: \emph{\textcolor{mygray}{// {thus $\Vec{x}_i$ is a local min}}} %\label{line:test}
			\begin{description} 
				\item Output $\Vec{x}^* \leftarrow \Vec{x}_i$ and halt all running \flsdescent{} and \dacs{} calls $\;$  \label{line:halt2} % \emph{// {round $r+i+1$}}
			\end{description}
		\end{enumerate}
		\item Return $\Vec{x}^* \leftarrow \Vec{x}_{s} \; \; $  \emph{\textcolor{mygray}{// executed in round $r+s$}} \label{line:ret1}
	\end{enumerate}
	
	\item For $i=0$ to $\widetilde{k}-1$: \emph{\textcolor{mygray}{// divide the whole task into $\widetilde{k}$ pieces; executed in rounds $r+i$ to $r+i+1$}}
	\begin{enumerate}% [label=\alph*.] 
		\item $\Vec{y}_{i+1} \leftarrow$\flsdescent{}$(step, D-1, \Vec{x}_i, r+i)$ \emph{\textcolor{mygray}{// execute call in parallel with current procedure}}\label{line:callf}
		\item Query the \emph{boundary} of $C(\Vec{x}_i, step)$ to find the point $\Vec{x}_{i+1}$ with minimum value on it \label{line:query} \emph{\textcolor{mygray}{// making a ``giant step'' of size $\mbox{step}$ by cheap substitute $\Vec{x}_i$}} %\emph{// {round $r+i$ to $r+i+1$}}
	\end{enumerate}
	
	\item For $i=0$ to $\widetilde{k}-1$: \emph{\textcolor{mygray}{// check if each giant step does make giant progress by using the feedback from sub-procedures; executed in round $r+i+(D-1) \cdot \widetilde{k}$, after $y_{i+1}$ was received in Step~\ref{line:callf}}}  
	\begin{enumerate}% [label=\alph*.] 
		\item  If $f(\Vec{y}_{i+1}) < f(\Vec{x}_{i+1})$ then: \emph{\textcolor{mygray}{// a solution exists in $C(\Vec{x}_i, step)$, call \dacs{} to find it}} \label{line:cmp}
		\begin{enumerate}
			\item Set $\Vec{x}^* \leftarrow $ \dacs{}$(C(\Vec{x}_i, step))$ \emph{\textcolor{mygray}{// stop and wait for the result of \dacs{}}} \label{line:call}
			\item Output $\Vec{x}^*$ and halt all running \flsdescent{} and \dacs{} calls \label{line:halt}
		\end{enumerate}
	\end{enumerate}
	\item Return $\Vec{x}^* \leftarrow \Vec{x}_{\widetilde{k}}$ \emph{\textcolor{mygray}{// executed in round $r+D \cdot \widetilde{k}$}} \label{line:ret2} 
\end{enumerate}

%%%%%%%%%%%%%%%%%%DACS%%%%%%%%%%%%%%%%%
\paragraph{\textcolor{midnightblue}{Procedure Divide-and-Conquer Search (DACS)}}  
\textcolor{midnightblue}{\emph{\textbf{Input:}}} cube $C_0$. 
\textcolor{midnightblue}{\emph{\textbf{Output:}}} Local minimum $\Vec{x}^*$ in $C_0$.
\begin{enumerate}[label=\textbf{\textcolor{midnightblue}{\arabic*.}}]
	\item Set $\Vec{x}^* \leftarrow $ {\emph{Null}}
	\item For $i = 1$ to $ \lceil \log_2 n\rceil$:
	\begin{enumerate}% [label=\alph*.]
		\item If $C_{i-1}$ contains only one point $\widetilde{\Vec{x}}$ then:
		\begin{enumerate}
			\item Set $\Vec{x}^* = \widetilde{\Vec{x}}$ ; {break}
		\end{enumerate}
		\item Partition $C_{i-1}$ into $2^d$ disjoint sub-cubes $C_i^1, \ldots, C_i^{2^d}$, each with side length half that of $C_{i-1}$
		\item Query all the points on the boundary of each sub-cube $C_i^1, \ldots, C_i^{2^d}$.
		\item  Let $\Vec{x}^*_i$ be the point with minimum value among \emph{all points queried} in $C_{i-1}$, including queries made by Algorithm 2 and all \flsdescent{} calls \emph{\textcolor{mygray}{// break ties lexicographically}} 
		\item  Let $C_i \in \{ C_i^1, \ldots, C_i^{2^d}\}$ be the unique sub-cube with $\Vec{x}^*_i \in C_i$.
	\end{enumerate}
	\item  Return $\Vec{x}^*$
\end{enumerate}

\paragraph{Analysis}

We first establish that Algorithm 2 is correct.

\begin{lemma}\label{prop:a1}
	If the procedure \flsdescent{}$(s, D, \Vec{x}, r)$
	\footnote{All the procedures \flsdescent{} we considered in the following analysis are initiated during the execution of Algorithm 2. Thus Lemma~\ref{prop:a1}, Lemma~\ref{prop:a2} and Lemma~\ref{prop:a4} may not work for \flsdescent{} with arbitrary parameters.}
	does return at Step~\ref{line:ret1} or Step~\ref{line:ret2}, it will return within $\widetilde{k} \cdot D$ rounds after the start round $r$;
	otherwise procedure \flsdescent{}$(s, D, \Vec{x}, r)$ will halt within at most $\widetilde{k}\cdot D+ \lceil \log n \rceil$ rounds after the start round $r$.
\end{lemma}
\begin{proof}
	We proceed by induction on the depth $D$.
	The base case is when $D=1$. By the definition of $\widetilde{k}$ and $h$, we have 
	$$\left(1/{\;\widetilde{k}\;}^h\right) \cdot n^{1+\frac{d-2}{d}\alpha} < 1\,.$$ 
	Also notice that the parameter $s$ will be divided by $\widetilde{k}$ when $D$ decreases by one, so when $D=1$, the current size $s$ will be at most $\widetilde{k}$, i.e., the steepest descent will return in $\widetilde{k}$ rounds. Assume it holds for $1, \ldots, D-1$. 
	
	For any $D > 1$, all the queries made by the procedure itself need $\widetilde{k}$ rounds and each sub-procedure will take at most $\widetilde{k}\cdot (D-1)$ rounds by the induction hypothesis. Since all the procedures are independent of each other and could be executed in parallel, the total number of rounds needed for this procedure is $\widetilde{k} \cdot D$. The first part of the lemma thus follows by induction.
	
	Finally, recall that the divide-and-conquer procedure \dacs{} takes $\lceil \log n \rceil$ rounds, so the procedure will halt within $\widetilde{k} \cdot D+ \lceil \log n \rceil$ rounds.
\end{proof}

\begin{lemma}\label{prop:a2}
	If the procedure \flsdescent{}$(s, D, \Vec{x}, r)$ does return a point $\Vec{x}^*$ at Step~\ref{line:ret1} or Step~\ref{line:ret2}, then $\Vec{x}^*$ is in the cube $C(\Vec{x}, s)$ and satisfy the inequality $\rank(\Vec{x}^*) \leq \rank(\Vec{x}) - s$.
\end{lemma}
\begin{proof}
	We proceed by induction on the depth $D$. The base case is when $D=1$. Then we know that $s \leq \widetilde{k}$ by the same argument in Lemma~\ref{prop:a1}, thus $s$ steps of steepest descent will ensure that $\Vec{x}^* \in C(\Vec{x}, s)$ and $\rank(\Vec{x}^*) \leq \rank(\Vec{x}) - s$. Assume it holds for $1, \ldots, D-1$ and show for $D$.
	
	For any $D > 1$, by Step~\ref{line:cmp} we have $\rank(\Vec{x}_{i}) \leq \rank(\Vec{y}_{i})$ for any $1\leq i \leq \widetilde{k}$; 
	by the induction hypothesis, we have $\rank(\Vec{y}_{i}) \leq \rank(\Vec{x}_{i-1}) - step$ for any $1\leq i \leq \widetilde{k}$. Combining them we get $\rank(\Vec{x}_{i}) \leq \rank(\Vec{x}_{i-1}) - step$ for any $1\leq i \leq \widetilde{k}$. Thus 
	\begin{equation*}
		\rank(\Vec{x}^*) = \rank(\Vec{x}_{\widetilde{k}}) \leq \rank(\Vec{x}_0) - \widetilde{k} \cdot step = \rank(\Vec{x})-s.
	\end{equation*}
	Also notice that the $L_{\infty}$ distance from $\Vec{x}^*$ to $\Vec{x}$ is at most $\widetilde{k} \cdot step = s$, i.e., $\Vec{x}^* \in C(\Vec{x}, s)$.
	This concludes the proof of the lemma for any depth $D$.
\end{proof}

\begin{lemma}\label{prop:a5}
	The point $\Vec{x}^*$ returned by FLSD at Step~\ref{line:call} is a local minimum.
\end{lemma}
\begin{proof}
	We use notation $\mathcal{D}.\Vec{x}$ to denote the variable $\Vec{x}$ in the procedure \dacs{} and $\mathcal{F}.\Vec{x}$ to denote the variable $\Vec{x}$ in the procedure \flsdescent{} which calls the procedure \dacs{}.
	
	Recall $\mathcal{D}.\Vec{x}^*_i$ is the point with minimum value among \emph{all points queried} in the sub-cube $C_{i-1}$, including queries made by the entire algorithm (Algorithm 2) and all \flsdescent{} calls.
	By Step~\ref{line:cmp}, we have $f(\mathcal{F}.\Vec{y}_{i+1}) < f(\mathcal{F}.\Vec{x}_{i+1})$. Then for each $\mathcal{D}.\Vec{x}^*_i$ and index $i$, we have the following inequality by definition of $\mathcal{D}.x_{i}^*$: $$f(\mathcal{D}.\Vec{x}^*_i) \leq f(\mathcal{F}.\Vec{y}_{i+1}) < f(\mathcal{F}.\Vec{x}_{i+1})\,.$$ Therefore the steepest descent from $\mathcal{D}.\Vec{x}^*_i$ will never leave the cube $\mathcal{D}.C_i$, especially the cube $\mathcal{D}.C_0$. 
	
	As noted above, $\mathcal{D}.\Vec{x}^*_i$ is taken as the minimum of \emph{all} queried points in $\mathcal{D}.C_{i-1}$, not just for points on the boundary. This makes sure that the point $\mathcal{F}.\Vec{y}_{i+1}$ is taken into account, thus any point with smaller or equal value than $\mathcal{F}.\Vec{y}_{i+1}$ could not escape from $\mathcal{D}.C_0$ by steepest descent; otherwise, the point $\mathcal{F}.\Vec{x}_{i+1}$ itself may be the minimum during all queries on the boundaries, but it could possibly escape the cube $\mathcal{D}.C_0$  by steepest descent, and cause the procedure \dacs{} return a non-local optimum. 
	
	Finally, let $\widetilde{C}$ be the cube that consists only of the point $\mathcal{D}.\widetilde{\Vec{x}}$ in the \dacs{} procedure. 
    By previous argument, the steepest descent from $\mathcal{D}.\Vec{x}^* = \mathcal{D}.\widetilde{\Vec{x}}$ doesn't leave the cube $\widetilde{C}$, which means that $\mathcal{D}.\Vec{x}^*$ is a local optimum.
\end{proof}

\begin{lemma}\label{prop:a3}
	Algorithm 2 outputs the correct answer with probability at least ${9}/{10}$.
\end{lemma}
\begin{proof}
	The point $\Vec{x}^*$ output at Step~\ref{line:halt2} is always a local optimum. By Lemma~\ref{prop:a5}, the point $\Vec{x}^*$ output at Step~\ref{line:halt} is also a local optimum. Thus, we only need to argue that Algorithm 2 will output the solution and halt with probability at least ${9}/{10}$. 

    Recall that Algorithm 2 queries ${Q}_1 = 100h \cdot n^{\beta}$ random points in the first round. %The probability that 
	We have \[\frac{n^d}{Q_1} = \frac{1}{100h} n^{1+\frac{d-2}{d}\alpha}.\]
    
    Then, by definition,
    \[ \Pr\left[ \rank(\Vec{x}_1) > \frac{1}{2h} \cdot n^{1+\frac{d-2}{d}\alpha}\right] = (1 - \frac{\frac{1}{2h} \cdot n^{1+\frac{d-2}{d}\alpha}}{n^d})^{{Q}_1} = (1 - \frac{50}{{Q}_1})^{{Q}_1} \leq  {1}/{10}.\]
    
	In other words, with probability at least ${9}/{10}$, we have 
	\begin{align} \label{eq:rank_bound_propa3} 
		\rank(\Vec{x}_1) \leq \frac{1}{2h} \cdot n^{1+\frac{d-2}{d}\alpha}\,.
	\end{align} 
	If inequality (\ref{eq:rank_bound_propa3}) holds, then the procedure \flsdescent{}$({1}/{h} \cdot {n^{1+\frac{d-2}{d}\alpha}}, h, \Vec{x}_1, 2)$ should halt within a number of rounds of at most 
	$$T = \widetilde{k} \cdot \left((h-1) + \frac{2}{3}\right) + \lceil \log n \rceil < \widetilde{k} \cdot h = k .$$
	Otherwise, let $t=\lfloor {2\widetilde{k}}/{3} \rfloor$. By Lemma~\ref{prop:a1} we know that $y_{t}$ is already available from Step~\ref{line:callf} by round $T$. Then by Lemma~\ref{prop:a2}, we have 
	$$\rank(\Vec{y}_{t}) < \left(\frac{1}{2}-\frac{2}{3}\right) \cdot \left(\frac{1}{h} n^{1+\frac{d-2}{d}\alpha}\right) < 0,$$ which is impossible. Thus, the call \flsdescent{}$(1/h \cdot {n^{1+\frac{d-2}{d}\alpha}}, h, \Vec{x}_1, 2)$ must halt within $T$ rounds in this case, which completes the argument.
\end{proof}

Now consider the total number of queries made by Algorithm 2. 

\begin{lemma}\label{prop:a4}
	A call of procedure \flsdescent{}$(s, D, \Vec{x}, r)$ will make $O(\lceil s/\widetilde{k} \rceil^{d-1}\cdot \widetilde{k})$ number of queries, including the queries made by its sub-procedure.
\end{lemma}
\begin{proof}
	We proceed by induction on the depth $D$.
	The base case is when $D=1$. In this case, the procedure \flsdescent{} performs $s$ steps of steepest descent, where $s \leq \widetilde{k}$. Thus it will make $s\cdot2d = O(\widetilde{k})$ queries.
	
	\medskip 
	%\noindent  
	For any depth $D > 1$,
	\begin{description}
		\item[$\bullet\;\;$] the number of queries made by Step~\ref{line:query} is at most $$2d\cdot (2\lceil s/\widetilde{k} \rceil + 1)^{d-1}\cdot\widetilde{k} = O(\lceil s/\widetilde{k} \rceil^{d-1}\cdot\widetilde{k});$$
		\item[$\bullet\;\;$] the number of queries made by all sub-procedures is bounded as follows by the induction hypothesis $$\widetilde{k} \cdot O\left(\lceil {s}/{{\widetilde{k}}^2} \rceil^{d-1}\cdot\widetilde{k}\right) = O(\lceil s/\widetilde{k} \rceil^{d-1}\cdot\widetilde{k});$$
		\item[$\bullet\;\;$] the number of queries made by \dacs{} is at most $$ (2\lceil s/\widetilde{k} \rceil + 1)^{d-1} \cdot 2d \cdot 2 = O(\lceil s/\widetilde{k} \rceil^{d-1}).$$
		% Why this is 3 not 2?
	\end{description}
	Thus the total number of queries is $O(\lceil s/\widetilde{k} \rceil^{d-1}\cdot\widetilde{k})$, which concludes the proof for all $D$.
\end{proof}

\begin{proof}[Proof of Upper Bound of Theorem~\ref{thm:2}]
	The correctness of Algorithm 2 is established by Lemma~\ref{prop:a3}.
	By Lemma~\ref{prop:a4}, the total number of queries issued by Algorithm 2 is 
	$$100h \cdot n^{(d-1) - \frac{d-2}{d}\alpha} + O((n^{1-\frac{2\alpha}{d}})^{d-1} \cdot ({n^{\alpha}}/{h})) = O(n^{(d-1) - \frac{d-2}{d}\alpha})\,.$$ This concludes the proof of the upper bound part in Theorem~\ref{thm:2}.
\end{proof}

\section{Randomized Lower Bound for Local Search in Constant Rounds}\label{sec:LS.ConstLB}
In this section we show the lower bound for local search in constant number of rounds. We start with a few definitions.

\paragraph{Notation and Definitions}

Recall that $\ell_i = n^{\frac{d^k-d^i}{d^k-1}}, 0 \leq i < k$. Let $m \coloneqq \sum_{0 \leq i < k} \ell_i \leq 2n$. We now consider the grid of side length $m$ in this subsection for technical convenience. 

For a point $\Vec{x}=(x_1, \ldots, x_d)$, let $W_i(\Vec{x})$ be the grid points that are in the cube region of size $(\ell_{i})^d$ with corner point $\Vec{x}$: $$W_i(\Vec{x}) = \{\Vec{y} = (y_1, \ldots, y_d) \in [m]^d : x_j \leq y_j < x_j + \ell_i, \forall j \in [d]\}  \,.$$

Next we define a basic structure called \emph{folded-segment}. Intuitively, the folded-segment connecting two points $\Vec{x},\Vec{y}$ is the following path of points: Starting at point $\Vec{x}$, we initially change the first coordinate of $\Vec{x}$ towards the first coordinate of $\Vec{y}$. Then we change the second coordinate, the third coordinate, and so on, until finally reaching the point $\Vec{y}$. The formal definition is as follows. 

\begin{definition}[folded-segment]\label{def:quasi}
	For any two points $\Vec{x},\Vec{y} \in [m]^d$ , the \emph{folded-segment} $FS(\Vec{x},\Vec{y})$ is a set of points connecting $\Vec{x}$ and $\Vec{y}$, defined as follows.
	
	Let $\Vec{x}=(x_1, \ldots, x_d), \Vec{y} = (y_1, \ldots, y_d)$. For any $1 \leq i \leq d$, define point set $$E_i(\Vec{x},\Vec{y}) \coloneqq \left\{\Vec{z}=(y_1, \ldots,y_{i-1}, z_i, x_{i+1}, \ldots, x_d)\in [m]^d: \min(x_i, y_i) \leq z_i \leq \max(x_i, y_i)\right\}.$$ Then define 
	$$FS(\Vec{x},\Vec{y}) = (\bigcup_{i\in [d]} E_i(\Vec{x},\Vec{y})) \backslash \{ \Vec{x}\}\,.$$
\end{definition}

\paragraph{Staircase Structure}
A staircase $s^{(i)}(\Vec{x}_0, \Vec{x}_1, \ldots, \Vec{x}_i)$ of length $i$, where $0 \leq i \leq k$, consists of a path of points defined by an array of $i+1$ ``connecting'' points $\Vec{x}_0, \Vec{x}_1, \ldots, \Vec{x}_i \in [m]^d$. We call $\Vec{x}_0$ as the start point and $\Vec{x}_i$ as the end point.
For each $0 \leq j < i$, the pair of consecutive connecting points $(\Vec{x}_j, \Vec{x}_{j+1})$ is connected by a folded-segment $FS(\Vec{x}_{j},\Vec{x}_{j+1})$. We denote the $j$-th ($0 \leq j \leq i$) connecting point of $s^{(i)}$ by $\Vec{x}_j$, and the $j$-th ($1 \leq j \leq i$) folded-segment of $s^{(i)}$ as $FS(\Vec{x}_{j-1},\Vec{x}_{j})$.

The probability distribution for staircases of length $i$, where $0 \leq i \leq k$, is as follows. The start point $\Vec{x}_0$ is always set to be the corner point $\mathbf{1}$; points $\Vec{x}_1, \Vec{x}_2, \ldots, \Vec{x}_i$ are picked in turn: the point $\Vec{x}_j$ is chosen uniformly at random from the sub-cube $W_{j-1}(\Vec{x}_{j-1})$.\footnote{Recall that we take the side length of the grid as $m$ rather than $n$. Thus the staircase will not be clipped out by the boundary of grid.}
Two staircases are different if they have different set of connecting points, even if the paths defined by the two sets of connecting points are the same. Thus, $L^{(i)} \coloneqq \Pi_{j=0}^{i-1}\, {\ell_j}^d$ is the number of all possible length $i$ staircases, and each staircase will be selected with the same probability ${1}/{L^{(i)}}$.

A staircase $s$ of length $i$ \textit{grows} from a staircase $s'$ of length $(i-1)$ if the first $i-1$ folded-segment of $s$ and $s'$ are same. A prefix of staircase $s$ is any staircase formed by a prefix of the connecting points sequence of $s$.
To simplify the analysis, we assume the algorithm is given the location of $\Vec{x}_i$ after round $i$, except for the round $k$; this only strengthens the lower bound.

\paragraph{Value Function}

After fixing a staircase $s$, we can define the value function $f^{s}$ corresponding to $s$. 
Within the staircase, the value of the point is \emph{minus} the distance to $\Vec{x}_0$ by following the staircase backward, except the end point $\Vec{x}_i$ of the staircase. The value of any point outside is the $L_1$ distance to $\Vec{x}_0$. \footnote{Though the distance to $\Vec{x}_0$ by following the staircase backward is same as the $L_1$ distance, since all staircases constructed above only \textit{grow} non-decreasingly at each coordinate. But the staircases we constructed in the next subsection for the lower bound of polynomial rounds algorithm doesn't have such properties, and thus the the distance by following the staircase backward is not same as the $L_1$ distance in that case.}
The value of the end point $\Vec{x}_i$ is set as minus the distance to the point $\Vec{x}_0$ by following the staircase backward with probability $1/2$, and the $L_1$ distance to $\Vec{x}_0$ with probability $1/2$. Thus the end point of the staircase must be queried, otherwise the algorithm will incorrectly guess the \emph{location} of the unique solution point with probability at least $1/2$.

This value function makes sure that for any two different staircases $s_1$ and $s_2$, the functions $f^{s_1}$ and $f^{s_2}$ have the same value on the common prefix \emph{and} on any point outside of both $s_1$ and $s_2$. Also notice that $f^{s}$ is deterministic on every point except the end point $\Vec{x}_i$ of staircase $s$.

\paragraph{Good Staircases}
Next we introduce the concept of \emph{good staircases}, on which the algorithm could only \emph{learn} one more connecting point in each round.

\begin{definition}[good staircase]\label{def:goods}
	A staircase $s^{(i)}(\Vec{x}_0, \ldots, \Vec{x}_i)$ of length $i$ is \emph{good} with respect to a deterministic algorithm $\mathcal{A}$ that runs in $k$-rounds, if the following condition is met:
	\begin{itemize}
		\item $\; \; \;$ when $\mathcal{A}$ is running on the value function $f^{s^{(i)}}$, for any $0<j<i$, any point \\ $\Vec{x} \in FS(\Vec{x}_j, \Vec{x}_{j+1})$ is \emph{not} queried by algorithm $\mathcal{A}$ in rounds $1, \ldots, j$.
	\end{itemize}
\end{definition}

We use good staircase to indicate the good staircase \emph{with respect to a fixed deterministic algorithm $\mathcal{A}$} for brevity. {Good staircases} have the following properties.

\begin{lemma}\label{lem:good0}
	\begin{enumerate}
		\item If $s$ is a {good staircase}, then any prefix $s'$ of $s$ is also a {good staircase}.
		
		\item Let $s_1, s_2$ be any two good staircases. If the first $i+1$ connecting points of the staircases are the same, then $\mathcal{A}$ will receive the same replies in rounds $1, \ldots, i$ and issue the same queries in rounds $1, \ldots, i+1$ while running on both $f^{s_1}$ and $f^{s_2}$.
	\end{enumerate}
\end{lemma}
\begin{proof}
	We first prove part $2$.
	By the definition of {good staircase}, in round $1, \ldots, i$, algorithm $\mathcal{A}$ will not query any point on $s_1$ or $s_2$ that is after the first $i+1$ connecting points,
	where $f^{s_1}$ and $f^{s_2}$ may have different value. Since $\mathcal{A}$ is deterministic, by induction from round $1$, we get that $\mathcal{A}$ will issue the same queries and receive the same replies in rounds $1, \ldots, i$. The queries in round $i+1$ are also same because they only depend on the replies in round $1, \ldots, i$. The proof of part $1$ is similar by taking $s_1 = s$ and $s_2 = s'$.
\end{proof}

If most of the possible staircases are {good} staircases, then $\mathcal{A}$ will fail on a constant fraction of the inputs generated by length $k$ staircases.

\begin{lemma}\label{lem:good1}
	If the algorithm $\mathcal{A}$ issues at most $\mathcal{Q}_k = 1/10 \cdot n^{\frac{d^{k+1} - d^k}{d^k - 1}}$ queries in each round, and $9/10$ of all possible length $k$ staircases are {good} , then $\mathcal{A}$ will fail to get the correct solution with probability at least $7/40$.
\end{lemma}
\begin{proof}
	If $s^{(k)}(\Vec{x}_0, \ldots, \Vec{x}_{k-1}, \Vec{x}_k)$ is {good}, then $\mathcal{A}$ could not distinguish it from another {good staircase} $s'(\Vec{x}_0, \ldots, \Vec{x}_{k-1}, \Vec{x}_k')$ with same first $k$ connecting points before the last round by Lemma~\ref{lem:good0}.
	
	Recall that before the last round, the algorithm is given the value of $\Vec{x}_{k-1}$. Let $z_{s^{(k-1)}}$ be the fraction of length $k$ {good staircases} among all length $k$ staircases growing from the length $(k-1)$ staircase $s^{(k-1)}$. Then define the random variable $Z$ such that $Z = z_{s^{(k-1)}}$ if the length $(k-1)$ staircase $s^{(k-1)}$ is chosen. Since each length $(k-1)$ staircase is selected with the same probability, the expectation of $Z$ is $9/10$.
	
	Thus by Markov inequality on $(1-Z)$, a half of all the length $(k-1)$ staircases $s^{(k-1)}$ satisfy $z_{s^{(k-1)}} \geq 8/10$. We say a length $(k-1)$ staircase $s^{(k-1)}$ is \textit{nice} if $z_{s^{(k-1)}} \geq 8/10$.
	\smallskip
	
	Recall there are $\left(\ell_{k-1}\right)^d = n^{\frac{d^{k+1} - d^k}{d^k - 1}}$  of length $k$ staircases growing from a staircase of length $(k-1)$.
	
	However, the algorithm can only query $1/10 \cdot n^{\frac{d^{k+1} - d^k}{d^k - 1}}$ points in the last round. Thus if $s^{(k)}$ is a {good staircase} growing from a length $k-1$ nice staircase $s^{(k-1)}$, then $\mathcal{A}$ will not query the endpoint $\Vec{x}_k$ of $s^{(k)}$ with probability at least $7/8$.
	We define the following events: 
	\begin{itemize}
		\item  $\; \; \;$ \emph{Fail} = \{$\mathcal{A}$ makes mistake\}
		\item  $\; \; \;$ \emph{Hit} = \{ $\mathcal{A}$ never queries the end point $\Vec{x}_k$ during the execution\}
		\item  $\; \; \;$ \emph{Nice} = \{the length $k-1$ prefix $s^{(k-1)}$ of $s^{(k)}$ is nice\}
		\item  $\; \; \;$ \emph{Good} = \{the whole length $k$ staircase $s^{(k)}$ is good\}.
	\end{itemize} 
	
	Thus $\mathcal{A}$ will make a mistake with probability at least 
	\begin{small} 
		\begin{align}
			\Pr[\mbox{Fail}] &\geq \frac{1}{2} \cdot \Pr[\mbox{Hit}] 
			\geq \frac{1}{2} \cdot \Pr[\mbox{Hit} \,| \, \mbox{Nice} \land  \mbox{Good}] \cdot \Pr[ \mbox{Nice}  \land  \mbox{Good}] \notag  \\
			& \geq \frac{1}{2} \cdot \frac{7}{8} \cdot \Pr[\mbox{Good} \,| \, \mbox{Nice}] \cdot \Pr[\mbox{Nice}] 
			\geq \frac{1}{2} \cdot \frac{7}{8}  \cdot \frac{8}{10} \cdot \frac{1}{2} = \frac{7}{40}\,, \notag
		\end{align}
	\end{small}
	where the probability is taken on the staircase $s^{(k)}$ and the value of $\Vec{x}_k$.
\end{proof}

\paragraph{Counting the number of good staircases}
Counting the number of {good staircases} is the major technical challenge in the proof. The concept of \textit{probability score function} is useful.

\begin{definition}[probability score function]
	For a fixed deterministic algorithm $\mathcal{A}$, let $Q_\mathcal{A}$ be the set of points queried by $\mathcal{A}$ during its execution. Given a point $\Vec{x} \in [m]^d$ , for any $0 \leq i < k$, define the set of points
	$$W_i(\Vec{x}, Q_\mathcal{A}) \coloneqq \{\Vec{y} \in W_i(\Vec{x}): FS(\Vec{x},\Vec{y}) \cap Q_\mathcal{A} = \emptyset \}.$$
	The \emph{probability score function} for the point $\Vec{x}$ is $SC_i(\Vec{x}, Q_\mathcal{A}) \coloneqq \frac{ \left| W_i(\Vec{x}, Q_\mathcal{A}) \right| }{{\ell_i}^d}$.
	
	The probability score function for a {good staircases} $s^{(i)}(\Vec{x}_0, \ldots, \Vec{x}_i)$ of length $i$, for $0 \leq i < k$, is defined as $SC(s^{(i)}) \coloneqq SC_i(\Vec{x}_i, Q_\mathcal{A})$, where $Q_\mathcal{A}$ is the set of points that have been queried by $\mathcal{A}$ after the round $i$, if $\mathcal{A}$ is executed on value function $f^{s^{(i)}}$.\footnote{By the definition of good staircase, $\mathcal{A}$ will not query the end point of staircase $s^{(i)}$ in rounds $1, \ldots, i-1$. Since $\mathcal{A}$ is a deterministic algorithm and $f^{s^{(i)}}$ is deterministic except at the end point of $s^{(i)}$, the set $Q_\mathcal{A}$ here is uniquely defined.}
\end{definition}

Informally, if $Q_\mathcal{A}$ is the set of points queried by $\mathcal{A}$ in the first $i$ rounds,
the probability score function $SC(s^{(i)})$ can be interpreted as the probability that the length $i+1$ staircase $s^{(i+1)}$ being good if it grows from the good prefix $s^{(i)}$.

Let $\Vec{x} \in [m]^d$, $\Vec{y} \in Q_\mathcal{A}$. Let $B_i(\Vec{x},\Vec{y}) \coloneqq \{\Vec{z} \in W_i(\Vec{x}): \Vec{y} \in FS(\Vec{x},\Vec{z}) \}.$ Thus $\left| B_i(\Vec{x},\Vec{y})\right|$ is the number of \textit{folded-segments} $FS(\Vec{x},\Vec{z}), \Vec{z} \in W_i(\Vec{x})$ that are intersected by a point $\Vec{y} \in Q_\mathcal{A}$. We then define the \emph{cost} incurred by a point $\Vec{y} \in Q_\mathcal{A}$ on the probability score function of a point $\Vec{x}$ for any $0 \leq i < k$ as 
$$cost_i(\Vec{x}, \Vec{y}) \coloneqq \frac{\left| B_i(\Vec{x},\Vec{y}) \right|}{{\ell_i}^d}.$$

The merit of our random staircase structure is that any single queried point will not hit too many of staircases. This property is quantitatively characterized by the following lemma via the \emph{cost}.

\begin{lemma}\label{lem:cost1}
	The \emph{total cost} incurred by one point $\Vec{y} \in Q_\mathcal{A}$ for all $\Vec{x} \in [m]^d$, i.e., $\sum_{\Vec{x}} {\left| B_i(\Vec{x},\Vec{y}) \right|}/({\ell_i}^d)$ is at most $d \cdot {\ell_i}$.
\end{lemma}

\begin{proof}%[Proof of Lemma~\ref{lem:cost1}]
	If $\Vec{y} \notin W_i(\Vec{x})$, we simply have $cost_i(\Vec{x}, \Vec{y})=0$. Therefore we only need to estimate $cost_i(\Vec{x}, \Vec{y})$ or equivalently, $\left| B_i(\Vec{x},\Vec{y}) \right|$ for pairs $\Vec{x}, \Vec{y}$ such that $\Vec{y} \in W_i(\Vec{x})$.
	
	For two points $\Vec{x}=(x_1, \ldots, x_d), \Vec{y}=(y_1, \ldots, y_d)$, let $t(\Vec{x},\Vec{y}) \in [d]$ denote the smallest index such that for any $j$ satisfying $t(\Vec{x},\Vec{y}) < j \leq d$, there is $x_j=y_j$. Then, by the definition of \textit{folded-segments}, we have 
	$$B_i(\Vec{x},\Vec{y}) = \{(y_1, \ldots,y_{t(\Vec{x},\Vec{y})-1}, z_{t(\Vec{x},\Vec{y})}, z_{t(\Vec{x},\Vec{y})+1}, \ldots, z_d)\in W_i(\Vec{x}):  z_{t(\Vec{x},\Vec{y})} \geq y_{t(\Vec{x},\Vec{y})}\}.$$ 
	Therefore, we have $\left| B_i(\Vec{x},\Vec{y}) \right| \leq {\ell_i}^{d-t(\Vec{x},\Vec{y})+1}$. 
	
	Finally, for a point $\Vec{y}$, the number of points $\Vec{x}$ such that $\Vec{y} \in W_i(\Vec{x})$ and $t(x,y)=j$ is at most ${\ell_i}^j$. Therefore, we have 
	\begin{align} 
		\sum_{\Vec{x} \in [m]^d} \frac{\left| B_i(\Vec{x},\Vec{y}) \right|}{{\ell_i}^d} 
		&= \sum_{j\in[d]} \sum_{\Vec{x} \in [m]^d} \frac{\left| B_i(\Vec{x},\Vec{y}) \right| }{{\ell_i}^d} \cdot \mathbf{1}[t(\Vec{x},\Vec{y})=j] \\
		&\leq \sum_{j \in [d]} \frac{{\ell_i}^j \cdot {\ell_i}^{d-j+1}}{{\ell_i}^d} \\
		&= d \cdot {\ell_i}
	\end{align} 
\end{proof}

We can now count the number of length $k$ good staircases with the \emph{two-stage analysis}, which is formally summarized in below.

\begin{lemma}\label{lem:good2}
	If the number of queries issued by algorithm $\mathcal{A}$ is at most $$\mathcal{Q}_k = 1/{ (10d\cdot k}) \cdot n^{\frac{d^{k+1} - d^k}{d^k - 1}},$$ then $9/10$ of all possible length $k$ staircases are {good} with respect to $\mathcal{A}$.
\end{lemma}
\begin{proof}
	For $0 \leq i \leq k$, let $G^{(i)}$ denote the set of all {good staircases} of length $i$. Then
	$| G^{(i)} |$ is the number of all length $i$ {good staircases}, and $R^{(i)} \coloneqq | G^{(i)} | / L^{(i)}$ is the fraction of length $i$ {good staircases}. In particular, $L^{(0)} = | G^{(0)} | = R^{(0)} = 1$.
	
	Let's first fix a length $i \; (0 \leq i<k-1)$ {good staircase} $s^{(i)} \in G^{(i)}$, and denote the set of all {good staircases} growing from $s^{(i)}$ by $G(s^{(i)})$.
	Now consider the sum of the probability score function for every staircase in $G(s^{(i)})$. By Lemma~\ref{lem:good0}, algorithm $\mathcal{A}$ will make the \emph{same} queries in rounds $1, \ldots, i+1$ when running on any instance $f^{s^{(i+1)}}$ generated by staircase $s^{(i+1)}$ in $G(s^{(i)})$. Thus, we can use Lemma~\ref{lem:cost1} and the union bound to upper bound the total cost to the sum of their probability score function:
	\begin{small}
		\begin{align} 
			\sum_{s^{(i+1)} \in G(s^{(i)})} SC(s^{(i+1)})
			& \geq {\ell_i}^d \cdot  SC(s^{(i)}) - \mathcal{Q}_k \cdot  d\ell_{i+1} 
			\geq {\ell_i}^d \cdot  \left(SC(s^{(i)}) - \frac{1}{10k}\right)
			\label{eq:sumpsf} 
		\end{align}
	\end{small}
	
	The second inequality comes from the fact that $\mathcal{Q}_k \cdot \ell_{i+1} = \frac{{\ell_{i}}^d}{d\cdot (10k)}  \,.$ 
	
	Next, let's consider the number of all length $i+2$ {good staircases}.
	\begin{small}
		\begin{align}
			| G^{(i+2)} | &= \sum_{s^{(i+1)} \in G^{(i+1)}} SC(s^{(i+1)}) \cdot {\ell_{i+1}}^d \label{eq:sum_1}\\ 
			& = \sum_{s^{(i)} \in G^{(i)}} \left( \sum_{s^{(i+1)} \in G(s^{(i)})} SC(s^{(i+1)}) \cdot {\ell_{i+1}}^d \right) \\
			& \geq ({\ell_{i+1}}^d \cdot {\ell_i}^d) \cdot \sum_{s^{(i)} \in G^{(i)}} \left(SC(s^{(i)}) - \frac{1}{10k}\right) \label{eq:sum_3}\\ 
			& = {\ell_{i+1}}^d \cdot \left(| G^{(i+1)} | - {\ell_i}^d \cdot \frac{| G^{(i)}|}{10k}\right)
			\label{eq:sum_4}
		\end{align}
	\end{small}
	Equality (\ref{eq:sum_1}) and equality (\ref{eq:sum_4}) are from the definition of probability score function. Inequality (\ref{eq:sum_3}) is from  inequality (\ref{eq:sumpsf}).
	Dividing each side of (\ref{eq:sum_4}) by $L^{(i+2)}$, we get 
	\begin{small}
		\begin{align}
			R^{(i+2)} \geq R^{(i+1)} - \frac{R^{(i)}}{10k} \geq R^{(i+1)} - \frac{1}{10k}\,.
		\end{align}
	\end{small}
	We know that $R^{(1)} = 1$ by definition, so $R^{(k)} \geq 9/10$ as required.
\end{proof}

The correctness of the lower bound statement in Theorem~\ref{thm:1} follows from Lemma~\ref{lem:good1} and Lemma~\ref{lem:good2}.

\section{Randomized Lower Bound for Local Search in Polynomial Rounds}\label{sec:LS.PolyLB}

%%%%%%%%%%%%%%%%%%%%%%%%%%%%%%%%%%%%%%%%%%%%%%%%%%%%%

In this section, we show the lower bound for local search in polynomial number of rounds. 
% \jiawei{Introduce the organization of this proof here?}

\subsection{Notation and Definitions}\label{sec:LS.PolyLB.ND}

Recall the number of rounds is $k = n^{\alpha}$, where $\alpha \in (0,d/2)$. Define $\ell = n^ {1-(2/d)\cdot \alpha}$ and $m = n/\ell$, where the definition of $m$ from section~\ref{sec:LS.ConstLB} is overridden here. 

Define $\mathcal{Q}_k$ to be the number of queries allowed for the $k$-round algorithm stated in Theorem~\ref{thm:2}:

\begin{equation}\label{equ:q_k}
	\mathcal{Q}_k = \left\{ 
	\begin{array}{ll} 
		c_d \cdot n^{(d-1) - \frac{d-2}{d}\alpha} & \textrm{if $d>4$}\\
		\frac{c_4}{\log n} \cdot n^{3 - \frac{1}{2}\alpha} & \textrm{if $d=4$}\\ 
		c_3 \cdot n^{(d-1) - \frac{2}{3}\alpha} & \textrm{if $d=3$} 
	\end{array} \right.
\end{equation}
The value $c_d$ above is a constant depending on dimension $d$.

For any point $x \in [n]$ in the 1D grid, define $W^1(x) $ as the set of points within $\ell$ steps of going right from $x$, and assuming $1$ is the next point of $n$ by wrapping around. Formally, $$W^1(x) = \{y \in [n] : x < y \leq x+ \ell \; \mathrm{or} \; y \leq x+ \ell - n \} \,.$$

Similarly, for a point $\Vec{x}=(x_1, \ldots, x_d)$ in the $d$ dimensional grid $[n]^d$, let $W(\Vec{x})$ be the set of grid points that are in the cube region of side length $\ell$ with corner point $\Vec{x}$, and wrapping around the boundary if exceeding. That is, $$W(\Vec{x}) = \{\Vec{y} = (y_1, \ldots, y_d) \in [n]^d : y_j\in W^1(x_j), \forall j \in [d]\}  \,.$$
For any index $i \in [d]$, define
\begin{align}
	& W^{-1}_i(\Vec{x}) = \{\Vec{y}\in [n]^d: \forall j \leq i, x_j\in W^1(y_j) ;\forall j > i, y_j = x_j \} \\
	& W^b_i(\Vec{x}) = \{\Vec{y}\in [n]^d: \forall j < i, x_j\in W^1(y_j); \max\{n- \ell, x_i\} < y_i \leq n;\forall j > i, y_j = x_j \}
\end{align}

Let $$W^{-1}(\Vec{x}) = \bigcup_{i\in [d]} W^{-1}_i(\Vec{x}), W^b(\Vec{x}) = \bigcup_{i\in [d]} W^b_i(\Vec{x}), W^r(\Vec{x}) =  W^{-1}(\Vec{x}) \cup W^b(\Vec{x}).$$

When calculating the coordinate of points, we always keep wrapping around it into the range $[n]^d$ for convenience.

\paragraph{Staircase structure}
Let $\Vec{x}_0 = \mathbf{1}$ be the corner point. The remaining connecting points $\Vec{x}_1, \Vec{x}_2, \ldots, \Vec{x}_i$ are selected in turn, where $\Vec{x}_j$ is chosen uniformly from 
\begin{enumerate}
	\item the set $W(\Vec{x}_{j-1})$, if $(j \; \mathrm{mod} \; m) \neq 0$;
	\item the set $[n]^d$, otherwise.
\end{enumerate}
We still use the folded-segment as in Definition~\ref{def:quasi} to link every two consecutive connecting points. The length of a staircase is the number of folded-segments in it. Note that the folded-segments in a staircase can now intersect; the value function on those intersecting points has to be handled with caution. 

Let $S^{(i)}$ be the set of all possible staircases of length $i$. Let $S^{(i)}(s_j)$ be the subset of $S^{(i)}$ in which every staircase has staircase $s_j$ as prefix.
Every possible staircase of length $i$ occurs with the same probability and their total number is $L^{(i)} = \left| S^{(i)} \right|$. Each staircase $s_i$ of length $i$ has the same number of length $j$ staircases growing from it, namely, $\left| S^{(i)}(s_j) \right| = L^{(i)} / L^{(j)} $. 

Let $p(\Vec{x}, \Vec{y}, i, t)$ to be the probability that point $\Vec{y}$ is on the $(i+t-1)$-th folded-segment, conditioned on point $\Vec{x}$ being the $i$-th connecting point\footnote{Ignore the restriction that the $0$-th connecting point $\Vec{x}_0$ has to be $\Vec{1}$ here, otherwise it may be impossible for an arbitrary point $\Vec{x}\in [n]^d$ being the $i$-th connecting point.}. Similarly, define $q(\Vec{x}, \Vec{y}, i, t)$ as the probability that point $\Vec{y}$ is the $(i+t)$-th connecting point, conditioned on point $\Vec{x}$ being the $i$-th connecting point.

\paragraph{Value function}

The value function $f^{s}$ for staircase $s$ is determined by the same rule as the constant rounds case in section~\ref{sec:LS.ConstLB}. 

A point is called a \emph{self-intersection point} if it lies on multiple folded-segments \footnote{By definition~\ref{def:quasi} of folded-segment, the $i$-th connecting point is on the $i$-th folded-segment, but not the $(i+1)$-th folded-segment} and 
we deem it to belong to the folded-segment closest to the solution. Thus, the distance from it to the start point is defined by tracing through all the previous points on the staircase.

Recall an important property of these value functions that if two staircases $s_1, s_2$ share a common length $i$ prefix, then $f^{s_1}$ and $f^{s_2}$ have same value except for points that are after the $i$-th connecting points\footnote{In particular, for a self-intersection point, it is after the $i$-th connecting points if any one of the folded-segments it intersects with are after the $i$-th connecting points} of $s_1$ or $s_2$. 

\bigskip 

In the rest of this section, we will show that the distribution of value functions generated by all length $K = 2k$ staircases is hard for any $k$-rounds deterministic algorithm $\mathcal{A}$ with at most $\mathcal{Q}_k$ queries.

\subsection{Useful Assumptions}\label{sec:LS.PolyLB.ass}
We make several assumptions to simply the proof in this section.

\begin{assumption}\label{ass:1}
	Once the algorithm $\mathcal{A}$ queried a point on the $i$-th folded-segment $FS(\Vec{x}_{i-1},\Vec{x}_{i})$, the location of point $\Vec{x}_{i}$ and all previous connecting points are provided to the algorithm $\mathcal{A}$.
\end{assumption}

With this assumption, we can quantify the \emph{progress} of algorithm $\mathcal{A}$ at a certain round by the number of connecting points it \emph{knows}, where we say algorithm $\mathcal{A}$ \emph{knows} or \emph{learns} the connecting point $\Vec{x}_i$ if $\mathcal{A}$ is already given the location of $\Vec{x}_i$.

\begin{assumption}\label{ass:2}
	Algorithm $\mathcal{A}$ learns at least one more connecting point in each round, and it \emph{succeeds} on a specific input if it knows the location of the end point $\Vec{x}_{K}$ of the staircase by the end of round $k$.
\end{assumption}

Notice that assumptions~\ref{ass:1} and \ref{ass:2} only make the bound stronger. 

If $\mathcal{A}$ learns $t$ more connecting points in one round, we say $\mathcal{A}$ \emph{saves} $t-1$ rounds, since we expect $\mathcal{A}$ learns exactly one more connecting points in each round by default.

\begin{assumption}
	Algorithm $\mathcal{A}$ issues the same number of queries in each round, i.e.  $\mathcal{Q} \coloneqq \mathcal{Q}_k/k$.
\end{assumption}

Note that every $k$ rounds algorithm $\mathcal{A}$ with $\mathcal{Q}_k$ queries in total can be converted to an algorithm $\mathcal{A}'$ in $2k$ rounds with $\mathcal{Q}_k / k$ queries in each round: if $\mathcal{A}$ makes $t$ queries in a round, $\mathcal{A}'$ can simulate it by submitting those queries in $\lceil t / (\mathcal{Q}_k / k) \rceil$ rounds. This simulation can only increase the total number of rounds by a factor of $2$, which is fine in the polynomial-rounds setting.

\begin{assumption}\label{ass:4}
	Algorithm $\mathcal{A}$ will keep running until it learns the end point of staircase, regardless of the round limit $k$.
\end{assumption}

A fixed round limit is more technically challenging to deal with. With assumption~\ref{ass:4}, we could instead argue that any such algorithm $\mathcal{A}$ will run more than $1.1k$ rounds in expectation. Then, by applying Markov inequality, we will show that $\mathcal{A}$ fails to learn the end point of the staircase with probability at least $0.1$ if the round limit $k$ is imposed. 

Let $SV(s)$ be the total number of rounds $\mathcal{A}$ saved when $\mathcal{A}$ is running on the input generated by a fixed staircase $s$. 
By definition, the total number of rounds needed for learning the end point of $s$ is $K - SV(s)$.

With the assumptions and arguments above, the next two subsections are devoted to proving the following lemma, which establishes $\mathcal{Q}_k$ as the lower bound.

\begin{lemma}\label{lem:sv2thm2}
	The following inequality holds:
	\begin{equation}
		\frac{L^{(K)} \cdot K - \sum_{s \in S^{(K)}} SV(s)}{L^{(K)}} \geq 1.1 k \; ,
	\end{equation}
	where the left hand side is the expected number of rounds needed to learn the end point considering the input distribution generated by length-$K$ staircases.
\end{lemma}

\subsection{Estimate the Savings}\label{sec:LS.PolyLB.sav}

Let $r^s_i$ be the number of the round in which the $i$-th connecting point on staircase $s$ is first known to $\mathcal{A}$ when running on the input $f^s$.

\begin{definition}[critical point]
	The $i$-th connecting point is a \emph{critical point} of a length $t$ staircase $s$ if $r^s_i \neq r^s_{i+1}$ or $i=t$.
\end{definition}

By assumption~\ref{ass:1}, an equivalent definition is that when the $i$-th connecting point is first learned at round $r^s_i$, $\mathcal{A}$ has not queried any point that is after the $i$-th connecting point; that is, the $i$-th connecting point is learned by querying a point on the $i$-th folded-segment.

Intuitively, each critical point takes one round to learn, while each non-critical points are given for free by assumption~\ref{ass:1}.

By amortizing the rounds saved for each non-critical points to the next closest critical point, we define $SV(s, i)$ as the number of rounds \emph{saved} by the $i$-th connecting point of staircase $s$. More formally,
\begin{enumerate}
	\item $SV(s,i)= i - \max \{j: r^s_j = r^s_i - 1\} - 1$, if $i$-th connecting point is a \emph{critical point};
	\item $SV(s,i)=0$, otherwise.
\end{enumerate}
By definition, $SV(s) = \sum_{i=1}^k SV(s,i)$.

Let $r^s_{i,j}, j \geq i$  be the number of the round in which the $i$-th connecting point on staircase $s$ is first learned by $\mathcal{A}$, when running on the input generated by the length $j$ prefix of $s$.

\begin{lemma}\label{lem:cri}
	The following inequality holds: $r^s_{i,j} \geq r^s_i$.
\end{lemma}
\begin{proof}
	Denote the instance generated by the staircase $s$ by $f^s$ and the instance generated by the length $j$ prefix of $s$ by $f^s_j$.
	
	Assume towards a contradiction that $r^s_{i,j} < r^s_i$. Then we prove by induction that the queries made by $\mathcal{A}$ from round $1$ to round $r^s_{i,j}$ are the same on inputs $f^s$ and $f^s_j$.
	
	\emph{Base case:} $\mathcal{A}$ makes same queries in round $1$, as $\mathcal{A}$ is a deterministic algorithm.
	
	\emph{Induction step:} Consider round $t \leq r^s_{i,j}$. By the induction hypothesis, all the queries in the previous $t-1$ rounds are the same. Notice that $f^s$ is equal to $f^s_j$ except for points on $s$ which are after $j$-th connecting point. Moreover, $\mathcal{A}$ never queried any such point in the first $t-1$ rounds; otherwise we would have $r^s_i \leq t-1$ by assumption~\ref{ass:1}, which is impossible. Thus $\mathcal{A}$ also gets the same feedback from queries in the first $t-1$ rounds. This directly implies that $\mathcal{A}$ will make the same queries in round $t$ since the algorithm is deterministic.
	
	When running on the instance $f^s_j$, $\mathcal{A}$ queried a point $\Vec{x}$ on the length $j$ prefix of $s$ at round $r^s_{i,j}$, which reveals the location of $i$-th connecting point by assumption~\ref{ass:1}. Obviously, $\Vec{x}$ is also on the whole staircase $s$, and $r^s_i$ should be equal to $r^s_{i,j}$, which  contradicts the assumption that $r^s_{i,j} < r^s_i$. Thus the assumption was false and the inequality stated in the lemma holds.
\end{proof}

\begin{lemma}\label{lem:cri2}
	If the $j$-th connecting point is a critical point of $s$, the following hold
	\begin{enumerate}
		\item the queries issued by $\mathcal{A}$ from round $1$ to round $r^s_j+1$ are the same when running on either $f^s$ or $f^s_j$;
		\item $r^s_{i,j} = r^s_i$ for all $i \leq j$.
	\end{enumerate}
\end{lemma}

\begin{proof}
	Recall the equivalent definition of critical point: when running on input $f^s$, $\mathcal{A}$ will not query any point on $s$ that is after the $j$-th connecting point in the first $r^s_j$ rounds.
	
	Notice that  $f^s$ and $f^s_j$ have the same value except for points on $s$ which are after $j$-th connecting point. 
	Using an induction argument similar to the one in Lemma~\ref{lem:cri}, we show that the queries made by $\mathcal{A}$ from round $1$ to round $r^s_j+1$ are the same when running on either $f^s$ or $f^s_j$. 
	
	\emph{Base case:} $\mathcal{A}$ makes the same queries in round $1$.
	
	\emph{Induction step:} Consider round $t \leq r^s_j + 1$. By the induction hypothesis, all the queries in the previous $t-1$ rounds are the same. Notice that $f^s$ is equal to $f^s_j$ except for points on $s$ which are after $j$-th connecting point. Moreover, $\mathcal{A}$ never queried any such point in the first $t-1$ rounds; otherwise we would have $r^s_{i+1} \leq t-1 \leq r^s_j$, which contradicts to the fact that the $j$-th connecting point is critical. Thus $\mathcal{A}$ also gets the same feedback from queries in the first $t-1$ rounds. This directly implies that $\mathcal{A}$ will make the same queries in round $t$.
	
	Since the $j$-th connecting point is a critical point, $\mathcal{A}$ must query a point that is on the $j$-th folded-segment in round $r^s_j$, when running on both input $f^s$ or $f^s_j$. We then have $r^s_{j,j} = r^s_j$, as the fact that $j$-th folded-segment is in the length $j$  prefix.
	
	By assumption~\ref{ass:1}, when running on either $f^s$ or $f^s_j$, for any $i < j$, the $i$-th connecting point is learned before round $r^s_{j,j} = r^s_j$, and all the queries before round round $r^s_j$ are the same. Thus, $r^s_{i,j} = r^s_i$ for any $i \leq j$.
\end{proof}

Let $s^{(i)}$ be the length $i$ prefix of $s$. For all $i$, define $$\overline{SV(s,i)} = SV(s^{(i)},i) = i - \max \{j: r^s_{j,i} = r^s_{i,i} - 1\} - 1\,.$$ 
The value of $\overline{SV(s,i)}$ only depends on the length $i$ prefix of $s$. Similarly, let \[ \overline{SV(s)} = \sum_{i=1}^{K} \overline{SV(s,i)} \,.\] 
\bigskip

The next lemma shows that $\overline{SV(s)}$ is an upper bound of $SV(s)$ for any $s$.

\begin{lemma}\label{lem:svbar}
	The inequality $\overline{SV(s)} \geq SV(s)$ holds for any staircase $s$.
\end{lemma}

\begin{proof}
	It suffices to prove  that $\overline{SV(s,i)} \geq SV(s,i)$ for any $i$. The case of the $i$-th connecting point being non-critical point is easier, since $\overline{SV(s,i)} \geq 0 = SV(s,i)$.
	
	If the $i$-th connecting point is a critical point, we have $r^s_{j,i} = r^s_j$ for any $j \leq i$ by Lemma~\ref{lem:cri2}. In this case, the definitions of $\overline{SV(s,i)}$ and $SV(s,i)$ are the same.
\end{proof}

Assume algorithm $\mathcal{A}$ has now learned the first $i$ connecting points and denote the $i$-th connecting point by $\Vec{x}_i$. Let point $\Vec{y}$ be a queried point in the next round and consider the number of rounds saved by point $\Vec{y}$. 

If point $\Vec{y}$ is on the $(i+t-1)$-th folded-segment, $t$ more connecting points are learned and $t-1$ rounds are saved by assumption~\ref{ass:1}. Suppose that algorithm $\mathcal{A}$ \emph{forgets} every query it previously made and the structure of the length $i$ prefix of the staircase except point $\Vec{x}_i$. Then the probability of $\Vec{y}$ being on the $(i+t-1)$-th folded-segment is $p(\Vec{x}_i, \Vec{y}, i, t)$ by definition.

In this worst-case for the algorithm, the number of rounds saved by $\Vec{y}$ in expectation is bounded by the term:
\begin{equation}\label{equ:rs_i}
	\Gamma(i) = \max_{\Vec{x}, \Vec{y} \in [n]^d} \sum_{t=1}^{K-i} (t-1) \cdot p(\Vec{x}, \Vec{y}, i, t).
\end{equation}
The value of $\Gamma(i)$ only depends on the random walk used for generating staircase, regardless the choice of algorithm $\mathcal{A}$. 
\medskip 

Define
\begin{equation}\label{equ:rs}
	\Gamma = \frac{1}{K} \sum_{i=0}^{K-1} \Gamma(i) \; .
\end{equation}
$\Gamma$ will be a good overall estimate even when the algorithm $\mathcal{A}$ does not forget (i.e. could benefit from the queries and knowledge acquired in previous rounds).

\medskip 

The following key lemma upper bounds the savings of all staircase by the value of $\Gamma$.

\begin{lemma}\label{lem:sav}
	\begin{equation}
		\sum_{s \in {S^{(K)}}} SV(s) \leq L^{(K)} \cdot K \cdot \mathcal{Q} \cdot \Gamma,
	\end{equation}
\end{lemma}
\begin{proof}
	We start by applying lemma~\ref{lem:svbar} and reformulating the summation in a few steps by definition. We obtain the following:
	\allowdisplaybreaks
\begin{small}
	\begin{align}
		\sum_{s \in {S^{(K)}}} SV(s) &\leq \sum_{s \in {S^{(K)}}} \overline{SV(s)}
		\label{equ:sav1} \\
		&= \sum_{s \in {S^{(K)}}} \sum_{i=1}^K \overline{SV(s,i)} 
		\label{equ:sav2} \\
		&= \sum_{s \in {S^{(K)}}} \sum_{i=1}^K i - \max \{j: r^s_{j,i} = r^s_{i,i} - 1\} - 1
		\label{equ:sav3} \\
		&= \sum_{s \in {S^{(K)}}} \sum_{j=0}^K \sum_{i=j+1}^K \mathbf{1}[r^s_{j,i} = r^s_{j+1,i} -1 = r^s_{i,i} -1] \cdot (i-j-1) 
		\label{equ:sav4} \\
		&= \sum_{j=0}^K \sum_{s_1 \in {S^{(j)}}} \sum_{i=j+1}^K \sum_{s_2 \in {S^{(i)}(s_1)}} \sum_{s \in {S^{(K)}(s_2)}}
		\mathbf{1}[r^s_{j,i} = r^s_{j+1,i} -1 = r^s_{i,i} -1] \cdot (i-j-1) 
		\label{equ:sav5} \\
		&= \sum_{j=0}^K \sum_{s_1 \in {S^{(j)}}} \sum_{i=j+1}^K (i-j-1) \sum_{s_2 \in {S^{(i)}(s_1)}}
		\mathbf{1}[r^{s_2}_{j} = r^{s_2}_{j+1} -1 = r^{s_2}_{i} -1] \cdot \left| S^{(K)}(s_2) \right|
		\label{equ:sav6} \\
		&= L^{(K)} \cdot \sum_{j=0}^K \sum_{s_1 \in {S^{(j)}}} \sum_{i=j+1}^K \frac{(i-j-1)}{L^{(i)}} \sum_{s_2 \in {S^{(i)}(s_1)}}
		\mathbf{1}[r^{s_2}_{j} = r^{s_2}_{j+1} -1 = r^{s_2}_{i} -1]
		\label{equ:sav7}
	\end{align}
	\end{small}
 
	Inequality (\ref{equ:sav1}) is implied by Lemma~\ref{lem:svbar}. Identities (\ref{equ:sav2}) and (\ref{equ:sav3}) are the definition of $\overline{SV(s)}$. Identities (\ref{equ:sav4}) and (\ref{equ:sav5}) switch the order of the summation and enumerate a length $K$ staircase $s$ by first enumerating its length $j$ and length $i$ prefix. For equation (\ref{equ:sav6}), the value  $$\mathbf{1}[r^s_{j,i} = r^s_{j+1,i} -1 = r^s_{i,i} -1]$$ only depends on the length $i$ prefix $s_2$. Finally, for equation (\ref{equ:sav7}), recall that for any length $i$ prefix $s_2$, the value of  $\left| S^{(K)}(s_2) \right|$ are same, namely $L^{(K)} / L^{(i)}$, which is guaranteed by the way we generate the staircases.
	An interpretation of equation (\ref{equ:sav7}) is that we first enumerate $j$ and all length $j$ prefix $s_1$, and then enumerate $i$ and all length $i$ prefix $s_2$ growing from $s_1$, if condition $r^{s_2}_{j} = r^{s_2}_{j+1} -1 = r^{s_2}_{i} -1$ is true, then there is a contribution of $(i-j-1)\cdot (L^{(K)} / L^{(i)})$ to the total sum of saved rounds made by $s_2$.
	
	Define a new condition: $H(s_1, s_2)$ is true if the $i$-th folded-segment of $s_2$ is queried in round $r^{s_1}_j+1$ when $\mathcal{A}$ is running on the instance generated by $s_1$. 
	\medskip 
	
	We show that condition $r^{s_2}_{j} = r^{s_2}_{j+1} -1 = r^{s_2}_{i} -1$ implies $H(s_1, s_2)$:
	\begin{itemize}
		\item If $r^{s_2}_i - 1 = r^{s_2}_j$ holds, then at least one point on $i$-th folded-segment of $s_2$ is queried in round $r^{s_2}_j+1$ when running on the instance generated by $s_2$.
		\item If $r^{s_2}_{j} = r^{s_2}_{j+1} -1$ holds, then by definition the $j$-th connecting point is a critical point of $s_2$. Further applying lemma~\ref{lem:cri2}, we know $r^{s_1}_j = r^{s_2}_{j,j} = r^{s_2}_j$ and the queries in the first $r^{s_2}_j + 1$ rounds are same when $\mathcal{A}$ is running on instance generated by either $s_1$ or $s_2$. 
	\end{itemize}
	Thus, $H(s_1, s_2)$ holds by combining the two facts above.
	
	\medskip 
	
	Let $Q_\mathcal{A}(s, t)$ be the set of queries made by $\mathcal{A}$ in round $t$ when running on instance generated by staircase $s$. For any point $\Vec{x} \in Q_\mathcal{A}(s_1, r^{s_1}_j+1)$, define $H(s_1, s_2, \Vec{x})$ to be the condition that point $\Vec{x}$ is on the $i$-th folded-segment of $s_2$.
	
	We continue the previous calculation by replacing the condition $r^{s_2}_{j} = r^{s_2}_{j+1} -1 = r^{s_2}_{i} -1$ with $H(s_1, s_2)$ in equation (\ref{equ:sav7}):
	\allowdisplaybreaks
	\begin{align}
		& \sum_{s \in {S^{(K)}}} SV(s) \notag \\
		&\leq L^{(K)} \cdot \sum_{j=0}^K \sum_{s_1 \in {S^{(j)}}} \sum_{i=j+1}^K \frac{(i-j-1)}{L^{(i)}} \sum_{s_2 \in {S^{(i)}(s_1)}}
		\mathbf{1}[H(s_1, s_2)] \label{equ:sav10}\\
		&\leq L^{(K)} \cdot \sum_{j=0}^K \sum_{s_1 \in {S^{(j)}}} \sum_{\Vec{x} \in Q_\mathcal{A}(s_1, r^{s_1}_j+1)}\sum_{i=j+1}^K \frac{(i-j-1)}{L^{(i)}} \sum_{s_2 \in {S^{(i)}(s_1)}}
		\mathbf{1}[H(s_1, s_2, \Vec{x})] \label{equ:sav11}\\
		&= L^{(K)} \cdot \sum_{j=0}^K \sum_{s_1 \in {S^{(j)}}} \sum_{\Vec{x} \in Q_\mathcal{A}(s_1, r^{s_1}_j+1)}\sum_{i=j+1}^K \frac{(i-j-1)}{L^{(i)}} \cdot \left| S^{(i)}(s_1) \right| \cdot
		\Pr_{s_2 \in {S^{(i)}(s_1)}}[H(s_1, s_2, \Vec{x})] \label{equ:sav12}\\
		&= L^{(K)} \cdot \sum_{j=0}^K \frac{1}{L^{(j)}} \sum_{s_1 \in {S^{(j)}}} \sum_{\Vec{x} \in Q_\mathcal{A}(s_1, r^{s_1}_j+1)}\sum_{i=j+1}^K (i-j-1) \cdot \Pr_{s_2 \in {S^{(i)}(s_1)}}[H(s_1, s_2, \Vec{x})] \label{equ:sav13}\\
		&\leq L^{(K)} \cdot \sum_{j=0}^K \frac{1}{L^{(j)}} \sum_{s_1 \in {S^{(j)}}} \sum_{\Vec{x} \in Q_\mathcal{A}(s_1, r^{s_1}_j+1)} \Gamma(j) \label{equ:sav14}\\
		&= L^{(K)} \cdot K \cdot \mathcal{Q} \cdot \Gamma \label{equ:sav15}
	\end{align}
	
	In the above, inequality (\ref{equ:sav11}) is the union bound and inequality (\ref{equ:sav14}) is implied by the definition of $\Gamma$ in expression (\ref{equ:rs}).
\end{proof}

\subsection{Estimate $\Gamma$}\label{sec:LS.PolyLB.est}

\medskip    We are left now with proving Lemma~\ref{lem:sv2thm2} by showing $\Gamma \leq 9 / 20\mathcal{Q}$.

We first consider several properties of the random walk to simplify expression (\ref{equ:rs_i}) for $\Gamma(i)$.

\begin{observation}
	For any $\Vec{x}, \Vec{y}\in [n]^d, i, t \in [K]$, we have 
	\begin{align*}
		& p(\Vec{x}, \Vec{y}, i, t) = p(\Vec{1}, \Vec{y}-\Vec{x}, i, t), \\
		& q(\Vec{x}, \Vec{y}, i, t) = q(\Vec{1}, \Vec{y}-\Vec{x}, i, t) \;.
	\end{align*}
	
\end{observation}

To simplify notation, let  $p(\Vec{x}, i,t) \coloneqq p(\Vec{1}, \Vec{x}, i, t)$ and $q(\Vec{x}, i,t) \coloneqq q(\Vec{1}, \Vec{x}, i, t)$. Now we have 
$$\Gamma(i) \coloneqq \max_{\Vec{x} \in [n]^d} \sum_{t=1}^{K-i} (t-1) \cdot p(\Vec{x}, i, t) \;.$$

As the value of the function $q$ is easier to estimate, the following lemma upper bounds the value of function $p$ by the function $q$.

\begin{lemma}\label{lem:p_q}
	For any $\Vec{x} \in [n]^d$, the following hold:
	\begin{itemize}
		\item $\; \; \; $ If $i+t \; \mathrm{mod} \; m \neq 0$, then: $$p(\Vec{x}, i, t) \leq 2\ell \cdot d \cdot \max_{\Vec{y} \in W^{-1}(\Vec{x})} q(\Vec{y}, i, t-1)  \;. $$
		\item  $\; \; \; $ Else:
		$$p(\Vec{x}, i, t) \leq n \cdot d \cdot \max_{\Vec{y} \in [n]^d} q(\Vec{y}, i, t-1) $$
	\end{itemize}
\end{lemma}

\begin{proof}	
	We will first prove the case of $i+t \; \mathrm{mod} \; m \neq 0$. Let $\Vec{y}, \Vec{z}$ be the $(i+t-1)$-th $(i+t)$-th connecting points. Recall definition~\ref{def:quasi}, folded-segment $FS(\Vec{y}, \Vec{z})$ consists of $d$ segments traversing in $d$ directions, namely $E_1(\Vec{y}, \Vec{z}), \ldots, E_d(\Vec{y}, \Vec{z})$. 
	
	If $\Vec{y} \in W^{-1}_j(\Vec{x}) \backslash W^{-1}_{j-1}(\Vec{x})$, the segment $E_j(\Vec{y}, \Vec{z})$ will visit $\Vec{x}$ only if the first $j-1$ coordinate of point $\Vec{z}$ is same to that of $\Vec{x}$. 
	
	Consider a fixed choice of $\Vec{x}$ and $\Vec{y}$. As point $\Vec{z}$ is uniformly drawn from $W(\Vec{y})$, the probability that $\Vec{x}$ is on $FS(\Vec{y}, \Vec{z})$ is at most $1 / \ell^{j-1}$. 
	Similarly, if $\Vec{y} \in W^b_j(\Vec{x})$, the segment $E_j(\Vec{y}, \Vec{z})$ will visit point $\Vec{x}$ in reverse direction only if the first $j-1$ coordinate of $\Vec{z}$ is same to that of $\Vec{x}$. The probability that $\Vec{x}$ is on $FS(\Vec{y}, \Vec{z})$ is also at most $1 / \ell^{j-1}$. 
	
	\medskip
	Therefore, we obtain
	\allowdisplaybreaks
	\begin{align}
		p(\Vec{x}, i, t) &= \sum_{\Vec{y} \in W^r(\Vec{x})} q(\Vec{y}, i, t-1) \cdot \Pr_{\Vec{z}\in W(\Vec{y})}[\Vec{x} \in FS(\Vec{y}, \Vec{z})] \label{equ:p_q_1}\\
		&\leq \sum_{j \in [d]}  \cdot \left( \sum_{\Vec{y} \in W^{-1}_j(\Vec{x})} q(\Vec{y}, i, t-1) \cdot \frac{1}{\ell^{j-1}} +  \sum_{\Vec{y} \in W^b_j(\Vec{x})} q(\Vec{y}, i, t-1) \cdot \frac{1}{\ell^{j-1}}  \right)\\
		&\leq \ell \cdot \sum_{j \in [d]} \frac{1}{\ell^{j}} \cdot \sum_{\Vec{y} \in W^{-1}_j(\Vec{x}) \cup W^b_j(\Vec{x})} q(\Vec{y}, i, t-1) \label{equ:p_q_3}\\
		&\leq 2\ell \cdot d \cdot \max_{\Vec{y} \in W^r(\Vec{x})} q(\Vec{y}, i, t-1)\;. \label{equ:p_q_4}
	\end{align}
	Identity (\ref{equ:p_q_1}) follows from $W^{-1}(\Vec{x}) \cap W^b(\Vec{x}) = \emptyset$;
	inequality (\ref{equ:p_q_4}) holds by applying the average principle on $q(\Vec{y}, i, t-1)$, noticing that $\left| W^{-1}_j(\Vec{x}) \right|, \left| W^b_j(\Vec{x}) \right| \leq \ell^j$.
	
	The case of $i+t \; \mathrm{mod} \; m = 0$ is simpler and follows from same argument.
\end{proof}

\begin{observation}\label{obs:rw3}
	$p(\Vec{x}, i,t) = p(\Vec{x}, i + j,t), q(\Vec{x}, i,t) = q(\Vec{x}, i + j,t)$
	for any $\Vec{x} \in [n]^d$, where $i+j+t < m$.
\end{observation}

For any $i < m$, define the following quantities:
\begin{itemize} 
\item $\; \; \;$ $p(\Vec{x}, i) \coloneqq p(\Vec{x}, 0, i)$;
\item $\; \; \;$ $q(\Vec{x}, i) \coloneqq q(\Vec{x}, 0, i)$; 
\item $\; \; \;$ $p_m(\Vec{x}, i) \coloneqq p(\Vec{x}, m-i, i)$.
\end{itemize}

\medskip

By definition, we have the following corollary of Lemma~\ref{lem:p_q}.

\begin{corollary}\label{cor:p1p2_q}
	For any $\Vec{x} \in [n]^d, i < m$,
	$$p(\Vec{x}, i) \leq 2\ell \cdot d \cdot \max_{\Vec{y} \in W^r(x)} q(\Vec{y}, i-1) \;,
	p_m(\Vec{x}, i) \leq n \cdot d \cdot \max_{\Vec{y} \in [n]^d} q(\Vec{y}, i-1) \; .$$
\end{corollary}
% \jiawei{Add a constant 2 to the remaining part!}

The value of $\max_{\Vec{y} \in [n]^d} q(\Vec{y}, i)$ could be estimated by the Gaussian distribution with mean value $\Vec{i} \cdot \ell / 2$ and co-variance matrix $\ell^2/12 \cdot \mathbf{I}$. 
The local central limit theorem (e.g., see \cite{lawler2010random}) guarantees the accuracy of such approximation.

\begin{lemma}[by Theorem 2.1.1 in \cite{lawler2010random}]\label{lem:lclt}
	There is a constant $c_L$ such that for any $i > 0$, 
	$$ \max_{\Vec{y} \in [n]^d} q(\Vec{y}, i) \leq \frac{c_L}{\ell^d \cdot i^{d/2}} \; .$$
\end{lemma}

\begin{observation}\label{obs:rw2}
	For any $\Vec{x}\in[n]^d, i, t \in [K]$,
	$$p(\Vec{x}, i,t) = p(\Vec{x}, i + m,t), q(\Vec{x}, i,t) = q(\Vec{x}, i + m,t).$$
\end{observation}

By observation~\ref{obs:rw2}, there is $\Gamma(i) \geq \Gamma(i+m)$ for any $0\leq i < K$. Therefore we have $$\Gamma \leq \frac{1}{m} \cdot \sum_{i=0}^{m-1} \Gamma(i) \; .$$

\begin{observation}\label{obs:rw4}
	For any $\Vec{x}\in [n]^d, i, t \in [K], i + t \leq K$, if $\lfloor \frac{i}{m} \rfloor < \lfloor \frac{i+t}{m} \rfloor$,
	$q(\Vec{x}, i,t) = {1}/{n^d}$.
\end{observation}

This observation suggests that the tail part, i.e., the summation term with index $i+t > m$ in $\Gamma(i)$, is easier to estimate.
Formally, define the prefix part $$\Gamma'(i) = \max_{\Vec{x} \in [n]^d} \sum_{t=1}^{m - i - 1} (t-1) \cdot p(\Vec{x}, t) + (m - i - 1) \cdot p_m(\Vec{x}, m - i) \;.$$ The following lemma shows that the tail part is small enough.

\begin{lemma}\label{lem:gamma_p_i}
	If $c_d \leq 1/40d$, for any $0 \leq i < m$, we have 
	$$\Gamma(i) \leq \Gamma'(i) + \frac{1}{10\mathcal{Q}} \;.$$
\end{lemma}

\begin{proof}
	Let function $o_{n, \ell}(i) = n$ if $i \; \mathrm{mod} \; m = 0$ and $o_{n, \ell}(i)=\ell$ otherwise. Combining Lemma~\ref{lem:p_q} and Observation~\ref{obs:rw4}, for any $0 \leq i < m$, we have
	\allowdisplaybreaks
	\begin{align}
		\Gamma(i) &= \max_{\Vec{x} \in [n]^d} \sum_{t=1}^{K-i} (t-1) \cdot p(\Vec{x}, i, t) \\
		&= \max_{\Vec{x} \in [n]^d} \sum_{t=1}^{m - i} (t-1) \cdot p(\Vec{x}, i, t) + \sum_{t = m - i+1}^{K-i} (t-1) \cdot p(\Vec{x}, i, t) \\
		&\leq \max_{\Vec{x} \in [n]^d} \sum_{t=1}^{m - i} (t-1) \cdot p(\Vec{x}, i, t) + 
		\sum_{t = m - i+1}^{K-i} (t-1) \cdot o_{n, \ell}(i+t) \cdot \frac{d}{n^d} \\
		&\leq \max_{\Vec{x} \in [n]^d} \sum_{t=1}^{m - i} (t-1) \cdot p(\Vec{x}, i, t) + \frac{K^2 \cdot \ell \cdot d}{n^d} \\
		&\leq \max_{\Vec{x} \in [n]^d} \sum_{t=1}^{m - i - 1} (t-1) \cdot p(\Vec{x}, t) + (m - i - 1) \cdot p_m(\Vec{x}, m - i) + \frac{4d \cdot c_d}{\mathcal{Q}} \\
		&\leq \Gamma'(i) + \frac{1}{10\mathcal{Q}}
	\end{align}
\end{proof}

\begin{lemma}\label{lem:gamma_val}
	The following inequality holds:
	$$\Gamma \leq \frac{9}{20\mathcal{Q}} \; .$$
\end{lemma}

\begin{proof}
	Define $$\Gamma' = \frac{1}{m} \cdot \sum_{i=0}^{m-1} \Gamma'(i) \;.$$
	By Lemma~\ref{lem:gamma_p_i}, it is equivalent to prove that $\Gamma' \leq {7}/{20\mathcal{Q}}$.
	
	\medskip 
	
	Let $\Phi_d(t) = \sum_{i=1}^{t} \frac{1}{t^{d/2-1}} \; $. We now estimate $\Gamma'$ as follows:
	\allowdisplaybreaks
 \begin{small}
	\begin{align}
		 \Gamma' &\leq \frac{1}{m} \cdot \sum_{i=0}^{m-1} \Gamma'(i) \\
		&= \frac{1}{m} \cdot \sum_{i=0}^{m-1} \left( \max_{\Vec{x} \in [n]^d} \sum_{t=1}^{m - i - 1} (t-1) \cdot p(\Vec{x}, t) + (m - i - 1) \cdot p_m(\Vec{x}, m - i) \right)  \\
		&\leq \frac{1}{m} \cdot \sum_{i=0}^{m-1} \left( \sum_{t=1}^{m - i - 1} (t-1) \cdot \max_{\Vec{x} \in [n]^d} p(\Vec{x}, t) + (m - i - 1) \cdot \max_{\Vec{x} \in [n]^d} p_m(\Vec{x}, m - i)\right)   \label{equ:gamma_p_3}\\
		&\leq \frac{1}{m} \cdot \sum_{i=0}^{m-1} \left( \sum_{t=1}^{m - i - 1} (t-1) \cdot 2\ell \cdot d \cdot \max_{\Vec{y} \in [n]^d} q(\Vec{y}, t-1)
		+ (m - i - 1) \cdot n \cdot d \cdot \max_{\Vec{y} \in [n]^d} q(\Vec{y}, m-i-1) \right) \label{equ:gamma_p_4} \\
		&\leq \frac{1}{m} \cdot \sum_{i=0}^{m-1} \left( \sum_{t=1}^{m - i - 1} (t-1) \cdot 2\ell \cdot d \cdot \frac{c_L}{\ell^d \cdot (t-1)^{d/2}}
		+ (m - i - 1) \cdot (m \cdot \ell) \cdot d \cdot \frac{c_L}{\ell^d \cdot (m-i-1)^{d/2}} \right) \label{equ:gamma_p_5} \\
		&\leq \frac{1}{m} \cdot \sum_{i=0}^{m-1} \sum_{t=1}^{m} (t-1) \cdot 2\ell \cdot d \cdot \frac{c_L}{\ell^d \cdot (t-1)^{d/2}}
		+ \sum_{i=0}^{m-1} (m - i - 1) \cdot \ell \cdot d \cdot \frac{c_L}{\ell^d \cdot (m-i-1)^{d/2}}  \\
		&\leq \frac{2\Phi_d(m)\cdot d \cdot c_L}{\ell^{d-1}} + \frac{\Phi_d(m)\cdot d \cdot c_L}{\ell^{d-1}} \label{equ:gamma_p_6} \\
		&= \frac{3\cdot \Phi_d(m)\cdot d \cdot c_L}{\ell^{d-1}}
	\end{align}
\end{small}

	Inequality (\ref{equ:gamma_p_4}) is implied by Corollary~\ref{cor:p1p2_q} and inequality (\ref{equ:gamma_p_5}) follows from Lemma~\ref{lem:lclt}. 
	
	By equality (\ref{equ:q_k}), there is $1/\mathcal{Q} = k/\mathcal{Q}_k = O(\Phi_d(m) / \ell^{d-1})$. Therefore, we conclude the proof by picking the constant $c_d$ small enough.% \simina{How small? Give a bound}
\end{proof}

Now we wrap everything up and prove the lower bound of Theorem~\ref{thm:2}.

\medskip 

\begin{proof}[Lower Bound of Theorem~\ref{thm:2}]
	Lemma~\ref{lem:sv2thm2} directly follows from Lemma~\ref{lem:sav} and Lemma~\ref{lem:gamma_val}. 
	Since the query complexity with round limits must be larger or equal to that of fully adaptive setting, we further improve our bound for $d=3$ by taking max with $\Omega(n^{3/2})$, which is the bound for the fully adaptive algorithm from \cite{zhang2009tight}. 
\end{proof}

\section{Brouwer in Rounds} \label{sec:BR}

In this section we consider the query complexity of finding fixed-points in rounds. For Brouwer we focus on constant rounds, since more than $\log{1/\epsilon}$ rounds do not improve the query complexity for Brouwer (\cite{chen2005algorithms,chen2007paths}). If the number of rounds is a non-constant function smaller than $\log{1/\epsilon}$, this only changes the bound by a sub-polynomial term.

\medskip

\noindent \textbf{Theorem} \ref{thm:brouwer_main} (Brouwer, constant rounds, restated): \emph{Let $k \in \mathbb{N}$ be a constant. For any $\epsilon > 0$, the query complexity of the $\epsilon$-approximate Brouwer fixed-point problem in the $d$-dimensional unit cube $[0,1]^d$ with $k$ rounds is $\Theta\left((1/\epsilon) ^{\frac{d^{k+1} - d^k}{d^k - 1}}\right)$, for both deterministic and randomized algorithms.}

Similarly to local search in constant rounds, when $k \rightarrow \infty$, this bound converges to $\Theta\left((1/\epsilon)^{d-1} \right)$ and the gap is smaller than any polynomial. Thus our result fills the gap between one round algorithm and fully adaptive algorithm (logarithmic rounds) except for a small margin. 

We first give some background on the Brouwer problem, and then explain how the upper and lower bounds are obtained.

\subsection{Background on Brouwer}

The existence of fixed points is guaranteed by the Brouwer fixed-point theorem~\cite{Brouwer11}.
\begin{theorem}[Brouwer fixed-point theorem]
	Let $K \subset \mathbb{R}^n$ be a compact and convex subset of $\mathbb{R}^n$ and $f : K \to K$ a continuous function. Then there exists a point $\Vec{x}^* \in K$ such that $f(\Vec{x}^*) = \Vec{x}^*$.
\end{theorem}

We study the computational problem of finding a fixed-point given the function $f$ and the domain $K$. We focus on the setting where $K = [0,1]^d$ and sometimes write fixed-point to mean Brouwer fixed-point for brevity. 

Since computers cannot solve this problem to arbitrary precision, we consider the approximate version with an error parameter $\epsilon$  and the goal is to find an $\epsilon$-fixed-point. In this case, the original function $f$ can be approximated by a Lipschitz continuous function $\widetilde{f}$, where the Lipschitz constant of $\widetilde{f}$ can be arbitrarily large and depends on the quality of approximation. % required. 

We also consider a discrete analogue that will be useful for understanding the complexity in the continuous setting. The discrete problem was shown to be equivalent to the  continuous (approximate) setting in~\cite{chen2005algorithms}.

\paragraph{Discrete Brouwer fixed-point.}

Given a vector $\Vec{x} \in [n]^d$, let $x_i$ denote the value at the $i$-th coordinate of $\Vec{x}$. 
Consider a function $f: [n]^d \rightarrow \{\mathbf{0}, \pm \Vec{e}^1, \pm \Vec{e}^2, \ldots, \pm \Vec{e}^d\}$, where $\Vec{e}^i$ is the $i$-th unit vector and $f$ satisfies the properties:
\begin{itemize}
	\item \textbf{direction-preserving}: for any $\|\Vec{x}-\Vec{y}\|_{\infty} \leq 1$, we have $\|f(\Vec{x})-f(\Vec{y})\|_{\infty} \leq 1$;
	\item \textbf{bounded}: for any $\Vec{x} \in [n]^d$, we have $\Vec{x}+f(\Vec{x}) \in [n]^d$. 
\end{itemize}
Then there exists a point $\Vec{x}^*$ such that $f(\Vec{x}^*) = \Vec{0}$ (\cite{iimura2003discrete}).\footnote{The approximate version and the discrete version problems are called $AFP$ and $ZP$ respectively in \cite{chen2005algorithms}}   

\medskip 

The following figure illustrates a bounded and direction preserving function in 2D.
\vspace{2mm}

\begin{figure}[h!]
	\centering
	\includegraphics[scale=0.5]{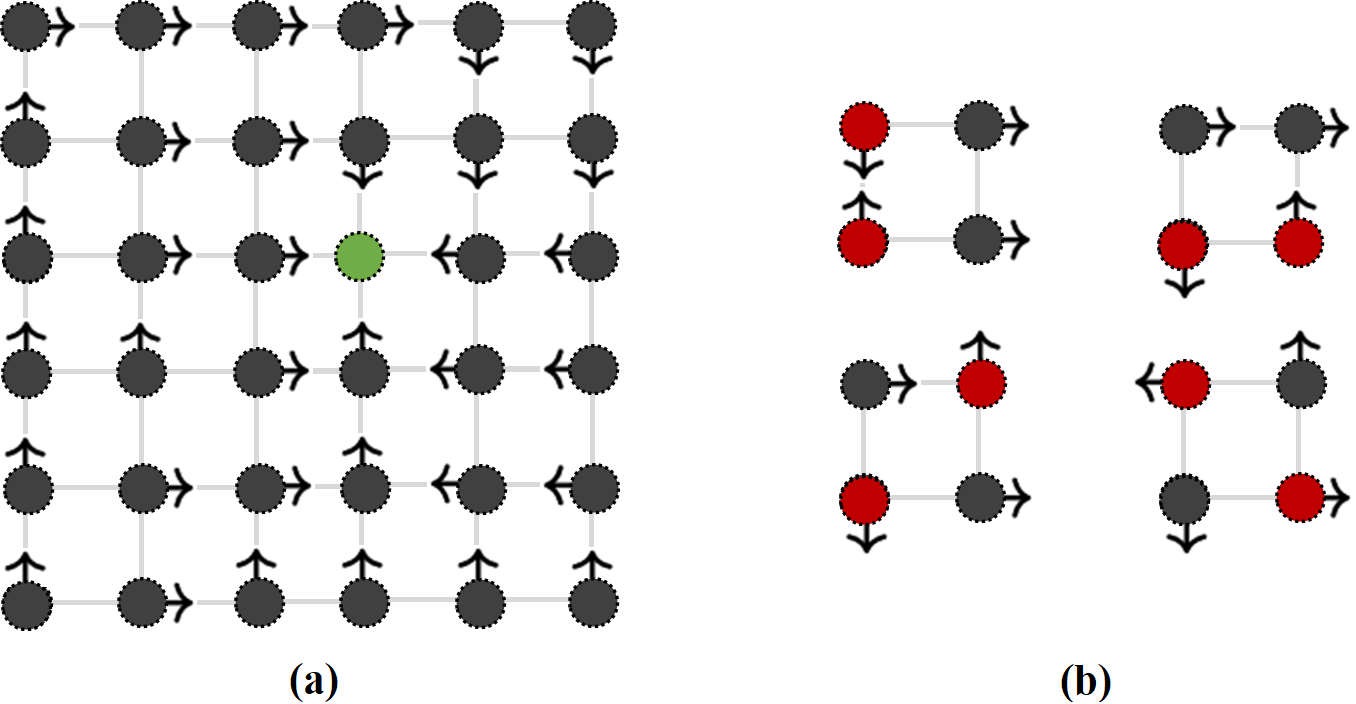}
	\caption{Subfigure (a) shows a bounded and direction preserving function in 2D, on a grid of size $6 \times 6$. The fixed-point is shown in green. Several examples of patterns forbidden by the direction preserving property can be seen in subfigure (b).}
\end{figure}

To compare the difficulty of problems under rounds limit on oracle evaluation, we use the \textit{round-preserving} reduction defined as follows.
\begin{definition}[\textit{round-preserving} reduction]\label{def:RP_reduction}
	A reduction from oracle-based problem \emph{P1} to oracle-based problem \emph{P2} is \textit{round-preserving} if for any instance of problem \emph{P1} with oracle $\mathcal{O}_1$, the instance of problem \emph{P2} with oracle $\mathcal{O}_2$ given by the reduction satisfies that
	\begin{enumerate}
		\item A solution of the \emph{P1} instance can be obtained from any solution of the \emph{P2} instance without any more queries on $\mathcal{O}_1$.
		\item Each query to $\mathcal{O}_2$ can be answered by a constant number of queries to $\mathcal{O}_1$ in one round.
	\end{enumerate}
\end{definition}

The following lemma established the equivalence between the approximate and the discrete version of fixed-point problem.

\begin{lemma}[see section 5~\cite{chen2005algorithms}]
	\bigskip 
	\begin{enumerate}
		\item There is a \textit{round-preserving} reduction such that any instance $(\epsilon, L, d, F)$ of approximate fixed-point is reduced to an instance $(C_1(d) \cdot (L+1) / {\epsilon}, d, f)$ of the discrete fixed-point problem, where $C_1(d)$ is a constant that only depends on the dimension $d$.
		
		\item There is a \textit{round-preserving} reduction with parameter $L > 1$ such that any instance $(n, d, f)$ of discrete fixed-point problem is reduced to an instance $(\epsilon, L, d, F)$ of approximate fixed-point problem, satisfy that $((L-1) / \epsilon) \geq (C_2(d) \cdot n)$ where $C_2(d)$ is a constant that only depends on $d$.
	\end{enumerate}
\end{lemma}

\begin{remark}
	The original reduction from the discrete fixed-point problem to the approximate the version in \cite{chen2005algorithms} will take one more round of queries of the function $f$ of discrete fixed-point problem after getting the solution point $\Vec{x}^*$ of approximate fixed-point problem. However, these extra queries can be avoided if $\epsilon$ is small enough. E.g., by taking $\epsilon = (L-1) / {(n \cdot d \cdot 2^d)}$, the closest grid point (in $L_{\infty}$-norm) to the point $\Vec{x}^*$ will be a zero point of $f$ under the construction in \cite{chen2005algorithms}.
\end{remark}

Since $d$ and $L$ are constants independent of $\epsilon$, we can study both the upper bound and lower bound of the discrete fixed-point problem first, then replace ``$n$'' with ``$1 / \epsilon$'' to get the bound for the approximate fixed-point computing problem.

\subsection{Bounds for Brouwer}

We prove the bounds for the discrete version of Brouwer first, and then obtain the bounds for the continuous problem.

\begin{theorem}(discrete Brouwer fixed-point)\label{thm:5}
	Let $k$ be a constant. The query complexity of computing a discrete Brouwer fixed-point on the $d$-dimensional grid $[n]^d$ in $k$ rounds is $\Theta\left(n^{\frac{d^{k+1} - d^k}{d^k - 1}}\right)$, for both deterministic and randomized algorithms.
\end{theorem}

Theorem~\ref{thm:brouwer_main} will follow by taking $\epsilon = 1/n$.

\subsubsection{Upper Bound for Brouwer}

Our constant rounds algorithm generalizes the divide-and-conquer algorithm in \cite{chen2005algorithms} in the same way as we generalizing the local search algorithm in \cite{llewellyn1989local} in Section~\ref{sec:LS.ConstAlg}. %Therefore it is not surprising that we get the same upper bound for the discrete Brouwer fixed-point problem.
We first present several necessary definitions and lemmas in \cite{chen2005algorithms}.

\begin{definition}[bad cube; see Definition 6 in~\cite{chen2005algorithms}]
	A zero dimensional unit cube $C^{0}=\{\Vec{x}\}$ is bad if $f(\Vec{x}) = \Vec{e}^1$.
	
	For each $i \geq 1$, an $i$-dimensional unit cube $C^i \subset [n]^d$ is bad with respect to function $f: [n]^d \rightarrow \{\mathbf{0}, \pm \Vec{e}^1, \pm \Vec{e}^2, \ldots, \pm \Vec{e}^d\}$ if 
	\begin{enumerate}
		\item $\{ f(\Vec{x}): \Vec{x}\in C^i \} = \{\Vec{e}^1,  \Vec{e}^2, \ldots, \Vec{e}^{i+1}\}$  %(where $f(C)=\{ f(\Vec{x}): \Vec{x}\in C \}$); 
		\item the number of bad $(i-1)$-dimensional unit cubes in $C^i$ is odd.
	\end{enumerate}
\end{definition}

The following theorem on the boundary condition of the existence of the solution is essential for the design of the divide-and-conquer based algorithm.

\begin{theorem}[see Theorem 3 in~\cite{chen2005algorithms}]\label{lem:badc}
	A $(d-1)$-dimensional unit cube $C$ is on the boundary of a $d$-dimensional cube $C'$ if every point in $C$ is on the boundary of $C'$.
	
	If the number of bad $(d-1)$-dimensional unit cubes on the boundary of the $d$-dimensional cube $C$ is odd, then the bounded direction-preserving function $f$ has a zero point within $C$. 
\end{theorem}

The final piece that enables us to use a divide-and-conquer approach is that we can pad the original problem instance on the grid $[n]^d$ to a larger grid: $\{0, 1, \ldots, n+1 \}^d$, and make sure that the new instance contains exactly one bad $(d-1)$-dimensional unit cube on its boundary, and no new solution is introduced. Let $f$ and $f'$ be the function for the original and the new instance, respectively. Then $f'(\Vec{x}) = f(\Vec{x})$ for any $\Vec{x} \in [n]^d$; for any other point $\Vec{x}=(x_1, \ldots, x_d)$, let $i$ be the largest number that $x_i \in \{0, n+1\}$, we have $f'(\Vec{x}) = \Vec{e}_i$ if $\Vec{x}_i = 0$ and $f'(\Vec{x}) = -\Vec{e}_i$ otherwise. The correctness of this reduction is showed in Lemma 5 of \cite{chen2005algorithms}.

\newpage 

\paragraph{Algorithm for Brouwer.}
\begin{enumerate} 
	\item  Initialize the cube $C_0$ as $\{0, 1, \ldots, n+1 \}^d$; for each $0 \leq i < k$, set $\ell_i = n^{\frac{d^k-d^i}{d^k-1}}$
	\item In each round $i \in \{1, \ldots, k-1\}$:
	\begin{itemize}
		\item $\; \; \;$  Divide the current cube $C_{i-1}$ into sub-cubes $C_i^1, \ldots, C_{i}^{n_i}$ of side length $\ell_i$ that cover $C_{i-1}$. These sub-cubes have mutually exclusive interior, but each $(d-1)$-dimensional unit cube that is not on the boundary of cube $C_{i-1}$ is either $(i) $ on the boundary of two sub-cubes $C_i^{j_1}, C_i^{j_2}$ or $(ii)$ not on the boundary of any sub-cubes
		\item $\; \; \;$ Query $f'$ with all the points on the boundary of sub-cubes $C_i^1, \ldots, C_{i}^{n_i}$
		\item $\; \; \;$  Set $C_i = C_i^j$, where $C_i^j$ is the sub-cube that has odd number of bad $(d-1)$-dimensional unit cubes on its boundary. Choose arbitrary one if there are many
	\end{itemize}
	\item In round $k$, query all the points in the current cube $C_{k-1}$ and get the solution point.
\end{enumerate} 

\medskip 
\begin{proof}[Upper Bound of Theorem~\ref{thm:5}]
	Notice that for any $i \in \{1, \ldots, k-1\}$, the sum of the numbers of bad $(d-1)$-dimensional cubes on the boundary of $C_i^1, \ldots, C_{i}^{n_i}$ is still odd, since each bad cubes that is not on the boundary of $C_{i-1}$ is counted twice. By the parity argument, $C_i$ always exists. Finally, Lemma~\ref{lem:badc} guarantees that there exists a solution in cube $C_{k-1}$, which concludes the correctness of this algorithm.
	The calculations in Section~\ref{sec:LS.ConstAlg} give that this algorithm makes $O\bigl(n^{\frac{d^{k+1} - d^k}{d^k - 1}}\bigr)$ queries in total. %Thus we prove the upper bound in the Theorem~\ref{thm:5}.
\end{proof}

\subsubsection{Randomized Lower Bound for Brouwer}\label{sec:BR.LB}

We reduce the local search instances \emph{generated by staircases}
to discrete Brouwer fixed-point instances in this section, and thus prove the lower bound part of Theorem~\ref{thm:5}.

We use the problem $GP$ defined in \cite{chen2007paths} as an intermediate problem in the reduction. 
\begin{definition}[$GP$; see Section 2.2 in~\cite{chen2007paths}]
	A graph directed $G=([n]^d,E)$ is grid PPAD graph if
	\begin{enumerate}
		\item the underlying undirected graph of $G$ is a subgraph of grid graph defined on $[n]^d$;
		\item there is one directed path $\mathbf{1} = \Vec{x}_0 \rightarrow \Vec{x}_1 \rightarrow \ldots \rightarrow \Vec{x}_{T-1} \rightarrow \Vec{x}_T$ with no self-intersection in the graph.
		Any other point outside of the path is an isolated point.
	\end{enumerate}
	
	% \jiawei{The notation $N_G$ doesn't consist with the whole paper.}
	
	The structure of $G$ is accessed by the mapping function $N_G(\Vec{x})$ from $[n]^d$ to $([n]^d \cup \{``no" \}) \times ([n]^d \cup \{``no'' \})$, such that
	\begin{enumerate}
		\item $N_G(\Vec{x}_0) = (``no", \Vec{x}_1)$;
		\item $N_G(\Vec{x}_T) = (\Vec{x}_{T-1}, ``no")$;
		\item $N_G(\Vec{x}_i) = (\Vec{x}_{i-1}, \Vec{x}_{i+1})$, for all $i \in \{1, \ldots, T-1\}$; %, \forall 1 \leq i \leq T-1$;
		\item $N_G(\Vec{x}) = (``no", ``no")$, otherwise;
	\end{enumerate}
	
	Then $GP$ is the following search problem: Given $n, d$ and the function $N_G(\Vec{x})$, find the end of the path $\Vec{x}_T$.
\end{definition}

The following lemma in \cite{chen2007paths} shows that we can get the lower bound of discrete fixed-point problem from the lower bound of problem $GP$ directly.
\begin{lemma}[see Theorem 3.2 in~\cite{chen2007paths}]\label{lem:GP2BR}
	There is a \textit{round-preserving} reduction such that any instance $(n, d, N_G)$ of problem $GP$ is reduced to an instance $(24n+7, d, f)$ of discrete fixed-point problem.
\end{lemma}

The remaining work is establishing the lower bound of problem $GP$ by reducing the local search instances generated by staircases to $GP$. 

\begin{lemma}\label{lem:LS2GP}
	There is a \textit{round-preserving} reduction such that any local search instance $(n, d, f)$ \emph{generated by possible staircase} defined in Section~\ref{sec:LS.ConstLB}
	is reduced to an instance $(n, d, N_G)$ of problem $GP$.
\end{lemma}
\begin{proof}
	Naturally, we can let the staircases\footnote{Recall that in Section~\ref{sec:LS.ConstLB}, the value of the end point $\Vec{x}$ of the staircases may be a positive value with probability $\frac{1}{2}$. In this case, $\Vec{x}$ is not considered to be on the path, and the previous point is the end of the path in problem $GP$.}
	in the local search instance to be the path in the $GP$ instance.
	Given a local search problem instance on grid $[n]^d$ with value function $f$ that is constructed by staircase as in Section~\ref{sec:LS.ConstLB}, we reduce it to a $GP$ instance on grid $[n]^d$ with function $N_G(\Vec{x})$.
	
	Each query of $N_G(\Vec{x})$ is answered by at most $2d+1$ of queries of value function $f$ for $\Vec{x}$ and all its neighbors. From these queries of $f$, 
	\begin{itemize}
		\item $\; \; \;$ if $f(\Vec{x}) > 0$, let $\Vec{x}$ be an isolated point, i.e., $N_G(\Vec{x}) = (``no", ``no")$;
		\item $\; \; \;$ otherwise, $\Vec{x}$ should be on the path. If $f(\Vec{x})=0$, $\Vec{x}$ is the start of the path; if $\Vec{x}$ is a local optima, $\Vec{x}$ is the end of the path.
		The direction within the paths is from the point with higher value to the point with lower value, and there are at most one neighbor with higher value and at most one neighbor with lower value for any value function $f$ constructed in Section~\ref{sec:LS.ConstLB}.
		Thus we can answer $N_G(\Vec{x})$ correctly by the values of $\Vec{x}$ and its neighbor.
	\end{itemize}
	
	Since the start point and the end point of a staircase are unique under the construction in Section~\ref{sec:LS.ConstLB}, $\Vec{x}_T$ is the solution of the original local search problem, and the start of the path $\Vec{x}_0$ satisfies $\Vec{x}_0 = \mathbf{1}$.
	This reduction does not change the number of rounds needed, and only increases the number of queries by a constant factor $2d+1$. Thus the reduction above is a \textit{round-preserving} reduction.
\end{proof}

\begin{proof}[Lower Bound of Theorem~\ref{thm:5}]
    Combining Lemma~\ref{lem:GP2BR}, Lemma~\ref{lem:LS2GP}, and Theorem~\ref{thm:1} concludes the proof of the lower bound in Theorem~\ref{thm:5}.
\end{proof}

\section{Local Search and Brouwer Fixed-Point in 1D}\label{sec:1D}

In this section we study the local search and discrete Brouwer fixed-point on one dimensional grid. 

Recall that the white-box version (where the function is given by a polynomial size circuit) of local search is PLS-Complete for any $d \geq 2$.  On the other hand, when $d=1$, a binary search procedure can solve the problem with $O(\log n)$ queries. Therefore, the query complexity in the 1-dimensional case exhibits different properties than that in dimension $d \geq 2$.

\begin{theorem}(1-dimensional case)\label{thm:1D}
	Let $k$ be a constant. The query complexity of computing a local minimum on the $1$-dimensional grid $[n]$ in $k$ rounds is $\Theta\left(n^{1/k}\right)$, for both deterministic and randomized algorithms. 
\end{theorem}
Note that in the fully adaptive case, the query complexity of this problem is $\Theta\left(\log n\right)$.

We have similar results for Brouwer.
Here we are given a function $f : [0,1] \to [0,1]$ with Lipschitz constant $L > 1$ and a parameter $\epsilon > 0$ and the goal is to find an $\epsilon$-approximate fixed point of $f$. We first show a bound for the discrete Brouwer problem, which will imply a tight bound for the continuous setting.

\begin{theorem}(Discrete Brouwer in 1D)\label{thm:1D_BR}
	Let $k$ be a constant. Given a bounded and direction-preserving function $f : [n] \to [n]$, the query complexity of computing a discrete Brouwer fixed point in $k$ rounds is $\Theta\left(n^{1/k}\right)$, for both deterministic and randomized algorithms. 
\end{theorem}

\begin{corollary}(Brouwer in 1D)\label{thm:1D_BR_epsilon}
	Let $k$ be a constant. Given a $L$-Lipschitz continuous function $f : [0,1] \to [0,1]$ with $L > 1$ and $\epsilon > 0$, the query complexity of computing an $\epsilon$-approximate fixed point of $f$ in $k$ rounds is $\Theta\left((1/\epsilon)^{1/k}\right)$, for both deterministic and randomized algorithms. 
\end{corollary}
Note that in the fully adaptive case, the query complexity of this problem is $\Theta\left(\log 1/\epsilon\right)$.

We conclude the proof of Theorem~\ref{thm:1D} and Theorem~\ref{thm:1D_BR} by a deterministic algorithm in Section~\ref{sec:1D:alg} and a matching lower bound for randomized algorithm in Section~\ref{sec:1D:LB}.  

\subsection{Algorithm in 1D}\label{sec:1D:alg}

The algorithm is a simple adaptation of previous constant rounds algorithms into 1D. Each sub-cube is now a sub-interval, and we take $\ell_i = n^{\frac{k-i}{k}}$ as the length of the sub-interval at the $i$-th round to balance the queries in each round.

\paragraph{Algorithm in constant rounds in 1D}
\begin{enumerate} 
	\item Initialize the current interval as $[n]$.
	\item In each round $i \in \{1, \ldots, k-1\}$:
	\begin{itemize}
		\item $\; \; \;$ Divide the current interval into $n^{1/k}$ of mutually exclusive sub-intervals, and query all the points on the boundary of these sub-intervals.
		\item $\; \; \;$ For local search, set the current interval to be the sub-interval with minimum value in it; for Brouwer, set the current interval to be the one has a fixed-point or both boundary points pointing inwards. Break tie arbitrarily.
	\end{itemize}
	\item In round $k$, query all the points in the current interval and find the solution point.
\end{enumerate} 

\begin{proof}[Upper Bound of Theorem~\ref{thm:1D} and Theorem~\ref{thm:1D_BR}]
	The correctness of the algorithm follows from similar argument as the $d \geq 2$ case: for local search, the steepest descent from the minimal point will not escape from the current sub-interval; for Brouwer, the direction-preserving function in 1D can not change its direction without \emph{stopping} at a fixed-point. %\simina{Mention the idea for the correctness proof rather than just saying it is similar.}
	
	The total number of queries is at most $(k-1) \cdot 2n^{1/k} + n^{1/k}= O(n^{1/k})$.
\end{proof}

\subsection{Lower Bound in 1D}\label{sec:1D:LB}

\begin{proof}[Lower Bound of Theorem~\ref{thm:1D}  and Theorem~\ref{thm:1D_BR}]
	By Yao's minimax theorem, we first construct a distribution of hard instances and then use standard decision tree method to show that $O(n^{1/k})$ queries are necessary for any deterministic algorithm.
	
	\paragraph{Hard Local Search Instances} A uniform distribution over $n$ possible local search instances $f_1,\ldots, f_n$, where $f_i$ is defined as below:
	\begin{enumerate}
		\item $f_i(i) = 0$;
		\item $f_i(j) = j$, if $j>i$;
		\item $f_i(j) = n-j+1$, if $j<i$
	\end{enumerate}
	
	The only solution of $f_i$ is the point $i$. A second observation is that for any set of $q$ queries points, there are at most $2q+1$ possible results considering all $f_i$, because there are at most $2q+1$ of possible different relative location between the $q$ queried points.
	
	Now consider the decision tree of any $k$-rounds deterministic algorithm $A$. If the number of queries on each node is $o(n^{1/k})$, then every node has at most $o(n^{1/k})$ branches. Since $A$ only has $k$ rounds,
    the number of leaf nodes at the bottom of the decision tree is only $o(n)$, i.e., $A$ only has at most $o(n)$ different possible outputs.
    Recall that there are $n$ possible local search instances with different solutions, thus, $A$ must fail with high probability.
	
	\paragraph{Hard Discrete Brouwer Fixed-Point Instances}
	
	A uniform distribution over $n$ possible instances $f_1,\ldots, f_n$, where $f_i$ is defined as below:
	\begin{enumerate}
		\item $f_i(i) = 0$;
		\item $f_i(j) = -1$, if $j>i$;
		\item $f_i(j) = 1$, if $j<i$
	\end{enumerate}
	
	Any function $f_i$ is bounded and direction-preserving, and the only solution is point $i$. The remaining proof follows from the same argument for local search: There are at most $2q+1$ different outcomes from any set of $q$ queries. Thus, for any $k$-rounds deterministic algorithm $A$, if the number of queries is $o(n^{1/k})$, the number of leaf nodes in the decision tree of $A$ is $o(n)$, which is far from enough to cover all $n$ possible instances.
    
    Therefore, we conclude that $\Omega(n^{1/k})$ queries are necessary for any $k$ rounds algorithm for local search or Brouwer in 1D.
\end{proof}

\end{document}